\documentclass[10pt,leqno]{article}

\RequirePackage[colorlinks,citecolor=blue,urlcolor=blue]{hyperref}
\RequirePackage{graphicx}

\RequirePackage{amsthm,amsmath,amsfonts,amssymb}
\RequirePackage[numbers]{natbib}
\usepackage{lmodern}
\usepackage{caption}
\usepackage{multirow,float,mathdots,color}
\makeatletter
\allowdisplaybreaks[4]

\newtheorem{thm}{Theorem}
\newtheorem{lem}{Lemma}
\newtheorem{coro}{Corollary}
\newtheorem{prop}{Proposition}

\usepackage{array}
\newcolumntype{H}{>{\setbox0=\hbox\bgroup}c<{\egroup}@{}}

\providecommand{\keywords}[1]{\textbf{Keywords}: #1}

\theoremstyle{plain}

\theoremstyle{remark}
\newtheorem{definition}{Definition}

\renewcommand{\hat}{\widehat}

\def\wh{\widehat}

\newcommand{\cor}{{\rm Corr}}
\newcommand{\cov}{{\rm Cov}}

\newcommand{\var}{{\rm Var}}

\newcommand{\ve}{{\varepsilon}}
\def\E{{\rm E}} 

\newcommand{\bA}{{\mathbf A}}
\newcommand{\bB}{{\mathbf B}}
\newcommand{\bD}{{\mathbf D}}
\newcommand{\bF}{{\mathbf F}}

\newcommand{\bG}{{\mathbf G}}

\newcommand{\bI}{{\mathbf I}}

\newcommand{\bL}{{\mathbf L}}

\newcommand{\bP}{{\mathbf P}}

\newcommand{\bS}{{\mathbf S}}
\newcommand{\bU}{{\mathbf U}}
\newcommand{\bV}{{\mathbf V}}
\newcommand{\bW}{{\mathbf W}}
\newcommand{\bX}{{\mathbf X}}
\newcommand{\bY}{{\mathbf Y}}
\newcommand{\bZ}{{\mathbf Z}}
\newcommand{\ba}{{\mathbf a}}
\newcommand{\bb}{{\mathbf b}}
\newcommand{\bd}{{\mathbf d}}
\newcommand{\bc}{{\mathbf c}}

\newcommand{\bx}{{\mathbf x}}
\newcommand{\by}{{\mathbf y}}
\newcommand{\bz}{{\mathbf z}}

\newcommand{\bbeta}  {\boldsymbol{\beta}}

\newcommand{\btheta} {\boldsymbol{\theta}}

\newcommand{\bC}{{\mathbf C}}

\newcommand{\calD}{{\mathcal D}}
\newcommand{\calJ}{{\mathcal J}}

\newcommand{\calT}{{\mathcal T}}

\def\6bullets{\bullet\bullet\bullet\bullet\bullet\bullet}

\begin{document}
	\title{A two-way heterogeneity model for dynamic networks}
	\date{}
	\author{
		\normalsize Binyan Jiang\\
		\normalsize Department of Applied Mathematics, The  Hong Kong Polytechnic University\\
		\normalsize Hong Kong, by.jiang@polyu.edu.hk
		\and
		\normalsize Chenlei Leng\\
		\normalsize Department of Statistics, University of Warwick  \\
		\normalsize United Kindom, C.Leng@warwick.ac.uk
		\and
		\normalsize Ting Yan\\
		\normalsize Department of Statistics, Central China Normal University\\
		\normalsize China, tingyanty@mail.ccnu.edu.cn
		\and
		\normalsize Qiwei Yao\\
		\normalsize Department of Statistics, London School of Economics\\
		\normalsize  United Kindom, q.yao@lse.ac.uk
		\and
		\normalsize Xinyang Yu\\
		\normalsize Department of Applied Mathematics, The Hong Kong Polytechnic University\\
		\normalsize Hong Kong, xinyang.yu@connect.polyu.hk
	}
	\maketitle
	
	\begin{abstract}
		Dynamic network data analysis requires joint modelling individual snapshots and time dynamics. This paper proposes a new two-way heterogeneity model towards this goal. The new model equips each node of the network with two 
		heterogeneity parameters, one to characterize the propensity of forming ties with other nodes and the other to differentiate the tendency of retaining existing ties over time. 
		Though the negative log-likelihood function
		is non-convex, it is locally convex in a neighbourhood of the true value of the parameter vector. By using
		a novel method of moments estimator as the initial value,  the {consistent local}  
		maximum likelihood estimator (MLE) can be obtained by a gradient descent algorithm.
		To establish the upper bound for the estimation error of the
		MLE, we derive a new uniform deviation bound,
		which is of independent interest. The usefulness of the model and the associated theory are further supported by extensive simulation and
		{  the analysis of some real network data sets.}
	\end{abstract}
	

\section{Introduction}
Network data featuring prominent interactions between subjects arise in various areas such as biology, economics, engineering, medicine, and social sciences \citep{newman2018networks,kolaczyk2020statistical}.  As a rapidly growing field of active research, statistical modelling of  networks aims to capture and understand the linking patterns in these data. A large part of the literature has focused on examining these patterns for canonical, static networks that are observed at a single snapshot. Due to the increasing availability of networks that are observed multiple times,  models for dynamic networks evolving in time
are of increasing interest now. These models typically assume, among
others, that networks observed at different time are independent \citep{p19, bbm18},  independent conditionally on some latent processes \citep{dd16, mm17}, or  drawn sequentially from an exponential random	graph model conditional on the previous networks \citep{hanneke2007discrete,hanneke2010discrete,krivitsky2014separable}. 


One of the stylized facts of real-life networks is that their nodes often have different tendencies to form ties and may evolve differently over time. The former is manifested	by the fact that the so-called hub nodes have many links while the peripheral nodes have  small numbers of connections in, for example, a big social network. The latter becomes evident when some individuals are more active in seeking new ties/friends than the others.
In this paper, we  refer to these two kinds of heterogeneity as static heterogeneity and dynamic heterogeneity respectively. Also known as degree heterogeneity in the static network literature, static heterogeneity has featured prominently in several popular models widely used in practice including the stochastic block model and its degree-corrected generalization \citep{karrer2011stochastic}. 
See also \cite{jin2015fast,sengupta2018block, jin2022optimal,ke2022score}, and the references therein. Another common and natural approach to capture the static heterogeneity  is to introduce node-specific parameters, one for each node. 	For single networks, this is often conducted via modelling the logit of the link probability between each pair of nodes as the sum of their heterogeneity parameters. Termed as the $\beta$-model \citep{chatterjee2011}, this model and its generalizations have been extensively studied 
when a single static network is observed \citep{yan2013central,karwa2016inference, graham2017econometric, yan2019statistical, chen2021analysis, Leng2020b}. 

With $n$ observed networks each having $p$ nodes, the goal of this study is two-fold: (i) 
We propose a dynamic network model named the two-way heterogeneity model that captures both static heterogeneity and dynamic heterogeneity, and develop the associate inference methodology;
(ii) We establish new generic asymptotic results that can be applied or extended to different models with a large number of parameters (in relation to $p$). {We focus on the scenario that the number of nodes $p$ goes to infinity. Our asymptotic results hold when  $np\rightarrow \infty$, though $n$ may be fixed. }
The main contributions of our paper can be summarized as follows.
\begin{itemize}
\item  {\color{black}	We introduce a reparameterization of the general autoregressive network model \citep{jiang2020autoregressive} to accommodate variations in both node degree and dynamic fluctuations. This novel approach can be regarded as an extension of the $\beta$-model \citep{chatterjee2011} to a dynamic framework. It encompasses two sets of parameters for heterogeneity: one governs static variations, akin to those in the standard $\beta$-model, while the other addresses dynamic fluctuations. Unlike the general model in \citep{jiang2020autoregressive}, which necessitates a large number of network observations (i.e. $n \rightarrow \infty$), we demonstrate the validity of our formulation even in scenarios where $n$ is small but $p$ is large.		}
\item  The formulation of our model gives rise to a high-dimensional non-convex loss function based on likelihood. 
By establishing the local convexity of the loss function in a neighborhood of the true parameters, 
we compute the local MLE by a standard gradient descent algorithm using a newly proposed method of moments estimator (MME) as its initial value.
To our best knowledge, this is the first result in network data analysis for solving such a  
non-convex optimization problem with algorithmic guarantees.		

\item Furthermore, to characterize the local MLE,  
we have derived its estimation error bounds  in  the $\ell_2$ norm and the $\ell_\infty$ norm when $np\rightarrow \infty$ in which $n\ge 2$ can be finite. 
Due to the dynamic structure of the data,
the Hessian matrix of the loss function exhibits a complex structure. As a result, existing analytical approaches, such as the  interior point theorem \citep{gragg1974optimal,yan2016asymptotics} developed for static networks, 
are no longer applicable; see Section \ref{S3.1} for further elaboration. 
We derive a novel locally uniform deviation bound in a neighborhood of the true parameters with a diverging radius. Based on this we first establish $\ell_2$ norm consistency of the MLE, which paves the way for the uniform consistency in $\ell_\infty$ norm.  

\item  
In establishing the locally uniform deviation bound, we have provided a general result for functions of the form  $L(\btheta)=\frac{1}{p}\sum_{1\leq i\neq j\leq p  }l_{i,j} \left( \theta_i, \theta_j\right) Y_{i,j}$ as defined in \eqref{Lfunction} below. 
This result explores the sparsity structure of $L(\btheta)$ in the sense that most of
its higher 
order derivatives are zero -- the condition which our model satisfies, and provides 
a new bound that substantially extends the scope of empirical processes for the M-estimators \citep{van1996weak} for the models with a fixed number of parameters to those with a growing number of parameters.
The result here is of independent interest as it can be applied to any model with {an objective function} taking the form of $L$. 

\end{itemize}

~\\
The rest of the paper is organized as follows. We  introduce in Section \ref{sec: TWHM} the new two-way heterogeneity model and  present its properties. The estimation of its local MLE in a neighborhood of the truth and the associated  theoretical properties are presented in Section \ref{sec: parameter estimation}. The development of these properties relies on new local deviation bounds which are presented in Section \ref{S.4}.
Simulation studies and an analysis of ants interaction data  are reported  in Section \ref{sec: numerical study}. We conclude the paper in Section \ref{sec: conclusion}. All technical proofs are relegated to Appendix A. 
{\color{black} Additional numerical results showcasing the effectiveness of our method in aiding community detection within stochastic block structures, along with an application aimed at understanding dynamic protein-protein interaction networks, are provided in Appendix B.}

\section{Two-way Heterogeneity Model}\label{sec: TWHM}
Consider a dynamic network defined on $p$
nodes which are unchanged over time. Denote by  a $p\times p$  matrix $\bX^t=(X_{i,j}^t)$ its adjacency matrix at time $t$, i.e. $X_{i,j}^t=1$ indicates the existence of a connection between 
nodes $i$ and $j$ at time $t$, and 0 otherwise. We focus on undirected networks without self-loops, i.e., $X_{i,j}^t = X_{j,i}^t$ for all $(i,j) \in {  \calJ \equiv \{ (i,j): 1\le i < j \le p\}}$,
and $X_{i,i}^t=0$ for $1\le i\le p$, though our approach can be readily extended to directed networks.

{\color{black}
To capture the autoregressive pattern in dynamic networks, \cite{jiang2020autoregressive} proposed to model the 
network process via the following stationary AR(1) framework:
\begin{equation*} \label{b1_jiang}
	X^t_{i,j} \;=\; X^{t-1}_{i,j}\, I( \ve_{i,j}^t=0) \;+\; I (\ve_{i,j}^t=1), \quad t\ge 1,
\end{equation*}
where $I(\cdot)$ denotes  the indicator function,  and the  
$\ve_{i,j}^t$, $(i,j) \in \calJ$ are independent innovations satisfying
\begin{equation*} \label{b2}
	P(\ve_{i,j}^t =1) = \alpha_{i,j}, \quad
	P(\ve_{i,j}^t = -1) = \beta_{i,j}, \quad
	P(\ve_{i,j}^t =0) = 1 - \alpha_{i,j} - \beta_{i,j},
\end{equation*}
for some positive parameters  $ \alpha_{i,j}$ and $\beta_{i,j}$.
This general model opts to neglect the inherent nature of the networks and chooses to estimate each pair $\left(\alpha_{i,j},\beta_{i,j} \right) $ independently. 
As a result, there are  $p(p-1)$  parameters and consistent model estimation requires  $n\rightarrow \infty$. 
Conversely, in numerous real-world scenarios, it is frequently noted that the number of network observations $n$ is modest, while the number of nodes $p$ can significantly exceed $n$. Under such a scenario of small-$n$-large-$p$, the conventional model outlined in  \cite{jiang2020autoregressive}  may not be suitable. To address this and to effectively capture node heterogeneity in dynamic networks, as well as accommodate small-$n$-large-$p$ networks, we propose the following reparameterization for the general AR(1) model mentioned above. This reparameterization not only accounts for inherent node heterogeneity but also reduces the parameter count from $p(p-1)$ to $2p$.	}
\begin{definition}\label{Def1}
{\bf Two-way Heterogeneity Model (TWHM)}. 	
The data generating process satisfies
\begin{equation} \label{b1}
	X^t_{i,j} = I( \ve_{i,j}^t=0)+X^{t-1}_{i,j} I( \ve_{i,j}^t=1)  , \qquad (i,j) \in \calJ,
\end{equation}
where   the   
$\ve_{i,j}^t$, for $(i,j) \in {\cal J}$ and $t\ge 1$  are independent innovations 
with their  distributions satisfying
\begin{equation} \label{bbeta}
	P(\ve_{i,j}^t =r) = \frac{e^{\beta_{i,r}+\beta_{j,r}}}{1+\sum_{k=0}^1e^{\beta_{i,k}+\beta_{j,k}}} ~~{\rm for}~ r=0,1, 	\quad
	P(\ve_{i,j}^t = -1) =  \frac{1}{1+\sum_{k=0}^1e^{\beta_{i,k}+\beta_{j,k }}}.
\end{equation}
\end{definition}

TWHM defined above is a reparametrization of the AR(1) network model \cite{jiang2020autoregressive} as it reduces the total number of parameters from $2p^2$ therein to $2p$. By Proposition 1 of \cite{jiang2020autoregressive}, the matrix process $\{ \bX^t, t \ge 1\} $ is strictly stationary with 
\begin{equation}\label{margP}   
P(X_{i,j}^t =1) = \frac{ e^{\beta_{i,0} + \beta_{j,0}}}{
	1 + e^{\beta_{i,0} + \beta_{j,0}}} = 1 -P(X_{i,j}^t =0),
\end{equation}
provided that we activate the process with $\bX^0 = (X_{i,j}^0)$ also following this stationary marginal distribution.

Furthermore,
\begin{equation*} \label{b8} \nonumber
\E(X_{i,j}^t)= \frac{e^{\beta_{i,0} +\beta_{j,0} }}{1+e^{\beta_{i,0} +\beta_{j,0}}}, \qquad
\var(X_{i,j}^t) =\frac{e^{\beta_{i,0} +\beta_{j,0} }}{(1+e^{\beta_{i,0} +\beta_{j,0}})^2},
\end{equation*}
\begin{equation} \label{b9}
\rho_{i,j}(|t-s|) \equiv
\cor(X_{i,j}^t, X_{i, j}^{s}) =\left(
\frac{e^{\beta_{i,1} +\beta_{j,1}} }{ 1+\sum_{r=0}^1e^{\beta_{i,r} +\beta_{j,r }}  }\right)^{ |t-s|} .
\end{equation}
Note that
the connection probabilities in (\ref{margP}) depend on 
$\bbeta_0 =(\beta_{1,0}, \cdots, \beta_{p,0})^\top$ only, and are of 
the same form as the (static) $\beta$-model  \cite{chatterjee2011}. Hence we call $\bbeta_0$ the static heterogeneity parameter. Proposition \ref{prop1} below confirms
that means and variances of node degrees in TWHM also depend on $\bbeta_0$ only, and that different values of $\beta_{i,0}$ reflect the heterogeneity in the degrees of nodes. 

Under TWHM, it holds that
\begin{equation}
\label{transP}
P(X^t_{i,j}=1|X^{t-1}_{i,j} =0) = \frac{e^{\beta_{i,0}+\beta_{j,0}}}{1+\sum_{k=0}^1e^{\beta_{i,k}+\beta_{j,k}}}, 
P(X^t_{i,j}=0|X^{t-1}_{i,j} =1) =
\frac{1}{1+\sum_{k=0}^1e^{\beta_{i,k}+\beta_{j,k }}}.
\end{equation}
Hence the dynamic changes (over time) of network $\bX^t$ depend on, in addition to $\bbeta_0$, $\bbeta_1 \equiv (\beta_{1,1}, \cdots, \beta_{p,1})^\top$: the larger $\beta_{i,1}$ is,  the more likely $X_{i,j}^t$ will retain the value of  $X_{i,j}^{t-1}$ for all $j$. {\color{black} Thus we call $\bbeta_1$ the dynamic heterogeneity parameter,   as its components reflect the different dynamic behaviours of the $p$ nodes.}
A schematic description of the model can be seen from Figure \ref{fig:0} where three snapshots of a dynamic network with four nodes are depicted.
\begin{figure}[tb] 
\centering
\includegraphics[width=13.5cm]{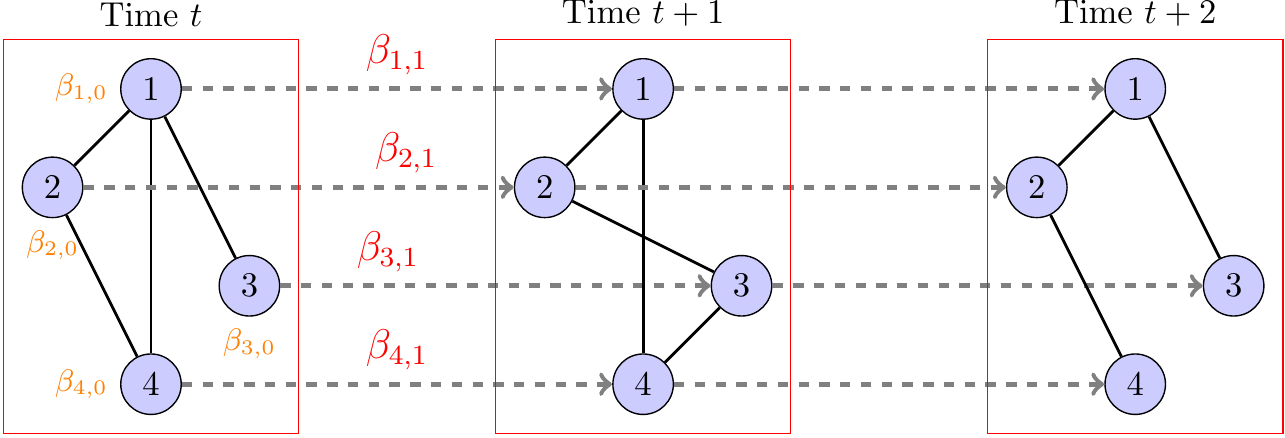} 
\caption{A schematic depiction of TWHM: $\beta_{i,0}, i=1,...,4$, are parameters to characterize the static heterogeneity of nodes, while $\beta_{i,1}$ characterize their dynamic heterogeneity.}\label{fig:0}
\end{figure}

From now on, let $\{\bX^{t}\} \sim P_{\btheta}$  denote  
the stationary TWHM  with parameters ${\btheta} = (\bbeta_0^\top, \bbeta_1^\top)^\top$, 
and  $d_i^t=\sum_{j=1}^p X_{i,j}^t$ be the degree of node $i$ at time $t$. The proposition below lists some properties of the node degrees. 
\begin{prop} \label{prop1}
Let $\{\bX^t\} \sim P_{\btheta}$. Then
$\{( d_1^t,\ldots,  d_p^t), t=0,1,2,\cdots\}$  is a strictly stationary process. Furthermore for any $1\leq i<j\leq p $ and $t, s \ge 0$,
\begin{equation*} \label{b10}
	\E (d_i^t)= \sum_{k=1,\:k \neq i}^{p}\frac{e^{\beta_{i,0} +\beta_{k,0} }}{1+e^{\beta_{i,0} +\beta_{k,0}}}, \qquad
	\var(d_i^t) =\sum_{k=1,\:k \neq i}^{p}\frac{e^{\beta_{i,0} +\beta_{k,0} }}{(1+e^{\beta_{i,0} +\beta_{k,0}})^2},
\end{equation*}
\begin{eqnarray*}
	\rho^d_{i,j}(|t-s|) &\equiv&
	\cor(d_{i }^t, d_{j}^{s}) \\
	&=&\begin{cases}
		C_{i,\rho}\sum_{k=1,\:k \neq i}^{p}\left( \frac{e^{\beta_{i,1} +\beta_{k,1}} }{ 1+\sum_{r=0}^1e^{\beta_{i,r} +\beta_{k,r }}  }\right)^{ |t-s|}\frac{e^{\beta_{i,0} +\beta_{k,0} }}{(1+e^{\beta_{i,0} +\beta_{k,0}})^2}\quad & {\rm if} \; i=j,\\
		0 & {\rm if} \; i\neq j,
	\end{cases}
\end{eqnarray*}
where $C_{i,\rho}=\left( \sum_{k=1,\:k \neq i}^{p}\frac{e^{\beta_{i,0} +\beta_{k,0} }}{(1+e^{\beta_{i,0} +\beta_{k,0}})^2}\right)^{-1} $.
\end{prop}
{\color{black}  Proposition \ref{prop1} implies that
when there exist constants  $\beta_{0}$ and $\beta_{1}$ such that  $\beta_{i,0}\approx \beta_{0}$ and  $\beta_{i,1}\approx \beta_{1}$ for all  $i$, the degree sequence $\{d_i^t, t=1,\ldots, n\}$ is   approximately AR(1).  }

\section{Parameter Estimation}\label{sec: parameter estimation}
We introduce some notation first. Denote by ${\bI}_{p}$ the $p\times p$ identity matrix.   For any $s\in \mathbb{R}$,  $\textbf{s}_{p}$ denotes the $p\times 1$ vector with all its elements  equal to $s$. 
For $ \ba=(a_1,\ldots,a_p)^\top\in \mathbb{R}^{p}$ and $\bA=(A_{i,j})\in \mathbb{R}^{p \times p}$, let 
$\|\ba\|_{q}=\left( a_{i}^q\right)^{1/q} $ for any $q\ge 1$,  
$\|\ba\|_{\infty}=\max_{i}|a_{i}|$, and  $\|\bA\|_{\infty} = \max_{i} \sum_{j=1}^{p} |A_{i,j}|$.
Furthermore, let $\|\bA\|_2$ denote   the spectral norm of $\bA$ which equals its largest eigenvalue. 
For a random matrix $\bW \in \mathbb{R}^{p \times p}$ with $\E\left( \bW\right) =\textbf{0}$, define its matrix variance  as $\var(\bW)= \max \left\{\|\E\left( \bW\bW^{\top}\right) \|_{2},\|\E\left( \bW^{\top}\bW\right) \|_{2}  \right\}$.   The notation $x \lesssim y $ means that there exists a constant $c_1>0$ such that $|x|\leq c_1 |y|$,  while notation $x \gtrsim y $ means there exists a constant $c_2>0$ such that $|x|\geq c_2 |y|$.
Denote by $\bB_{\infty}\left( \bx,r\right) = \left\{ \by: \|\by-\bx\|_{\infty}\leq r \right\}$ the ball centred at $\bx$ with $\ell_\infty$ radius $r$. Let $c,c_0, c_1, \ldots, C, C_0, C_1,\ldots$ denote some generic constants that may be different in different places.
Let ${\btheta}^* =(\bbeta_0^{* \top }, \bbeta_1^{* \top})^\top = (\beta^*_{1,0},\cdots,\beta^*_{p,0},\beta^*_{1,1},\cdots,\beta^*_{p,1})^\top$ be the true unknown parameters.   We assume:
\begin{itemize}
\item [(A1)] There exists a constant $K$ such that for any $i=1,2,\cdots, p$, the true parameters satisfy 
{\color{black}$\beta^*_{i,1}-\max\big( \beta^*_{i,0},0\big) <K$.} 
\end{itemize}
{\color{black} Condition (A1)  ensures that the autocorrelation functions (ACFs)  in (\ref{b9}) are bounded away from 1 for any $(i,j) \in \calJ$.  It is worth noting that both $\beta^*_{i,1}$ and $\beta^*_{i,0}$ are allowed to vary with $p$, thus accommodating sparse networks in our analysis.  In practical terms, $\beta_{i,0}^*$, which reflects the sparsity of the stationary network, tends to be very small for large networks. Consequently, condition (A1) holds when   $\beta_{i,1}^*$ is bounded from above.  }

\subsection{Maximum likelihood estimation}\label{S3.1}
With the available observations $\bX^0,  \cdots,\bX^n$, 
the log-likelihood function conditionally on $\bX^0$ is
of the form $L(\btheta; \bX^n, \cdots, \bX^1|\bX^0 )= \prod_{t=1}^{n} L(\btheta; \bX^{t}|\bX^{t-1} ) $.
Note $\{ X^t_{i,j}\}$ for different $(i,j) \in \calJ$ are independent with each other. By (\ref{transP}),
a (normalized) negative log-likelihood admits the following form:
\begin{align}  \label{eq:lossfunction}
& l(\btheta) = - {\frac{1}{np}} L(\btheta; \bX^n,\bX^{n-1}, \cdots, \bX^1|\bX^0 )  \\ \nonumber
&=  - {\frac{1}{p}}\sum_{1\leq i<j\leq p  }\log \Big({1+ e^{\beta_{i,0} +\beta_{j,0 } }+e^{\beta_{i,1} +\beta_{j,1 } }}\Big)  + {\frac{1}{np}}{ \sum_{1\leq i<j\leq p  }   \Bigg\{   \left( \beta_{i,0} +\beta_{j,0 }\right)    \sum_{t=1}^nX_{i,j}^t}     \nonumber\\ \nonumber
&+\log\left(  1+e^{\beta_{i,1} +\beta_{j,1} }\right)\sum_{t=1}^n\left( 1-X_{i,j}^t\right) \left( 1-X_{i,j}^{t-1}\right)\\
&+  \log\big( 1+e^{\beta_{i,1} +\beta_{j,1}-\beta_{i,0} -\beta_{j,0 } }\big)\sum_{t=1}^n X_{i,j}^t X_{i,j}^{t-1} \Bigg\}.
\nonumber
\end{align}
For brevity, write
\begin{equation}\label{abc}
a_{i,j}  =  \sum_{t=1}^{n} X_{i,j}^{t}, \quad
b_{i,j}=  \sum_{t=1}^{n} X_{i,j}^{t}X_{i,j}^{t-1}  , \quad d_{i,j} =  \sum_{t=1}^{n} \Big( 1-X_{i,j}^{t}\Big) \Big( 1-X_{i,j}^{t-1}\Big).
\end{equation}
Then the Hessian matrix of $l(\btheta)$ is of the form 
\[ \bV(\btheta)= \frac{\partial^2{l(\btheta)}}{\partial{\btheta}\partial{\btheta^\top}}=
\begin{bmatrix} 
\frac{\partial^2{l(\btheta)}}{\partial{\bbeta_0}\partial{\bbeta_0^\top}} & \frac{\partial^2{l(\btheta)}}{\partial{\bbeta_0}\partial{\bbeta_1^\top}} \\
\frac{\partial^2{l(\btheta)}}{\partial{\bbeta_1}\partial{\bbeta_0^\top}} &\frac{\partial^2{l(\btheta)}}{\partial{\bbeta_1}\partial{\bbeta_1^\top}} 
\end{bmatrix} := \begin{bmatrix} 
\bV_{1}(\btheta) & \bV_{2}(\btheta)  \\
\bV_{2}(\btheta) & \bV_{3} (\btheta)
\end{bmatrix},
\]
where for $i\not= j$, 
\begin{eqnarray*}
\frac{\partial^2{l(\btheta)}}{\partial{\beta_{i,0}}\partial{\beta_{j,0}}} &=& \frac{1}{p}\frac{e^{ \beta_{i,0} +\beta_{j,0}}(1+ e^{\beta_{i,1} +\beta_{j,1}})} {(1+e^{\beta_{i,0} +\beta_{j,0}}+e^{\beta_{i,1} +\beta_{j,1}})^2} -\frac{1}{np}b_{i,j}  \frac{e^{\beta_{i,0} +\beta_{j,0}+\beta_{i,1} +\beta_{j,1}}} {(e^{\beta_{i,0} +\beta_{j,0}}+e^{\beta_{i,1} +\beta_{j,1}})^2}, 
\\ \frac{\partial^2{l(\btheta)}}{\partial{\beta_{i,0}}\partial{\beta_{j,1}}} &=& -\frac{1}{p}\frac{e^{\beta_{i,0} +\beta_{j,0} +\beta_{i,1} +\beta_{j,1}}} {(1+e^{\beta_{i,0} +\beta_{j,0}}+e^{\beta_{i,1} +\beta_{j,1}})^2} +\frac{1}{np}b_{i,j} \frac{e^{\beta_{i,0} +\beta_{j,0}+\beta_{i,1} +\beta_{j,1}}} {(e^{\beta_{i,0} +\beta_{j,0}}+e^{\beta_{i,1} +\beta_{j,1}})^2}, \\
	\frac{\partial^2{l(\btheta)}}{\partial{\beta_{i,1}}\partial{\beta_{j,1}}} 
	&=&\frac{1}{p}\frac{e^{ \beta_{i,1} +\beta_{j,1}}(1+ e^{\beta_{i,0} +\beta_{j,0}})} {(1+e^{\beta_{i,0} +\beta_{j,0}}+e^{\beta_{i,1} +\beta_{j,1}})^2}  - \frac{1}{np} d_{i,j} \frac{e^{\beta_{i,1} +\beta_{j,1}}}{(1+e^{\beta_{i,1} +\beta_{j,1}})^2}\\&&-\frac{1}{np}b_{i,j} \frac{e^{\beta_{i,0} +\beta_{j,0}+\beta_{i,1} +\beta_{j,1}}} {(e^{\beta_{i,0} +\beta_{j,0}}+e^{\beta_{i,1} +\beta_{j,1}})^2}.
\end{eqnarray*}
Note that matrix $\bV_{2}(\btheta)$ is symmetric.
Furthermore, the three matrices $\bV_{1}(\btheta), \bV_{2}(\btheta)$ and $\bV_{3}(\btheta)$ are diagonally balanced \citep{hillar2012inverses} in the sense that their diagonal elements are the sums of their respective rows, namely, 
\[(\bV_k(\btheta))_{i,i}=\sum_{j=1,\:j \neq i}^{p} (\bV_k(\btheta))_{i,j}, \quad k=1,2,3.\] 

Unfortunately the Hessian matrix $\bV(\btheta)$ is not uniformly positive-definite. Hence  
$l(\btheta)$ is not convex; 
{  see Section \ref{Sim.H} for an example.  }
Therefore, finding the global MLE by minimizing $l(\btheta)$ would be infeasible, especially given the large dimensionality of $\btheta$.
To overcome the obstacle, we propose the following roadmap  to search for the local MLE over a neighbourhood of the true parameter values $\btheta^*$.
\begin{itemize}
	\item[(1)] First we show that $l(\btheta)$ is locally convex in a neighbourhood of 
	$\btheta^*$  (see Theorem \ref{thm1} below). Towards this end, we first prove that $\E(\bV(\btheta))$ is positive definite in a neighborhood of $\btheta^*$. Leveraging on some newly proved concentration results, we show that $\bV(\btheta)$ converges to $\E(\bV(\btheta))$ uniformly over the neighborhood.
	\item[(2)]  Denote by $\wh \btheta$  the local MLE in the neighbourhood identified above. We derive
	the bounds for $\wh\btheta-\btheta^*$ respectively in both $\ell_2$ and $\ell_\infty$ norms (see Theorems \ref{thm2} and \ref{thm3} below). The $\ell_2$ convergence is established by providing a uniform upper bound for the local deviation {  between} $l(\btheta)-\E(l(\btheta))$ {  and} $l(\btheta^*)-\E(l(\btheta^*))$ (see Corollary \ref{Cor1} in Section \ref{S.4}). The $\ell_\infty$ convergence of $\hat\btheta$ is established by further exploiting the special structure of the objective function.

	\item[(3)]  
	We propose a new method of moments estimator (MME) which is proved to lie asymptotically in the neighbourhood specified in (1) above. With this MME as the initial value, the local MLE $\wh \btheta$ can be simply obtained via a gradient decent algorithm.
\end{itemize} 	
The main technical challenges in the roadmap above can be summarized as follows. 

Firstly,	to establish the upper bounds as stated in (2) above,
we need to evaluate the uniform local deviations of the  loss function. While the theoretical framework for deriving similar deviations of  M-estimators has been well established in, for example, \cite{van1996weak, van2000asymptotic}, classical techniques in empirical process for establishing uniform laws \cite{wainwright2019high} are not applicable because the number of parameters in TWHM diverges. 



Secondly, 
for the classical $\beta$-model, proving the existence and convergence of its MLE relies strongly on the interior point theorem \citep{gragg1974optimal}. In particular, this theorem is applicable only because the 
Hessian matrix of the $\beta$-model
admits a nice structure, i.e. it is diagonally dominant and all its elements are {positive}  depending on the parameters only \citep{chatterjee2011,yan2013central,yan2016asymptotics,han2020bivariate}. However the Hessian matrix of $l(\btheta)$ for TWHM depends on random variables $X_{i,j}^t$'s in addition to the parameters, making it impossible to verify if the score function is uniformly Fr\'{e}chet differentiable or not,  a key  assumption required by the  interior point theorem. 

Lastly, the higher order derivatives of $l(\btheta)$ { may} diverge as the order 
increases.  To see this,  notice that for any integer $k$, the $k$-th order derivatives of $l(\btheta)$ is closely related to the $(k-1)$-th order derivatives of the Sigmoid function $S(x)=\frac{1}{1+e^{-x}}$ in that
$
\frac{\partial^k S\left( x\right) }{\partial x^{k} }=\frac{\sum_{m=0}^{k-2} -A\left(k-1,m \right) \left(-e^{x} \right)^{m+1}  }{\left(1+e^{x} \right)^{k} } 
$, 
where $A\left(k-1,m \right)$ is the Eulerian number  \citep{minai1993derivatives}. Some of the coefficients $A\left(k-1,m \right)$ can diverge very quickly as $k$ increases. Thus, loosely speaking, $l(\btheta)$ is not smooth.
{{This non-smoothness}} and the need to deal with a growing number of parameters  make various local approximations based on Taylor expansions highly non-trivial; noting that the consistency of MLEs in many finite-dimensional models is often established via these approximations.

In our proofs, we have made great use of the special sparse structure of the loss function in  the form \eqref{Lfunction} below.  This sparsity structure stems from the fact that most of its higher order derivatives are zero.  
Based on the uniform  local deviation bound obtained in Section \ref{S.4}, we have established an upper bound for the error of the local MLE under the $l_2$ norm. Utilizing the structure of the marginalized functions of the loss we have further established an upper bound for the estimation error under the $l_\infty$ norm thanks to an iterative procedure stated in Section \ref{S3.3}.

\subsection{Existence of  the local MLE}\label{S3.2}
To establish the convexity of $l(\btheta)$ in a  neighborhood of $\btheta^*$,
we first show that such a local convexity holds
for  $\E(\bV(\btheta))$.

\begin{prop}\label{prop2}
	Let $\bA$ be a $2p\times2p$ matrix
	defined as
	$
	\bA = \left[
	\begin{array}{ccc}
		\bA_{1} & \bA_{2}  \\
		\bA_{2} & \bA_{3}
	\end{array}
	\right],
	$
	where $\bA_1$, $\bA_2$, $\bA_3$ are $p\times p$ symmetric matrices.
	Then $\bA$ is positive (negative) definite if  $-\bA_2,\bA_2+\bA_3,\bA_2+\bA_1$ are all positive (negative) definite.
\end{prop}
\begin{proof}
	Consider any nonzero $\bx = (\bx_1^\top  ,\bx_2^\top  )^\top  \in \mathbb{R}^{2p} $ where $\bx_1,\bx_2\in \mathbb{R}^{p}$, we have:
	\begin{eqnarray*}
		\bx^{T}\bA \bx&=& \bx_1^\top  \bA_{1} \bx_1 + \bx_2^\top  \bA_{3} \bx_2 + 2\bx_1^\top  \bA_{2} \bx_2  \\
		&=& \bx_1^\top   (\bA_{1} +\bA_{2}) \bx_1 +\bx_2^\top   (\bA_{3} +\bA_{2}) \bx_2 - (\bx_1-\bx_2)^\top   \bA_{2} (\bx_1-\bx_2).
	\end{eqnarray*}
	This proves the proposition. 
\end{proof}
Noting that $-\bV_2(\btheta), \bV_2(\btheta)+\bV_3(\btheta) $ and $ \bV_2(\btheta)+\bV_1(\btheta) $ are all diagonally balanced matrices,  with some  routine  calculations it can be shown that  $-\E\bV_2(\btheta^*),\E(\bV_2(\btheta^*)+\bV_3(\btheta^*))$ and $\E(\bV_2(\btheta^*)+\bV_1(\btheta^*))$ have only positive elements, and thus are all positive definite. Therefore, $\E\bV(\btheta^*)$ is positive definite by Proposition \ref{prop2}. By continuity,  when $\btheta$ is close enough to $\btheta^*$, 
$\E\bV(\btheta)$ is also positive definite,  and hence 
$\E l(\btheta)$ is   strongly convex in a neighborhood  of  $\btheta^*$. Next we want to show the local convexity of $l(\btheta)$ whose  second order derivatives  depend on the sufficient statistics   
$ b_{i,j}  =  \sum_{t=1}^{n} X_{i,j}^{t}X_{i,j}^{t-1}  $, and $d_{i,j} =  \sum_{t=1}^{n} \Big( 1-X_{i,j}^{t}\Big) \Big( 1-X_{i,j}^{t-1}\Big)  $. By noticing that the network process is $\alpha$-mixing with an exponential decaying mixing coefficient,
we first obtain the following concentration results for  $ b_{i,j} $ and $d_{i,j}$, which ensure element-wise convergence of $\bV(\btheta)$ to $\E\bV(\btheta)$ for a given $\btheta$ when $np\rightarrow \infty$. 

\begin{lem}\label{mixing}
	Suppose $\{\bX^t\} \sim P_{\btheta}$ for some ${\btheta} = (\beta_{1,0},\cdots,\beta_{p,0},\beta_{1,1},\cdots,\beta_{p,1})^\top$ satisfying condition (A1). Then    for any $(i,j)\in {\cal J}$,
	$\{X^t_{i,j}, t\ge 1\}$ is $\alpha$-mixing with exponential decaying rates. 
	Moreover, for any
	positive constant $c>0$, by choosing $c_1>0$ to be large enough, it holds with probability greater than $1-(np)^{-c}$ that
	\[ \max_{1\le i<j\le p}\left\{ n^{-1}\left|\sum_{t=1}^{n}\left\{ X_{i,j}^{t}- \E\left(X_{i,j}^{t} \right) \right\}\right|, n^{-1}\left|b_{i,j}-\E(b_{i,j}) \right|, n^{-1}\left|d_{i,j}-\E(d_{i,j})\right|\right\} \le c_1   r_{n,p},\]
	where
	$r_{n,p}= \sqrt{n^{-1}\log(np)}+ n^{-1}\log\left(n\right)\log\log\left(n\right)\log\left(np\right)$.
\end{lem}
The following lemma provides a lower bound  for the smallest eigenvalue of $\E(\bV(\btheta))$.

\begin{lem}\label{bound_l2}
	{\color{black}
		Let $\{\bX^t\} \sim P_{\btheta^*}$, $\bB_{\infty}\left( \btheta^*,r \right): = \left\{ \btheta: \|\btheta-\btheta^*\|_{\infty}\le r \right\}$ and $\bB\left( \kappa_{0},\kappa_{1} \right) := $ $\Big\{ \left(\bbeta_{0}, \bbeta_{1}\right) : \|\bbeta_{0}\|_{\infty}\le \kappa_{0},\|\bbeta_{1}\|_{\infty}\le \kappa_{1} \Big\}$.
		Under condition (A1), for any $\kappa_{0},\kappa_{1}$ and $r=c_r e^{-4\kappa_{0}-4\kappa_{1}}$ with $c_r>0$ being a small enough constant,  
	}
	there exists a constant $C>0$ such that
	\begin{equation*} 
		\inf_{\btheta \in \bB_{\infty}\left( \btheta^*,r \right)\cap\bB\left( \kappa_{0},\kappa_{1} \right);\|\ba\|_{2}=1}\ba^{\top}\E\left( \bV(\btheta)\right) \ba\geq  C e^{-4\kappa_{0}-4\kappa_{1}}.
	\end{equation*}
\end{lem}

{\color{black}Examining the proof indicates that the lower bound in Lemma \ref{bound_l2} is attained when  $\bbeta_0= (\kappa_{0},\ldots, \kappa_{0})^\top$ and $\bbeta_{1}=(-\kappa_{1},\ldots, -\kappa_{1})^\top$. 
}
Hence the smallest eigenvalue of $\E\left( \bV(\btheta)\right)$ can decay exponentially in $\kappa_0$ and $\kappa_1$. Consequently, an upper bound for the radius $\kappa_0$ and $\kappa_1$ must be imposed so as to ensure the positive definiteness of the sample analog $\bV(\btheta)$. Moreover, Lemma \ref{bound_l2} also indicates that the positive definiteness of   $\E\left( \bV(\btheta)\right)$ can  be guaranteed  when $\btheta$ is within the $\ell_\infty$ ball $\bB_{\infty}\left( \btheta^*,r \right)$.		To establish the existence  of the local MLE in the neighborhood, we need to evaluate the closeness of $\E\left( \bV(\btheta)\right)$ and   	$\bV(\btheta)$ in terms of the operator norm. Intuitively, for some appropriately chosen {\color{black} $\kappa_{0},\kappa_{1}$, if  $\|\E\left( \bV(\btheta)\right) -\bV(\btheta)\|_2$  has a  smaller order than $e^{-4\kappa_{0}-4\kappa_{1}}$ uniformly over the parameter space   $\{ \btheta: \|\bbeta_{0}\|_{\infty}\le \kappa_{0},\|\bbeta_{1}\|_{\infty}\le \kappa_{1}$ and $\btheta \in \bB_{\infty}\left( \btheta^*,r \right)\}$}, the positive definiteness of $\bV(\btheta)$  can be concluded.

Note that $\bV_2(\btheta)-\E\bV_2(\btheta), \bV_2(\btheta)+\bV_3(\btheta)-\E\left( \bV_2(\btheta)+\bV_3(\btheta)\right)$ and $\bV_2(\btheta)+\bV_1(\btheta) -\E\left(\bV_2(\btheta)+\bV_1(\btheta) \right)$ are all centered and diagonally balanced matrices 
{  which} can be decomposed into  sums of independent random matrices. 
The following lemma provides a bound for evaluating the moderate deviations of these centered matrices. 

\begin{lem}\label{BersM}
	Let $\bZ=(Z_{i,j})_{1\le i,j\le p}$ be a symmetric $p\times p$ random matrix such that the off-diagonal elements $Z_{i,j}, 1\le  i< j\le p$ are independent of each other and satisfy
	\begin{equation*}
		Z_{i,i}= \sum_{j=1,\:j \neq i}^{n} Z_{i,j} ,\quad  \E\left( Z_{i,j} \right)=0,\quad  \var\left( Z_{i,j} \right) \leq \sigma^2, \quad {\rm and} \quad  Z_{i,j} \leq b \quad {\rm almost~~ surely}.
	\end{equation*}
	Then it holds that
	\begin{equation*}
		P\left( \left\| \bZ \right\|_{2} > \epsilon\right) \leq 2p \ \exp\left( -\frac{\epsilon^2}{ 2\sigma^2(p-1) + 4b\epsilon } \right).
	\end{equation*}
\end{lem}

Proposition \ref{prop2}, Lemma \ref{bound_l2} and Lemma \ref{BersM} imply the theorem below.

\begin{thm}\label{thm1}
	Let condition (A1) hold, assume $\{\bX^t\} \sim P_{\btheta^*}$, 
	{\color{black} and $\kappa_{r} :=\|\bbeta_{r}^*\|_\infty$ where $r=0,1$ with $\kappa_{0}+\kappa_{1} \le c\log (np)$ for some small enough constant $c >0$.   Then as $np\rightarrow \infty$ with $n\ge 2$, we have that, with probability tending to one, there exists a unique MLE   in the $\ell_\infty$ ball $\bB_{\infty}\left( \btheta^*,r \right) = \left\{ \btheta: \|\btheta-\btheta^*\|_{\infty}\le r \right\}$ for some 	$r = c_r e^{-4\kappa_{0}-4\kappa_{1}}$, 	where $c_r>0$ is a constant.
	}	
\end{thm}
In the proof of Theorem \ref{thm1}, we have shown that with probability tending to 1, $l(\btheta)$ is convex  in the convex and closed set $\bB_{\infty}\left(\btheta^*,r \right)$. Consequently, we conclude that there exists a unique local MLE  in $\bB_{\infty}\left(\btheta^*,r \right)$.
From Theorem \ref{thm1} we can also see that when  $\kappa_{0}+\kappa_{1}$ becomes larger, the radius $r$ will be smaller, and when $\kappa_{0}+\kappa_{1}$ is bounded away from infinity, $r$ has a constant order. {\color{black} From the proof   we can also see that the constant $c_r$ can be larger if the smallest eigenvalue of the expected Hessian matrix $\E(\bV(\btheta))$ is larger.
}
{\color{black}		Further, by allowing the upper bound of $\|\bbeta^*_0\|_\infty$ to  grow to infinity, our theoretical analysis covers the case where networks are sparse.  Specifically, 
	under the condition that $\|\bbeta_0^*\|_\infty \le \kappa_0$, from \eqref{b8} we   can obtain the following  lower bound (which is achievable when $\beta_{1,0}^*=\ldots=\beta_{p,0}^*=-\kappa_0$) for the density of the stationary network: 
	\begin{equation*}
		\rho:=	\frac{2}{p(p-1)}\sum_{1\leq i< j\leq p}\bP\left(X_{i,j}^t=1 \right) \geq \frac{e^{-2\kappa_0}}{ 1+e^{-2\kappa_0}}=O\left(e^{-2\kappa_0} \right).
	\end{equation*}
	In particular, when $\kappa_0\le c\log (np)$ for some constant $c >0$, we have $\rho \geq  \frac{1}{ [1+(np)^{2c}]}$. 	}
Thus, compared to full dense network processes where the total number of edges for each network is of the order $p^2$, TWHM allows the networks with much fewer edges.

\subsection{Consistency of the local MLE}\label{S3.3}
In the previous subsection, we have proved that with probability tending to one, $l(\btheta)$ is  convex in $\bB_{\infty}\left( \btheta^*, r\right)$, where
{\color{black} $r=c_{r}e^{-4\kappa_{0}-4\kappa_{1}}$ 
}
is defined in Theorem \ref{thm1}.
Denote by $\hat\btheta$ the (local) MLE in $\bB_{\infty}\left( \btheta^*, r\right)$. We now evaluate the $\ell_2$ and $\ell_\infty$
distances between $\hat\btheta$ and the true value $\btheta^*$.

Based on Theorem \ref{thm5} we obtain a local deviation bound for $l(\btheta)$ as in Corollary \ref{Cor1} in Section \ref{S.4}, from which we  establish the following upper bound for the estimation error of $\hat{\btheta}$ under the $\ell_2$ norm:
\begin{thm}\label{thm2}
	Let condition (A1) hold, assume $\{\bX^t\} \sim P_{\btheta^*}$, 
	{\color{black} and $\kappa_{r} :=\|\bbeta_{r}^*\|_\infty$ where $r=0,1$ with $\kappa_{0}+\kappa_{1} \le c\log (np)$ for some small enough constant $c >0$.   Then as $np\rightarrow \infty$ with $n\ge 2$,   it holds with probability converging  to one that
		\begin{equation*}
			\frac{1}{\sqrt{p}}\left\|\hat{\btheta}-\btheta^*\right\|_{2}\lesssim e^{4\kappa_{0}+4\kappa_{1}}\sqrt{\frac{\log(np)}{np}}\left(1+ \frac{\log(np)}{\sqrt{ p}} \right).
		\end{equation*}
	}
\end{thm}
We discuss the implication of this theorem. When $n\to \infty$ and $p$ is finite,  that is, when we have a fixed number of nodes but a growing number of network snapshots, Theorem \ref{thm2} indicates that 
{$\left\|\hat{\btheta}-\btheta^*\right\|_{2}=O_{p}\left(\sqrt{\frac{\log^3 n}{n}} e^{4\kappa_{0}+4\kappa_{1}} \right)=o_p(1)$ when $c$ is small enough. 
}  
On the other hand, when $n$, $\kappa_{0}$ and $\kappa_{1}$ are finite,  Theorem \ref{thm2} indicates that as the number of parameters   $p$ increases, the $\ell_2$ error bound of $\hat{\btheta}$ increases at a much slower rate $O\left(\sqrt{\log p}\right)$. 

Although Theorem \ref{thm2} indicates that  $\frac{1}{\sqrt{p}}\left\|\hat{\btheta} -\btheta^*\right\|_{2}=o_p(1)$ as $np\rightarrow \infty$, it does not guarantee the uniform convergence of all the elements in $\hat{\btheta}$. To prove the uniform convergence in the $\ell_\infty$ norm, we exploit a special structure of the loss function and the $\ell_{2}$ norm bound obtained in Theorem \ref{thm2}.	Specifically, denote $l(\btheta)$ in \eqref{eq:lossfunction} as $l(\btheta)=l(\btheta_{(i)},\btheta_{(-i)})$ where $\btheta_{(i)}:=(\beta_{i,0} ,\beta_{i,1})^{\top}$, and $\btheta_{(-i)}$ contains the remaining elements of $\btheta$ except $\btheta_{(i)}$. Using this notation, we can analogously define $\btheta^{*}_{(i)}$ and $\btheta^{*}_{(-i)}$ for the true parameter 	$\btheta^{*}$, and $\hat{\btheta}_{(i)}$ and $\hat{\btheta}_{(-i)}$ for the local MLE $\hat{\btheta}$. 
We then have that $\btheta_{(i)}^*$ is the mimizer of $\E l\left( \cdot, \btheta_{(-i)}^* \right)$ while  $\hat{\btheta}_{(i)}$ is the minimizer of  $l\left( \cdot, \hat{\btheta}_{(-i)} \right)$. 
The  error of $\hat{\btheta}_{(i)}$ in estimating $\btheta_{(i)}^*$ then relies on the distance between   $\E l\left( \cdot, \btheta_{(-i)}^* \right)$ and $l\left( \cdot, \hat{\btheta}_{(-i)} \right)$, which on the other hand depends on both the $\ell_2$ bound of $\|\hat{\btheta}- \btheta^*\|_2$ and the uniform local deviation bound of  $l\left(\btheta_{(i)},\btheta_{(-i)} \right)$. 
Based on Theorem \ref{thm2}, Corollary \ref{Cor1} in Section \ref{S.4}, and a sequential approach (see equations \eqref{s_seq1} and \eqref{s_seq2} in the appendix), we obtain the following bound for the estimation error under the $l_\infty$ norm.
\begin{thm}\label{thm3}
	Let condition (A1) hold,  assume $\{\bX^t\} \sim P_{\btheta^*}$, 
	{\color{black}and $\kappa_{r} :=\|\bbeta_{r}^*\|_\infty$ where $r=0,1$ with $\kappa_{0}+\kappa_{1} \le c\log (np)$ for some small enough constant $c >0$. Then as $np\rightarrow \infty, n\ge 2$,  it holds with probability converging to one that
		\begin{equation*}
			\left\|\hat{\btheta}-\btheta^*\right\|_{\infty}\lesssim e^{8\kappa_{0}+8\kappa_{1}} \log\log(np) \sqrt{\frac{\log(np)}{np}}\left(1+ \frac{\log(np)}{\sqrt{p}} \right).
		\end{equation*}
	}
\end{thm}
Theorem \ref{thm3} indicates that $\left\|\hat{\btheta}-\btheta^*\right\|_{\infty}=o_{p}(1)$ as $np\rightarrow \infty$. Thus all the components of $\hat\btheta$ converge uniformly. On the other hand, when $\kappa= c\log(np)$ for some small enough positive constant $c$, we have
{  $e^{8\kappa_{0}+8\kappa_{1}} \log\log(np) \sqrt{\frac{\log(np)}{np}}\left(1+ \frac{\log(np)}{\sqrt{p}} \right)\leq o(c_{r}e^{-4\kappa_{0}-4\kappa_{1}})$.}	Compared with Theorem \ref{thm1}, we observe that although the radius $r$ in Theorem \ref{thm1} already tends to zero  when {\color{black}$\|\bbeta_{0}^*\|_{\infty}\le \kappa_{0},\|\bbeta_{1}^*\|_{\infty}\le \kappa_{1}$ and $ \kappa_{0}+ \kappa_{1} \le c\log(np)$ for some small enough constant $c>0$,
}
the $\ell_\infty$ error bound of $\hat{\btheta}$ has a smaller order asymptotically and thus gives a tighter convergence rate. 

We remark that in the MLE, $\bbeta_{0}^*$ and  $\bbeta_{1}^*$ are estimated jointly. As we can see from the log-likelihood function, the information related to $\beta_{i,0}$ is captured by $X_{i,j}^t$ and $X_{i,j}^tX_{i,j}^{t-1}$, $t=1,\ldots, n, j\ne i$, while that  related to $\beta_{i,1}$ is captured by $(1-X_{i,j}^t)(1-X_{i,j}^{t-1})$ and $X_{i,j}^tX_{i,j}^{t-1}$, $t=1,\ldots, n , j\ne i$. This indicates that the effective ``sample sizes” for estimating $\beta_{i,0}$ and $\beta_{i,1}$ are both of the order $O(np)$.  While the theorems we have established in this section is for $\hat{\btheta}=({\hat\bbeta_{0}}^\top,{\hat \bbeta_{1}}^\top)^\top$ jointly, we  would  expect ${\hat\bbeta_{0}}$ and  ${\hat \bbeta_{1}}$ to have the same rate of convergence.

\subsection{A method of moments estimator} \label{S3.4}
Having established the existence of a unique local MLE in $\bB_{\infty}\left(\btheta^*,r \right)$
and proved its convergence,  we still need to specify how to find this local MLE.  To this end, we propose an initial estimator lying in this neighborhood. Consequently we can adopt any convex optimization method such as the coordinate descent algorithm to locate the local MLE, thanks to the convexity of the loss function in this neighborhood. 
Based on (\ref{margP}),  an initial estimator of $\bbeta_0$ denoted as $\tilde{\bbeta}_{(0)}$ can be found by solving the following method of moments equations
\begin{equation}
	\label{MME1} \frac{\sum_{t=1}^{n} \sum_{j=1,\:j \neq i}^{p} X_{i,j}^t}{n} - \sum_{j=1,\:j \neq i}^{p} \frac{ e^{\beta_{i,0} +\beta_{j,0}} }{1+ e^{\beta_{i,0} +\beta_{j,0 }}}=0,\quad i=1,\cdots,p.
\end{equation}
These equations can be viewed as the score functions of the  pseudo loss function   $f(\bbeta_0):=\sum_{1\leq i,j\leq p}\log\{1+e^{\beta_{i,0} +\beta_{j,0 }}\}-n^{-1}\sum_{i=1}^p \{\beta_{i,0}{\sum_{t=1}^{n} \sum_{j=1,\:j \neq i}^{p} X_{i,j}^t}\}$. Since the Hessian matrix of  $f(\bbeta_0)$ is diagonally balanced with positive elements,  the Hessian matrix is positive definite, and, hence, $f(\bbeta_0)$ is strongly convex.   With the strong convexity, the solution of \eqref{MME1} {  is the minimizer of $f(\cdot)$} which can be easily obtained using any standard algorithms such as the gradient descent.  
On the other hand, note that
\begin{equation*} \E (X_{i,j}^tX_{i,j}^{t-1})  = \frac{e^{{\beta}_{i,0} + {\beta}_{j,0}}}{1+e^{{\beta}_{i,0} + {\beta}_{j,0}}}\left(1- \frac{1}{1+e^{{\beta}_{i,0} + {\beta}_{j,0}}+e^{\beta_{i,1} +\beta_{j,1}}}\right),
\end{equation*}
which motivates the use of the following estimating equations to obtain $\tilde\bbeta_1$, the initial estimator of $\bbeta_1$, 
\begin{equation}\label{MME2}
	\sum_{t=1}^{n} \sum_{j=1,\:j \neq i}^{p} \left\{X_{i,j}^tX_{i,j}^{t-1} -  \frac{e^{\tilde{\beta}_{i,0} + \tilde{\beta}_{j,0}}}{1+e^{\tilde{\beta}_{i,0} + \tilde{\beta}_{j,0}}}\left(1- \frac{1}{1+e^{\tilde{\beta}_{i,0} + \tilde{\beta}_{j,0}}+e^{\beta_{i,1} +\beta_{j,1}}}\right)\right\}=0,
\end{equation}
with $ i=1,\cdots,p$. {  Similar to \eqref{MME1}, we can formulate a pseudo loss function such that given  $\tilde{\bbeta}_{0}$, its Hessian matrix corresponding to the score equations \eqref{MME2} is positive definite, and hence	\eqref{MME2}  can also be solved via the standard gradient descent algorithm.  } 
Since $\tilde{\btheta}=(\tilde\bbeta_0^\top, \tilde\bbeta_1^\top)^\top$ is obtained by solving two sets of moment equations, we call it the method of moments estimator (MME). 
An interesting aspect of our construction of these moment equations is that the equations corresponding to the estimation of $\bbeta_0$ and $\bbeta_1$ are decoupled. While the estimator error in estimating $\bbeta_0$ propagates clearly in that of estimating $\bbeta_1$, we have the following existence, uniqueness, and a uniform upper bound for the estimation error of $\tilde{\btheta}$. Our results build on a novel application of  
the classical interior mapping theorem  \citep{gragg1974optimal, yan2012approximating, yan2016asymptotics}.
\begin{thm}\label{thm4}
	Let condition (A1) hold,  and $\{\bX^t\} \sim P_{\btheta^*}$. The MME $\tilde{\btheta}$ defined by  equations \eqref{MME1} and \eqref{MME2}  exists and is unique in probability.	{\color{black}Further, assume that  $\kappa_{r} :=\|\bbeta_{r}^*\|_\infty$ where $r=0,1$ with $\kappa_{0}+\kappa_{1} \le c\log (np)$ for some small enough constant $c >0$. Then as $np \to \infty$ and $n\ge 2$, it holds that	
		\begin{equation*}
			\left\|\tilde{\btheta}-\btheta^{*}\right\|_{\infty}\leq O_{p}\left(  e^{14\kappa_{0}+6\kappa_{1}}\sqrt{\frac{\log(n)\log(p)}{np}}\right) .
		\end{equation*}
	}
\end{thm}

When $np\to \infty$ and {\color{black}$\kappa_{0},\kappa_{1}$ are finite}, Theorem \ref{thm4} gives $\left\|\tilde{\btheta}-\btheta^{*}\right\|_{\infty}= O_p\left(\sqrt{\frac{\log(n)\log(p)}{np}}\right)$.  When { $\kappa_{0}+\kappa_{1}\asymp \log(np)$}, we see that the upper bound for the local MLE in Theorem \ref{thm3} is dominated by the upper bound of the MME in Theorem \ref{thm4}.   Moreover, when {\color{black}$\kappa_{0}+\kappa_{1} \le c\log(np)$ for some small enough constant $c>0$, }
we have  $\tilde{\btheta}\in \bB_{\infty} \left(\btheta^*,r \right)$, where $r$ is defined in Theorem \ref{thm1}. Thus,  $\tilde{\btheta}$ is in the small neighborhood of $\btheta^*$ as required.    
{\color{black}
	\subsection{The sparse case}
	In the previous results, the estimation error bounds depend on $\kappa_0$ and $\kappa_1$, i.e.,  the upper bounds  on $\|\bbeta_{0}^*\|_\infty$ and  $\|\bbeta_{1}^*\|_\infty$. Clearly, the larger $\kappa_0$ is, the more sparse the networks could be, and the larger $\kappa_1$ is, the lag-one correlations (c.f. equation \eqref{b9}) could be closer to one, indicating fewer fluctuations in the network process. To further characterize the effect of network sparsity,  in this section, we derive further properties under a relatively sparse scenario where $-\kappa_{0}\leq \beta_{i,0}^* \leq C_{\kappa}$ and $-\kappa_{1}\leq \beta_{i,1}^*\leq \kappa_{1}$ for all $i=1,\ldots, p$ and $C_{\kappa}>0$ here is a constant. Under this case,  there exist  constants $C>0$ and $C_{1}>0$ such that 
	$
	Ce^{-2\kappa_{0}}\leq \E \left( X_{i,j}^{t}\right) \leq C_{1}< 1.
	$
	In the most sparse case where $\beta_{0,i}=-\kappa_0, i=1,\ldots, p$,  the density of the stationary network is of the order $O(e^{-2\kappa_{0}})$.
	Similar to  Lemma \ref{bound_l2} and Theorem \ref{thm1}, the following corollary provides a lower bound for the smallest eigenvalue of $\E(\bV(\btheta))$ and the existence of the MLE.  
	\begin{coro}\label{cor1} 
		Let $\{\bX^t\} \sim P_{\btheta^*}$, 
		$\bB_{\infty}\left( \btheta^*,r \right) = \left\{ \btheta: \|\btheta-\btheta^*\|_{\infty}\le r \right\}$ for some 	$r = c_r e^{-2\kappa_{0}-4\kappa_{1}}$ where $c_r>0$ is a small enough constant. 
		Denote  $\bB'\left( \kappa_{0},\kappa_{1} \right) := \Big\{ \left(\bbeta_{0}, \bbeta_{1}\right) :- \kappa_{0}\le \bbeta_{i,0}\le C_{\kappa}, i=1,\ldots, p ,$ $ \|\bbeta_{1}\|_{\infty}\le \kappa_{1} \Big\}$ for some   constant $C_\kappa>0$. Then, 
		under condition  (A1), there exists a constant $C>0$ such that
		\begin{equation*} 
			\inf_{\btheta \in \bB_{\infty}\left( \btheta^*,r \right)\cap \bB'\left( \kappa_{0},\kappa_{1} \right) ;\|\ba\|_{2}=1}\ba^{\top}\E\left( \bV(\btheta)\right) \ba\geq  C e^{-2\kappa_{0}-4\kappa_{1}}.
		\end{equation*}
		Further, assume that $\btheta^*\in \bB'\left( \kappa_{0},\kappa_{1} \right)$ and $\kappa_{0} +2\kappa_{1}<c\log(np)$ for some positive constant $c<1/6$. Then,  as $np\rightarrow \infty$ with $n\ge 2$,   with probability tending to 1, there exists a unique MLE   in $\bB_{\infty}\left( \btheta^*,r \right)$.
		
	\end{coro}
	
	Following Theorems \ref{thm2}-\ref{thm4}, we also   establish the estimation errors for the MLE and MME in the subsequent corollaries below. 
	\begin{coro}\label{cor2} Let condition (A1) hold. Assume  $\{\bX^t\} \sim P_{\btheta^*}$, $\|\bbeta_{1}^*\|_{\infty}\le \kappa_{1} $, and $-\kappa_{0}\leq \beta_{i,0}^*\leq C_{\kappa}$ for $i=1, \ldots, p$,  and some constant $C_\kappa>0$.
		Then as $np\rightarrow \infty$ with $n\ge 2$, it holds with probability tending  to one that
		\begin{equation*}
			\frac{1}{\sqrt{p}}\left\|\hat{\btheta}-\btheta^*\right\|_{2}\leq C e^{2\kappa_{0}+4\kappa_{1}}\sqrt{\frac{\log(np)}{np}}\left(1+ \frac{\log(np)}{\sqrt{ p}} \right),
		\end{equation*}	
		
		\begin{equation*}
			{\rm and}	~~~	\left\|\hat{\btheta}-\btheta^*\right\|_{\infty}\leq C e^{4\kappa_{0}+8\kappa_{1}} \log\log(np) \sqrt{\frac{\log(np)}{np}}\left(1+ \frac{\log(np)}{\sqrt{p}} \right).
		\end{equation*}
	\end{coro}

	\begin{coro}\label{cor3}
		Let condition (A1) hold. Assume  $\{\bX^t\} \sim P_{\btheta^*}$, $\|\bbeta_{1}^*\|_{\infty}\le \kappa_{1} $, and $-\kappa_{0}\leq \beta_{i,0}^*\leq C_{\kappa}$ for $i=1, \ldots, p$  and some constant $C_\kappa>0$.
		Then as $np\rightarrow \infty$ with $n\ge 2$, it holds with probability tending  to one that
		the MME $\tilde{\btheta}$ defined in  equations \eqref{MME1} and \eqref{MME2}  exists uniquely, and when $\kappa_{0} +2\kappa_{1}<c\log(np)$ for some constant $c<1/12$, it holds that 
		\begin{equation*}
			\left\|\tilde{\btheta}-\btheta^{*}\right\|_{\infty}\leq O_{p}\left(  e^{4\kappa_{0}+6\kappa_{1}}\sqrt{\frac{\log(n)\log(p)}{np}}\right) .
		\end{equation*}
	\end{coro}
	
	From Corollary \ref{cor2}, we can see that when $\kappa_1\asymp O(1)$, the MLE is consistent when $\kappa_0\le c\log(np)$ for some positive constant $c<1/8$, with the corresponding lower bound in the density as $O(e^{-2c\log(np)}) \succ O((np)^{-1/4})$.
	Similarly,  from Corollary \ref{cor3} we can see that when $\kappa_1\asymp O(1)$, the density of the networks can be as small as $O(e^{-2c\log(np)})$ for some constant $c<1/12$, i.e., the density has a larger order than  $(np)^{-1/6}$ for the estimation of the MME. 
	Further, 
	when $6\kappa_{0}+10\kappa_{1} \le c_1\log(np)$ for some constant $c_1<1/2$, 
	we have  $\tilde{\btheta}\in \bB_{\infty} \left(\btheta^*,r \right)$, where $r$ is defined in Corollary \ref{cor1}. This implies the validity of using  $\tilde{\btheta}$   as an initial estimator for computing the local MLE.

}
\section{A  uniform local deviation bound under high dimensionality}\label{S.4}
As we have discussed, a key to establish the consistency of the local MLE is to evaluate the magnitude of $\big|[  l(\btheta) - \E  l(\btheta)]-[ l(\btheta^*) - \E  l(\btheta^*)]\big|$ for all $\btheta \in \bB_{\infty}\left(\btheta^*,r \right)$ with $r$ specified in Theorem \ref{thm1}. Such local deviation bounds are important for establishing error bounds for general M-estimators in the empirical processes  \citep{van1996weak}.  
Note that
\begin{align}\label{decomp}
	l(\btheta)-\E l(\btheta)&
	=	   -\frac{1}{p}  \sum_{1\leq i<j\leq p  } \Bigg\{  \left( \beta_{i,0}+\beta_{j,0}\right)     \Big(\frac{a_{i,j}-\E(a_{i,j})}{n}\Big) \\+  & \log\left(1+e^{\left( \beta_{i,1}+ \beta_{j,1}\right)}\right) \Big(\frac{d_{i,j}-\E(d_{i,j})}{n}\Big)\nonumber\\
	+& \log\left(1+e^{\left( \beta_{i,1}-\beta_{i,0}\right)  +\left( \beta_{j,1}-\beta_{j,0}\right)}\right)\Big(\frac{b_{i,j}-\E(b_{i,j})}{n}\Big) \Bigg\} \nonumber 
\end{align}
where  $a_{i,j}, b_{i,j}$ and $d_{i,j}$ are defined in \eqref{abc}.
The three terms on the right-hand side all admit the following form
\begin{equation}\label{Lfunction}
	\bL\left( \btheta\right) = \frac{1}{p}\sum_{1\leq i\neq j\leq p  }l_{i,j} \left( \theta_i, \theta_j\right) Y_{i,j},
\end{equation}
for some functions $\bL:\mathbb{R}^{p}\to  \mathbb{R} $, $l_{i,j}:\mathbb{R}^{2}\to \mathbb{R} $,  and centered random variables $Y_{i,j}$ $(1\leq i,j\leq p)$. Instead of establishing the uniform bound for each term in \eqref{decomp} separately, below we will establish a unified result for bounding $|\bL\left( \btheta\right)- \bL\left( \btheta'\right)|$ over a local $\ell_\infty$ ball defined as $ \btheta\in\bB_{\infty}( \btheta',\cdot)$ for 
a general $\bL$ function as in \eqref{Lfunction}. We remark that in general without further assumptions on $\bL$, establishing uniform deviation bounds is impossible when the dimension of the problem diverges. For our TWHM however, the decomposition  \eqref{decomp} is of a particularly appealing structure in the sense that only two-way interactions between parameters $\theta_i$ exist. Based on
this ``sparsity” structure, we develop a novel reformulation (c.f. equation \eqref{S2_bound}) for the main components of the Taylor series of 
$\bL(\btheta)$ satisfying the following two conditions. 

\begin{itemize}
	\item [(L-A1)]     There exists a constant $\alpha>0$, such that for any $1\le i\ne j\le p$, any positive integer $k$, and any non-negative integer $s\le k$, we have:
	\begin{equation*}
		\frac{\partial^k l_{i,j}\left( \theta_i,\theta_j\right) }{\partial \theta_{i}^{s}\partial \theta_{j}^{k-s} }\leq \frac{\left(k-1 \right)!}{\alpha^k}.
	\end{equation*}
	
	\item [(L-A2)]   Random variables $Y_{i,j}, 1\le i \ne j\le p$ are independent satisfying  $\E\left(Y_{i,j} \right)=0 $, $|Y_{i,j}| \le b_{(p)}$ and $\var\left( Y_{i,j}\right)\le \sigma_{(p)}^2 $ for any $i$ and $j$, where $b_{(p)}$ and $\sigma_{(p)}^2$ are constants depending on $n$ and $p$ but independent of $i$ and $j$.
\end{itemize}
Loosely speaking, Condition  (L-A1)  can be seen as a smoothness assumption on the higher order derivatives of  $l_{i,j}\left( \theta_i,\theta_j\right)$ so that we can properly bound these derivatives when Taylor expansion is applied. On the other hand, the upper bound for these derivatives is mild as it can diverge very quickly as $k$ increases.  
{\color{black} For our TWHM, 
	it can be verified that  (L-A1) holds for $l_{i,j}(\theta_i,\theta_j)=\theta_i+\theta_j$ and $l_{i,j}(\theta_i,\theta_j)=\log(1+e^{\theta_i+\theta_j})$; see \eqref{eq:lossfunction}.  }
For the latter,  note that the first derivative of function $l(x)=\log(1+e^x)$ is seen as the Sigmoid function:
\begin{equation*}
	S\left( x\right)= \frac{e^x}{1+e^x}= \frac{1}{1+e^{-x}}.
\end{equation*}
By the expression of the higher order derivatives of the Sigmoid function \citep{minai1993derivatives}, the $k$-th order derivative of $l$ is
\begin{equation*}
	\frac{\partial^k l\left( x\right) }{\partial x^{k} }=\frac{\sum_{m=0}^{k-2} -A\left(k-1,m \right) \left(-e^{x} \right)^{m+1}  }{\left(1+e^{x} \right)^{k} }, 
\end{equation*} 
where $k\geq 2$ and $A\left(k-1,m \right)$ is the Eulerian number. 	Now for any $x$, we have
\begin{equation*}
	\left|\frac{\sum_{m=0}^{k-2} -A\left(k-1,m \right) \left(-e^{x} \right)^{m+1}  }{\left(1+e^{x} \right)^{k} }\right|\leq \sum_{m=0}^{k-2}A\left(k-1,m \right)=\left( k-1\right) !.
\end{equation*}
Therefore,
\begin{equation*}
	\left|\frac{\partial^k l\left( x\right) }{\partial x^{k} }\right|\leq  {\left(k-1 \right)! }
\end{equation*}
holds for all $x\in \mathbb{R}$ and $k\ge 2$. With extra arguments using the chain rule, this in return implies that (L-A1) is satisfied with $\alpha=1$ when   $l_{i,j}(\btheta) =\log\left(1+e^{\theta_i+\theta_j} \right)$.

Condition (L-A2) is a regularization assumption for the random variables $Y_{i,j}, 1\leq i, j\le p$, and the bounds on their moments are imposed to ensure point-wise concentration. For our TWHM,  from     Lemma \ref{mixing} and Lemma \ref{var_bound}, we have that there exist  large enough constants $C>0$ and $c>0$ such that  with probability greater than $1-(np)^{-c}$, the random variables $\frac{a_{i,j}-\E(a_{i,j})}{n}$, $\frac{b_{i,j}-\E(b_{i,j})}{n}$ and $\frac{d_{i,j}-\E(d_{i,j})}{n}$ all 
satisfy condition (L-A2) with   $ b_{(p)}= C \sqrt{n^{-1}\log(np)}+C n^{-1}\log\left(n\right) \log\log\left(n\right) \log\left(np\right)$ and $\sigma_{(p)}^2=Cn^{-1} $.

We present the uniform upper bound on the deviation of  $\bL(\btheta)$ below.
\begin{thm}\label{thm5}
	Assume conditions  (L-A1)  and  (L-A2). For any given $\btheta'\in \mathbb{R}^{p}$ and $\alpha_{0}\in (0,\alpha/2)$,   there exist large enough constants $C>0$ and $c>0$ which are independent of $\btheta'$, such that, as $np\to \infty$, with probability greater than $1-(np)^{-c}$,  
	\begin{equation*}
		\left|\bL\left( \btheta\right)-\bL\left( \btheta'\right)\right|
		\le C\frac{ b_{(p)}\log(np)+ \sigma_{(p)}\sqrt{p\log(np)}}{p} \left\|\btheta-\btheta'\right\|_{1} 
	\end{equation*}
	holds uniformly for all $\btheta\in \bB_{\infty}\left( \btheta',\alpha_{0}\right)$.
\end{thm}
One of the main difficulties in analyzing $\bL(\btheta)$ defined in \eqref{Lfunction} is that $l_{i,j}(\theta_i,\theta_j)$ and $Y_{i,j}$ are coupled, giving rise to complex terms involving both in the  Taylor expansion of $\bL(\btheta)$.
When Taylor expansion with order $K$ is used, condition (L-A1) can reduce the number of higher order terms from $O(p^{K})$ to $O(p^22^{K})$. On the other hand, by formulating the main terms in the Taylor series into a matrix form in  \eqref{S2_bound}, the 
uniform convergence of the sum of these terms is equivalent to that of the spectral norm of a centered random matrix, which is independent of the parameters.  Further details can be found in the proofs of Theorem \ref{thm5}.

Define the marginal functions of $\bL\left( \btheta\right)$ as
\begin{equation*}
	\bL_{i}\left( \btheta\right) = \frac{1}{p}\sum_{j=1,\:j \neq i}^{p} l_{i,j}\left( \theta_i,\theta_j \right) Y_{i,j}, \quad i=1, \ldots, p,
\end{equation*}
by retaining only those terms related to $\theta_i$. 
Similar to Theorem \ref{thm5}, we state the following upper bound for these marginal functions. With some abuse of notation, let $\btheta_{-i}:=\left(\theta_{1},\cdots,\theta_{i-1},\theta_{i+1},\cdots,\theta_{p} \right)^\top$   be the vector containing all the elements in  $\btheta$ except $\theta_{i}$.
\begin{thm}\label{thm6}
	If conditions  (L-A1)  and  (L-A2)  hold, then     for any given $\btheta'\in \mathbb{R}^{p}$ and $\alpha_{0}\in (0,\alpha/2)$,   there exist large enough constants $C>0$ and $c>0$ which are independent of $\btheta'$, such that, as $np\to \infty$,   with probability greater than $1-(np)^{-c}$,
	\begin{multline*}
		\left|\bL_{i}\left( \btheta\right)-\bL_{i}\left( \btheta'\right)\right|\\
		\le C\frac{b_{(p)}}{p} \left\|\btheta_{-i}-\btheta_{-i}' \right\|_{1}+C\left( \left\|\btheta_{-i}-\btheta_{-i}' \right\|_{1}+1\right) \left|\theta_{i} -\theta'_{i} \right|\frac{b_{(p)} \log(np)+ \sigma_{(p)}\sqrt{p\log(np)}}{p}
	\end{multline*}
	holds uniformly for all $\btheta\in \bB_{\infty}\left( \btheta',\alpha_{0}\right)$,
	and $i=1,\cdots,p$. 
\end{thm}

Similar  to \eqref{decomp}, we can also decompose  $ l\left(\btheta_{(i)} ,\btheta_{(-i)} \right)- \E l \left(\btheta_{(i)},\btheta_{(-i)} \right)$ into the sum of three components taking the form \eqref{Lfunction}. Consequently, by setting $\btheta'$ in Theorems \ref{thm5} and \ref{thm6} to be   the true parameter  $\btheta^*$, we can obtain   the following upper bounds. 
\begin{coro}\label{Cor1} 
	For any given $0<\alpha_{0}<1/4$, there exist large enough positive constants $c_1, c_2,$ and $C$  such that    
	\begin{itemize}
		\item[(i)]   with probability greater than $1-(np)^{-c_1}$, 
		\begin{equation}\label{L_up1}
			\left| \big(l(\btheta)-l(\btheta^*)\big) - \big( \E l(\btheta)-\E l(\btheta^*)\big)\right| 
			\leq C_1\left(1+ \frac{\log(np)}{\sqrt{ p}} \right)\sqrt{\frac{\log(np)}{n}} \left\|\btheta-\btheta^*\right\|_{2}
		\end{equation}
		holds  uniformly for all $\btheta\in \bB_{\infty}\left(\btheta^*,\alpha_{0} \right)$ with some constant $\alpha_{0}<1/2$;
		\item[(ii)]   with probability greater than $1-(np)^{-c_2}$, 
		\begin{multline}\label{L_up2}
			\left|l\left(\btheta_{(i)},\btheta_{(-i)}^* \right) -l\left(\btheta_{(i)}^*, \btheta_{(-i)}^{*}\right)-\left[\E l\left( \btheta_{(i)} ,\btheta_{(-i)}^* \right) -\E l\left(\btheta_{(i)}^*, \btheta_{(-i)}^{*}\right) \right]\right|\\
			\leq C_{2}\left(1+ \frac{\log(np)}{\sqrt{ p}} \right)\sqrt{\frac{\log(np)}{n}}\left\|\btheta_{(i)}-\btheta^*_{(i)}\right\|_{2}
		\end{multline}
		holds  uniformly for all $\btheta_{(i)}\in \bB_{\infty}\left(\btheta^*_{(i)},\alpha_{0} \right)$ with some constant $\alpha_{0}<1/2$.
	\end{itemize}
	
\end{coro}
In \eqref{L_up1} and \eqref{L_up2} we have replaced the $\ell_1$ norm based upper bounds in Theorems \ref{thm5} and \ref{thm6} with $\ell_2$ norm based upper bounds using the fact that for all $\bx\in \mathbb{R}^{p}$, $\|\bx\|_{1}\leq \sqrt{p}\|\bx\|_{2}$. {\color{black} 
	It is recognized that networks often exhibit diverse characteristics, including dynamic changes, node heterogeneity, homophily, and transitivity, among others. In this paper, our primary emphasis is on addressing node heterogeneity within dynamic networks. When integrated with other stylized features, the objective function may adopt a similar structure to the $\bL(\btheta)$ function defined in Equation  { $\eqref{Lfunction}$}. Moreover, many other models that incorporate node heterogeneity can express their log-likelihood functions in a form analogous to Equation \eqref{Lfunction}. For instance, the general category of network models with edge formation probabilities represented as $f(\alpha_i,\beta_j)$, where $f(\cdot)$ is a density or probability mass function, and $(\alpha_i, \beta_i)$ denote node-specific parameters for node  $i$. This encompasses models such as the $p_1$ model \citep{holland1981}, the directed $\beta$-model \citep{Yan:Leng:Zhu:2015}, and the bivariate gamma model \citep{han2020bivariate}. Additionally, in the domain of ranking data analysis, it is common to introduce individual-specific parameters or scores for ranking, as observed in the classical Bradley-Terry model and its variants \citep{han2023general}. Our discoveries have the potential for application in the theoretical examination of these models or their modifications when considering additional stylized features alongside node heterogeneity.
}

\section{Numerical study}\label{sec: numerical study}
In this section, we assess the  performance of the local MLE. For comparison, we have also computed a regularized MME that is numerically more stable than the vanilla MME in \eqref{MME2}. Specifically, for the former, we solve
\begin{equation}\label{MME3}
	\frac{-1}{np}\sum_{t=1}^{n} \sum_{j=1,\:j \neq i}^{p} \left\{X_{i,j}^tX_{i,j}^{t-1} -  \frac{e^{\tilde{\beta}_{i,0} + \tilde{\beta}_{j,0}}}{1+e^{\tilde{\beta}_{i,0} + \tilde{\beta}_{j,0}}}\left(1- \frac{1}{1+e^{\tilde{\beta}_{i,0} + \tilde{\beta}_{j,0}}+e^{\beta_{i,1} +\beta_{j,1}}}\right)\right\}+ \lambda\beta_{i,1} =0,
\end{equation}
with $i=1,\cdots,p$, where $\lambda\beta_{i,1}$ can be seen as 
a ridge penalty with $\lambda>0$ as the regularization parameter.   Denote the 
regularized MME as $\tilde{\btheta}_\lambda$.  
Similar to Theorem \ref{thm4}, by choosing  $\lambda=C_{\lambda}e^{2\kappa}\sqrt{\frac{\log(np)}{np}}$ for some constant $C_{\lambda}$, we can show that
$
\left\|\tilde{\btheta}_\lambda-\btheta^{*}\right\|_{\infty}\leq O_{p}\left(  e^{26\kappa}\sqrt{\frac{\log(n)\log(p)}{np}}\right).
$ In our implementation we take $\lambda=\sqrt{\frac{\log(np)}{np}}$. The MLE of TWHM is obtained via gradient descent using $\tilde{\btheta}_\lambda$ as the initial value.

\subsection{Non-convexity of $l(\btheta)$ and $\E l(\btheta)$}\label{Sim.H}
{ Given the form of $l(\btheta)$, it is intuitively true that it may not be convex everywhere. 
	We confirm this via a simple example. 
		Take $(n, p)=(2, 1000)$ and set $\bbeta^*_0, \bbeta^*_1$ to be $\textbf{0.2}_{p}$, $\textbf{0.5}_{p}$ or $\textbf{1}_{p}$. We evaluate the smallest eigenvalue of the Hessian matrix of $l(\btheta)$  and its expectation $\E l(\btheta)$  at  the true parameter value $\btheta^*=(\bbeta^{*\top}_0, \bbeta^{*\top}_1)^\top$, or at $\btheta=\textbf{0}_{2p}$ in one experiment. 	 
		\begin{table}[tb]\caption{Signs of the smallest eigenvalues of the Hessian matrices of $l(\btheta)$ and   $\E l(\btheta)$  evaluated at $\btheta=\btheta^*$ or $\textbf{0}_{2p}$ when different values of $\btheta^*=(\bbeta^{*\top}_0, \bbeta^{*\top}_1)^\top$ are used to generate data.}\label{Table_Convexity}
			\centering
			\begin{tabular}{cccc|cccc}
				\hline
				\multicolumn{4}{c|}{Sign of the smallest eigenvalue of $l(\btheta^*)$}&\multicolumn{4}{c}{Sign of the smallest eigenvalue of $\E l(\btheta^*)$}\\
				\hline
				&$\bbeta^*_{0}$= $\textbf{0.2}_{p}$&$\bbeta^*_{0}$= $\textbf{0.5}_{p}$&	$\bbeta^*_{0}$=  $\textbf{1}_{p}$&& $\bbeta^*_{0}$= $\textbf{0.2}_{p}$	&$\bbeta^*_{0}$= $\textbf{0.5}_{p}$&	$\bbeta^*_{0}$=  $\textbf{1}_{p}$		\\				
				$\bbeta^*_{1}$= $\textbf{0.2}_{p}$&	$+$&	$+$&	$+$&$\bbeta^*_{1}$= $\textbf{0.2}_{p}$&	$+$ &  $+$&  $+$	\\
				$\bbeta^*_{1}$= $\textbf{0.5}_{p}$&	$+$&	$+$&	$+$&$\bbeta^*_{1}$= $\textbf{0.5}_{p}$&	$+$&$+$&	$+$	\\
				$\bbeta^*_{1}$= $\textbf{1}_{p}$	&$+$	&$+$	&$+$&$\bbeta^*_{1}$= $\textbf{1}_{p}$&	$+$&	$+$&	$+$	\\
				\hline 
				\multicolumn{4}{c|}{Sign of the smallest eigenvalue of $l(\textbf{0}_{2p})$}&\multicolumn{4}{c}{Sign of the smallest eigenvalue of $\E l(\textbf{0}_{2p})$}\\\hline
				&$\bbeta^*_{0}$= $\textbf{0.2}_{p}$	&$\bbeta^*_{0}$= $\textbf{0.5}_{p}$&	$\bbeta^*_{0}$=  $\textbf{1}_{p}$&&$\bbeta^*_{0}$= $\textbf{0.2}_{p}$	&$\bbeta^*_{0}$= $\textbf{0.5}_{p}$&	$\bbeta^*_{0}$=  $\textbf{1}_{p}$		\\	
				$\bbeta^*_{1}$= $\textbf{0.2}_{p}$&	$+$&	$+$&	$-$&$\bbeta^*_{1}$= $\textbf{0.2}_{p}$&	$+$&	$+$&	$-$\\
				$\bbeta^*_{1}$= $\textbf{0.5}_{p}$&	$-$	&$-$&	$-$&$\bbeta^*_{1}$= $\textbf{0.5}_{p}$&	$+$&	$+$&	$-$\\
				$\bbeta^*_{1}$= $\textbf{1}_{p}$	&$-$&$-$	&$-$&$\bbeta^*_{1}$= $\textbf{1}_{p}$&	$-$&$	-$&	$-$\\
				\hline	
			\end{tabular}
		\end{table}	
		From the top half of Table \ref{Table_Convexity} we can see that, when evaluated at $\btheta^*$, the Hessian matrices are all positive definite. However, when evaluated at $\btheta=\textbf{0}_{2p}$, from the bottom half of the table we can see that the Hessian matrices are no longer positive definite when $\btheta^*$ is far away from $\textbf{0}_{2p}$. Even when  the Hessian matrix of $\E l(\btheta)$ is so at $\btheta=\textbf{0}_{2p}$ with $\btheta^*=\textbf{0.5}_{2p}$, the corresponding Hessian matrix of $l(\btheta)$ at this point has a negative eigenvalue. Thus, $\E l(\btheta)$ and $l (\btheta)$ are not globally convex.
	}

	\subsection{Parameter estimation} 
	
	We first evaluate the error rates of the MLE and MME under different combinations of $n$ and $p$.  We set $n=2, 5, 10,$ or $20$ and $p \in \lfloor 200\times 1.2^{0:6}\rfloor = \{200, 240, 288, 346, 415, 498, 598\}$, which results in a total of 28 different combinations of $(n,p)$. 
	For each $(n,p)$, the data are generated  such that $\{\bX^t\} \sim P_{\btheta^*}$ where the parameters $\beta_{i,0}^*$ and $\beta_{i,1}^*$($1 \leq i\leq p$) are drawn independently  from the uniform distribution with parameters in $(-1,1)$. Each experiment is repeated 100 times under each setting.  
	Denote the estimator (which is either the MLE or the MME)   as $\hat{\btheta} $, and the true parameter value as $\btheta^*$. We report the average $\ell_2$ error in terms of $\frac{\|\hat{\btheta}-\btheta^*\|_2}{\sqrt{p}}$ and the average $\ell_\infty$ error  $\|\hat{\btheta}-\btheta^*\|_{\infty}$ in Figure \ref{fig:error}.
	\begin{figure}[tb]
		\includegraphics[width=1\textwidth]{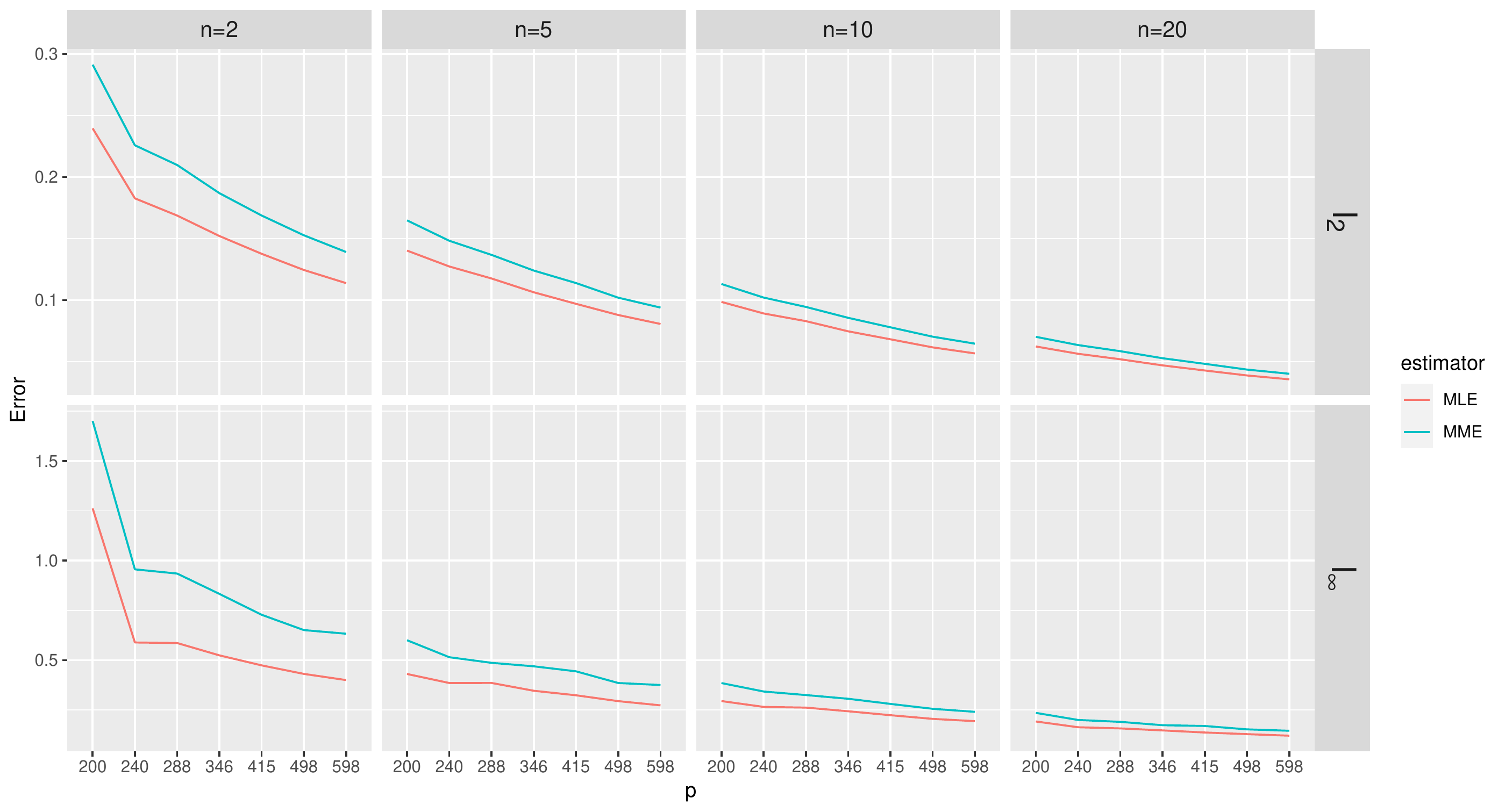}
		\caption{Mean errors of  MME and MLE in terms of the $\ell_2$ and $\ell_{\infty}$ norm.  
		}
		\label{fig:error}
	\end{figure}
	From this figure, we can see that the errors in terms of the   $\ell_{\infty}$ norm and the   $\ell_2$ norm decrease for MME and MLE as $n$   or $p$ increases, while the errors of MLE are smaller across all settings. These observations are 	 consistent with our findings in the main theory.

	Next, we provide more numerical simulation to evaluate the performance of MLE and MME by imposing different structures on $\bbeta_{0}^*$ and $\bbeta_{1}^*$. 
	In particular, we want to evaluate how the estimation accuracy changes by varying the sparsity of the networks 
	as well as varying the correlations of the network sequence.
	Note that the expected density of the stationary distribution of the network process is simply  
	\begin{equation*}\label{ED}
		\frac{1}{p(p-1)}\E\left( \sum_{1\leq i\neq j\leq p}X_{i,j}^{t}\right)= \frac{1}{p(p-1)} \left( \sum_{1\leq i\neq j\leq p}\frac{e^{\beta_{i,0}^*+\beta_{j,0}^*}}{1+e^{\beta_{i,0}^*+\beta_{j,0}^*}}\right).
	\end{equation*}
	In the sequel, we will use two parameters $a$ and $b$ to generate
	$\bbeta_{r}^*$, $r=0,1$, according to the following four settings:    
	
	
	\begin{itemize}   
		\item[Setting 1.] $\{a\}$: all the elements    in $\bbeta_r^*$ are set to be equal to $a$.
		\item[Setting 2.] ${\{a,b\}}$: the first 10$\%$ elements of $\bbeta_r^*$ are set to be equal to  $a$, while the other elements are set to be equal to $b$. 
		\item[Setting 3.] $\mathcal{L}_{(a,b)}$: the parameters take	values in a linear form  as $\beta_{i,r}^*= a + (a-b)*(i-1)/(p-1)$, $i=1,\cdots,p$. 
		\item[Setting 4.]  $U_{(a,b)}$: the $p$ elements in $\bbeta_r^*$  are generated independently from the uniform distribution with parameters $a$ and $b$. 	
	\end{itemize}
	In Table \ref{T2}, we generate $\bbeta_1^*$ using Setting 1 with $a=0$, and generate $\bbeta_0^*$ using Setting 2 with different choices for $a$ and $b$ to  obtain networks with different expected density. In Table \ref{T3}, we generate $\bbeta_0^*$ and $\bbeta_1^*$ using combinations of these four settings with different parameters  such that the resulting networks have expected density either around 0.05 (sparse) or 0.5 (dense). The number of networks in each process and the number of nodes in each network are set as  $(n, p)=(20, 200),(20, 500),(50, 200)$ or $(50, 500)$. The errors for estimating $\btheta^*$ in terms of the $\ell_\infty$ and    $\ell_2$  norms are reported via 100 replications.  To further compare the  accuracy for estimating $\bbeta_0^*$ and $\bbeta_1^*$, in Table \ref{T5}, we have  conducted experiments under Settings 3 and 4, and reported the estimation errors for $\bbeta_0^*$ and $\bbeta_1^*$ separately. We summarize the simulation results below:
	\begin{table}[tb]
		\centering
		\caption{The estimation errors of  MME and MLE  under Setting 1 and Setting 2 for $\bbeta_0^*$ by setting $\bbeta^*_1=\boldsymbol{0}_p$.}
		\begin{tabular}{|c|c|c|Hc|c|c|cH|}
			\hline
			$n$     & $p$     & $\bbeta_{0}^*$ & $\bbeta_{1}^*$ & MME, $\ell_{2}$ & MME, $\ell_\infty$ & MLE, $\ell_{2}$ & MLE, $\ell_\infty$  & Density \\
			\hline
			20    & 200   & \{0\} & \{0\}     & 0.074&	0.219&	0.071&	0.212	 & 0.5 \\
			\hline
			50    & 200  & \{0\} & \{0\}    & 0.046&	0.138&	0.045&	0.136 & 0.5 \\
			\hline
			20    & 500   & \{0\} & \{0\}     &0.046&	0.150 &	0.045&	0.146& 0.5 \\
			\hline
			50    & 500   & \{0\} & \{0\}     & 0.029&	0.093&	0.028&	0.092& 0.5 \\
			\hline
			20    & 200   & $\{0.5,-0.5\}$ & \{0\}     &0.092  & 0.222 & 0.091  & 0.217 & 0.32 \\
			\hline
			50    & 200   & $\{0.5,-0.5\}$ & \{0\}     & 0.058 & 0.140 & 0.058 & 0.139 & 0.32 \\
			\hline
			20    & 500   &$\{0.5,-0.5\}$ & \{0\}     & 0.058 & 0.154 &0.057  & 0.148 & 0.32 \\
			\hline
			50    & 500   & $\{0.5,-0.5\}$& \{0\}     & 0.036  & 0.095 & 0.036  & 0.093 & 0.32 \\
			\hline
			20    & 200   & $\{1,-1\}$& \{0\}    &0.120	&0.305	&0.117&0.284	 & 0.19 \\
			\hline
			50    & 200   & $\{1,-1\}$ & \{0\}     &0.074	&0.186&	0.074&0.177	 & 0.19 \\
			\hline
			20    & 500   & $\{1,-1\}$ & \{0\}    & 0.075	&0.200	&0.073&	0.190	& 0.19 \\
			\hline
			50    & 500   & $\{1,-1\}$ & \{0\}   & 0.038&	0.125&	0.036&	0.119		 & 0.19 \\
			\hline
			20    & 200   & $\{1.5,-1.5\}$ & \{0\}     & 0.164  & 0.436 & 0.156  & 0.397 & 0.14 \\
			\hline
			50    & 200   & $\{1.5,-1.5\}$ & \{0\}    & 0.102  & 0.255 & 0.097  & 0.236 & 0.14 \\
			\hline
			20    & 500   & $\{1.5,-1.5\}$  & \{0\}    & 0.103  & 0.287 & 0.097 & 0.262 & 0.14 \\
			\hline
			50    & 500   & $\{1.5,-1.5\}$  & \{0\}    & 0.065  & 0.178 & 0.061  & 0.164 & 0.14\\
			\hline
		\end{tabular}%
		\label{T2}
	\end{table}%
	
	\begin{table}[tb]
		\centering
		\caption{The average estimation errors  of MME and MLE under combinations of different settings}
		\begin{tabular}{|c|c|c|c|c|c|c|c|}
			\hline
			\multicolumn{8}{|c|}{Density = 0.05}\\
			\hline
			$n$     & $p$     & $\bbeta_{0}^*$ & $\bbeta_{1}^*$ & MME, $\ell_{2}$ & MME, $\ell_\infty$ & MLE, $\ell_{2}$ & MLE, $\ell_\infty$   \\
			\hline
			20    & 200   & $\mathcal{L}_{(-4,0)}$ &  $U_{(-1,1)}$ & 0.419  & 1.833 & 0.392 & 1.8   \\
			\hline
			50    & 200   & $\mathcal{L}_{(-4,0)}$  &  $U_{(-1,1)}$ & 0.253  & 0.913 & 0.227  & 0.82  \\
			\hline
			20    & 500   & $\mathcal{L}_{(-4,0)}$  &  $U_{(-1,1)}$ & 0.246  & 1.119  & 0.218 & 0.9   \\
			\hline
			50    & 500   & $\mathcal{L}_{(-4,0)}$  &  $U_{(-1,1)}$ & 0.170  & 0.626 & 0.148   & 0.621 \\
			\hline
			20    & 200   & $\mathcal{L}_{(-4,0)}$ &  \{0\} &0.275&	1.452&	0.280&	1.516  \\
			\hline
			50    & 200   & $\mathcal{L}_{(-4,0)}$  &  \{0\} & 0.161&	0.771&	0.162&	0.774   \\
			\hline
			20    & 500   & $\mathcal{L}_{(-4,0)}$  &  \{0\} & 0.160&	0.892&	0.162&	0.904   \\
			\hline
			50    & 500   & $\mathcal{L}_{(-4,0)}$  &  \{0\} &0.098&	0.506&	0.099&	0.507 \\
			\hline
			20    & 200   & $\{-1.47\}$ &  $U_{(-1,1)}$ &0.187&	0.588&	0.161&	0.514		 \\
			\hline
			50    & 200   & $\{-1.47\}$  &  $U_{(-1,1)}$ & 0.116&	0.351&0.099&	0.305	 \\
			\hline
			20    & 500   & $\{-1.47\}$  &  $U_{(-1,1)}$ & 0.114&	0.387&	0.099&	0.339	 \\
			\hline
			50    & 500   & $\{-1.47\}$  &  $U_{(-1,1)}$ & 0.073&	0.246&	0.062&	0.208	 \\
			\hline
			20    & 200   & $\{-1.47\}$ &  \{0\} &0.150&	0.482&	0.151&	0.484		 \\
			\hline
			50    & 200   & $\{-1.47\}$  &  \{0\} & 0.93&	0.289&0.093&	0.29 \\
			\hline
			20    & 500   & $\{-1.47\}$  &  \{0\} & 0.93&	0.309	&0.093&	0.311	 \\
			\hline
			50    & 500   & $\{-1.47\}$  &  \{0\} & 0.058	&0.195	&0.058	&0.195	 \\
			\hline
			\multicolumn{8}{|c|}{Density = 0.5}\\
			\hline
			20    & 200   & $\mathcal{L}_{(-2,2)}$ & $U_{(-0.1,0.1)}$ & 0.132  & 0.415 & 0.012  & 0.318  \\
			\hline
			50    & 200   & $\mathcal{L}_{(-2,2)}$ & $U_{(-0.1,0.1)}$ & 0.080 & 0.238 & 0.069 & 0.194  \\
			\hline
			20    & 500   & $\mathcal{L}_{(-2,2)}$ & $U_{(-0.1,0.1)}$ & 0.080  & 0.272 & 0.068  & 0.217  \\
			\hline
			50    & 500   & $\mathcal{L}_{(-2,2)}$ &$U_{(-0.1,0.1)}$ & 0.050  & 0.168 & 0.043 & 0.135  \\
			\hline
			20    & 200   & $\mathcal{L}_{(-1,1)}$ & $U_{(-1,1)}$ & 0.107 & 0.324 & 0.095  & 0.264  \\
			\hline
			50    & 200   & $\mathcal{L}_{(-1,1)}$ &  $U_{(-1,1)}$ & 0.067 & 0.194 & 0.060 & 0.163  \\
			\hline
			20    & 500   & $\mathcal{L}_{(-1,1)}$ &  $U_{(-1,1)}$ & 0.071 & 0.267 & 0.061  & 0.205  \\
			\hline
			50    & 500   & $\mathcal{L}_{(-1,1)}$ &  $U_{(-1,1)}$ & 0.044 & 0.156 & 0.039 & 0.130   \\
			\hline
			20    & 200   & $\mathcal{L}_{(-2,2)}$ &  $U_{(-1,1)}$ & 0.137  & 0.478 & 0.112  & 0.329  \\
			\hline
			50    & 200   & $\mathcal{L}_{(-2,2)}$ &  $U_{(-1,1)}$ & 0.084  & 0.274 & 0.070 & 0.205  \\
			\hline
			20    & 500   & $\mathcal{L}_{(-2,2)}$ &  $U_{(-1,1)}$ & 0.087  & 0.352 & 0.071  & 0.250   \\
			\hline
			50    & 500   & $\mathcal{L}_{(-2,2)}$ &  $U_{(-1,1)}$ & 0.054  & 0.211 & 0.044 & 0.150   \\
			\hline
		\end{tabular}%
		\label{T3}
	\end{table}
	
	\begin{table}[tb]
		\small
		\centering
		\caption{The means and standard deviations of the errors of MME and MLE for estimating $\bbeta_0^*$ and $\bbeta_1^*$.}
		\begin{tabular}{|c|c|c|c|c|c|}
			\hline
			\multicolumn{2}{|c|}{$n$} & 20        & 100   & 20        & 100 \\
			\hline
			\multicolumn{2}{|c|}{$p$} & 200     & 200   & 500      & 500 \\
			\hline
			\multicolumn{6}{|c|}{$\bbeta_0^*\sim \mathcal{L}_{(-1,1)}$ and $\bbeta_1^*\sim U_{(0,2)}$}\\
			\hline
			\multirow{2}{*}{MME, $\ell_{2}$} & $\bbeta_{0}^*$ & 0.163(0.010)&		0.096(0.006)&	0.099(0.004)	&	0.057(0.002)\\
			\cline{2-6}          & $\bbeta_{1}^*$ & 0.177(0.010)	&	0.084(0.005)&	0.104(0.004)&	0.050(0.002)\\
			
			\hline
			\multirow{2}{*}{MME, $\ell_{\infty}$} & $\bbeta_{0}^*$ & 0.570(0.103)  & 0.367(0.085) & 0.395(0.070) &  0.241(0.042) \\
			\cline{2-6}          & $\bbeta_{1}^*$ & 0.658(0.137) & 0.421(0.079)  & 0.438(0.076)  & 0.214(0.037) \\
			\hline
			\multirow{2}{*}{MLE, $\ell_{2}$} & $\bbeta_{0}^*$ & 0.211(0.013)	&	0.091(0.006)&	0.121(0.005)&	0.054(0.002)\\
			\cline{2-6}          & $\bbeta_{1}^*$ & 0.166(0.011)&	0.072(0.005)&	0.096(0.004)&		0.043(0.002)\\
			\hline
			\multirow{2}{*}{MLE, $\ell_{\infty}$ } & $\bbeta_{0}^*$ & 0.809(0.180)  & 0.354(0.076) & 0.532(0.098) & 0.232(0.041) \\
			\cline{2-6}          & $\bbeta_{1}^*$ & 0.617(0.116)  & 0.265(0.052) & 0.399(0.065) &  0.172(0.028) \\
			\hline
			\multicolumn{6}{|c|}{$\bbeta_0^*\sim \mathcal{L}_{(-2,0)}$ and $\bbeta_1^*\sim U_{(0,2)}$}\\
			\hline
			\multirow{2}{*}{MME, $\ell_{2}$} & $\bbeta_{0}^*$ & 0.133(0.012)&	0.080(0.007)&	0.081(0.004)&		0.047(0.002)\\			
			\cline{2-6}          & $\bbeta_{1}^*$ & 0.093(0.006)&	0.053(0.004)&	0.056(0.003)&	0.032(0.002)\\
			\hline
			\multirow{2}{*}{MME, $\ell_{\infty}$} & $\bbeta_{0}^*$ & 0.568(0.104) & 0.365(0.087) & 0.394(0.071) &  0.241(0.042) \\
			\cline{2-6}          & $\bbeta_{1}^*$ & 0.387(0.069) & 0.236(0.043) & 0.258(0.037) &  0.162(0.025) \\
			\hline
			\multirow{2}{*}{MLE, $\ell_{2}$} &$\bbeta_{0}^*$ &0.176(0.016)&		0.076(0.007)&	0.100(0.006)&		0.044(0.002)	 \\
			\cline{2-6}          & $\bbeta_{1}^*$ & 0.116(0.009)&	0.051(0.004)&	0.068(0.003)&		0.031(0.002) \\
			\hline
			\multirow{2}{*}{MLE, $\ell_{\infty}$ } &$\bbeta_{0}^*$ & 0.809(0.181) & 0.351(0.078) & 0.531(0.099) &  0.232(0.041) \\
			\cline{2-6}          & $\bbeta_{1}^*$ & 0.513(0.088) &  0.227(0.047) & 0.348(0.058) & 0.158(0.024) \\
			\hline
		\end{tabular}%
		\label{T5}
	\end{table}%

	\begin{itemize}
		\item  The effect of $(n,p)$.  Similar to what we have observed in Figure \ref{fig:error}, the estimation errors become smaller when $n$ or $p$ becomes larger.    Interestingly, from Tables \ref{T2}--\ref{T3} we can observe that, under the same setting,  the errors in $\ell_{2}$ norm when $(n,p)=(50,200)$ are very close to those when $(n,p)=(20,500)$. This is to some degree consistent with our finding in Theorem \ref{thm2}  where the upper bound depends on $(n,p)$ through their product $np$. 
		\item  The effect of sparsity.  From Table \ref{T2} we can see that, as the expected density decreases, the estimation errors increase in almost all the cases. On the other hand, even though the parameters take different values in Table \ref{T3}, the errors in the sparse cases are in general larger than those in the dense cases.

		\item   {\color{black} The impact of  $\kappa_{0}:=\|\bbeta_{0}^*\|_\infty$. Typically, estimation errors tend to increase with larger values of $\kappa_0$, as evidenced in Table \ref{T2}. Additionally, when maintaining the same overall sparsity level, larger $\kappa_0$ values are correlated with greater estimation errors, as illustrated in Table \ref{T3}.		}

		\item  MLE vs MME.  In general,  the estimation errors of the MLE are  smaller than those of the MME in most cases as can be seen in Tables \ref{T2} and Table \ref{T3}. In Table \ref{T5} where the estimation errors for $\bbeta_0^*$ and $\bbeta_1^*$ are reported separately, we 
		can see that the estimation errors of the MME of $\bbeta_{1}^*$ are generally larger  than those of the MLE of $\bbeta_{1}^*$, especially when $n$ is large. 
	\end{itemize}

	\subsection{Real data}
	In this section, we  apply our TWHM to a real dataset to examine an insect interaction network process \citep{mersch2013tracking}. 	We focus on a subset of the data named insecta-ant-colony4  that contains the social interactions of  102 ants in 41 days. In this dataset, the position and orientation of all the ants were recorded twice per second to infer their movements and interactions, based on which 41 daily networks were constructed.  
	More specifically, 
	$X_{i,j}^{t}$ is 1  if there is an interaction between ants $i$ and $j$ during day $t$, and 0 otherwise.
	In the ACF and PACF plots of the degree sequences of selected ants (c.f. Figure 1   in Appendix B.1),  we  can observe  patterns similar to those of a first-order autoregressive   model with long memory. This motivates the use of TWHM for the analysis of this dataset.
	
	In \cite{mersch2013tracking}, the 41 daily networks  were split  into four periods with 11, 10, 10, and 10 days respectively, because the corresponding days separating these periods were identified as change-points. By excluding ants that did not interact with others, we are left with $p=102$ nodes in period one, $p=73$ nodes in period two, $p=55$ nodes in period three and $p=35$ nodes in period four.  Thus we take the networks on day 1, day 12, day 22 and day 32  as the initial networks and fit four different TWHMs, one for each of the four periods.

	To appreciate how TWHM captures static heterogeneity, we present a subgraph of 10 nodes during the fourth period ($t=32$--$41$), 5 of which have the largest and 5 have the smallest fitted $\beta_{i,0}$ values. The edges of this subgraph are drawn to represent aggregated static connections  defined as
	$(\bX^{32}+\cdots+\bX^{42})/10$ between these ants. We can see from the left panel of Figure \ref{fig:sub}  that the magnitudes of the fitted static heterogeneity parameters agree in principle with the activeness of each ant making connections.
	\begin{figure}[tb]
		\begin{minipage}[b]{0.45\textwidth}
			\centering
			\includegraphics[width=1\textwidth]{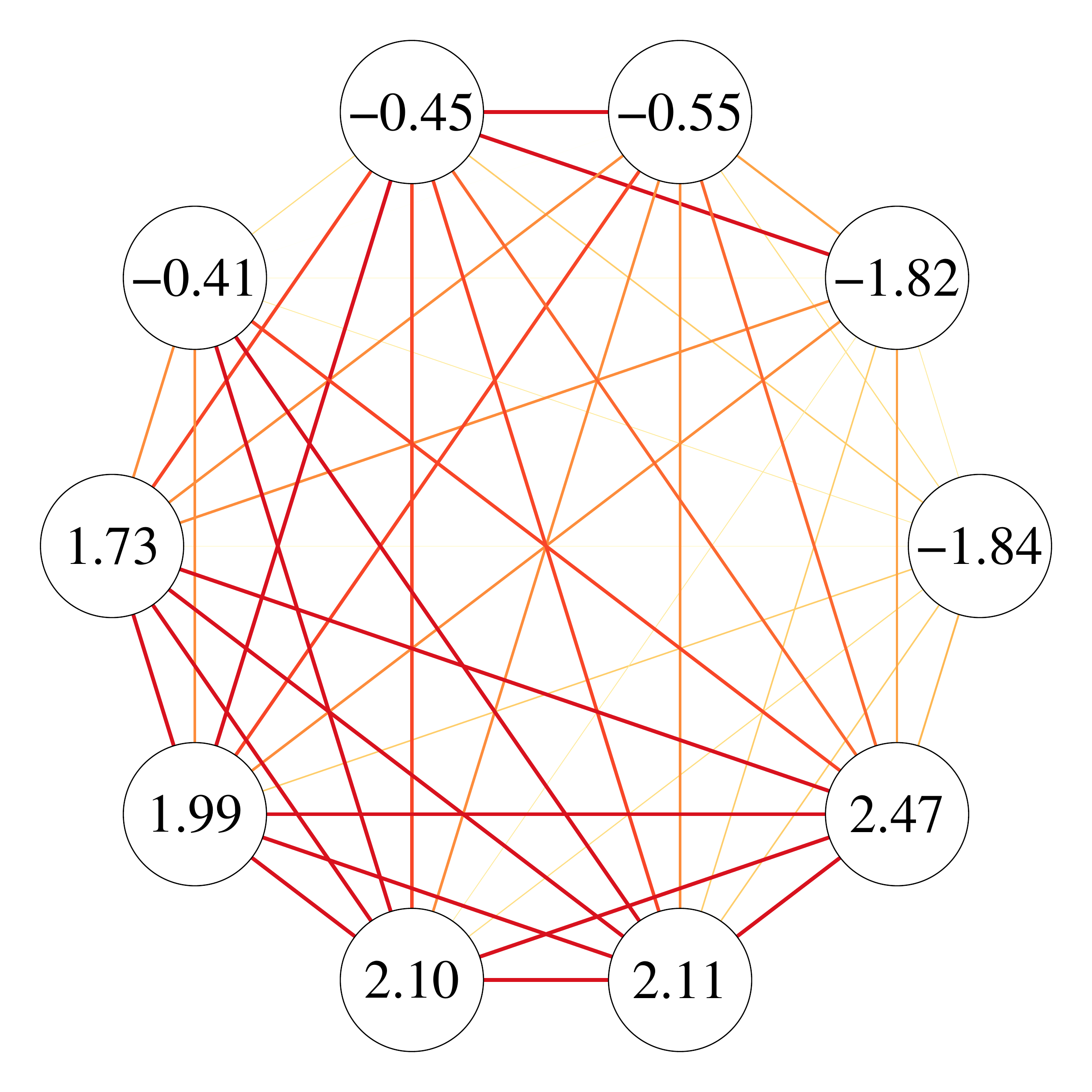}
		\end{minipage}
		\hspace{5mm}
		\begin{minipage}[b]{0.45\textwidth}
			\centering
			\includegraphics[width=1\textwidth]{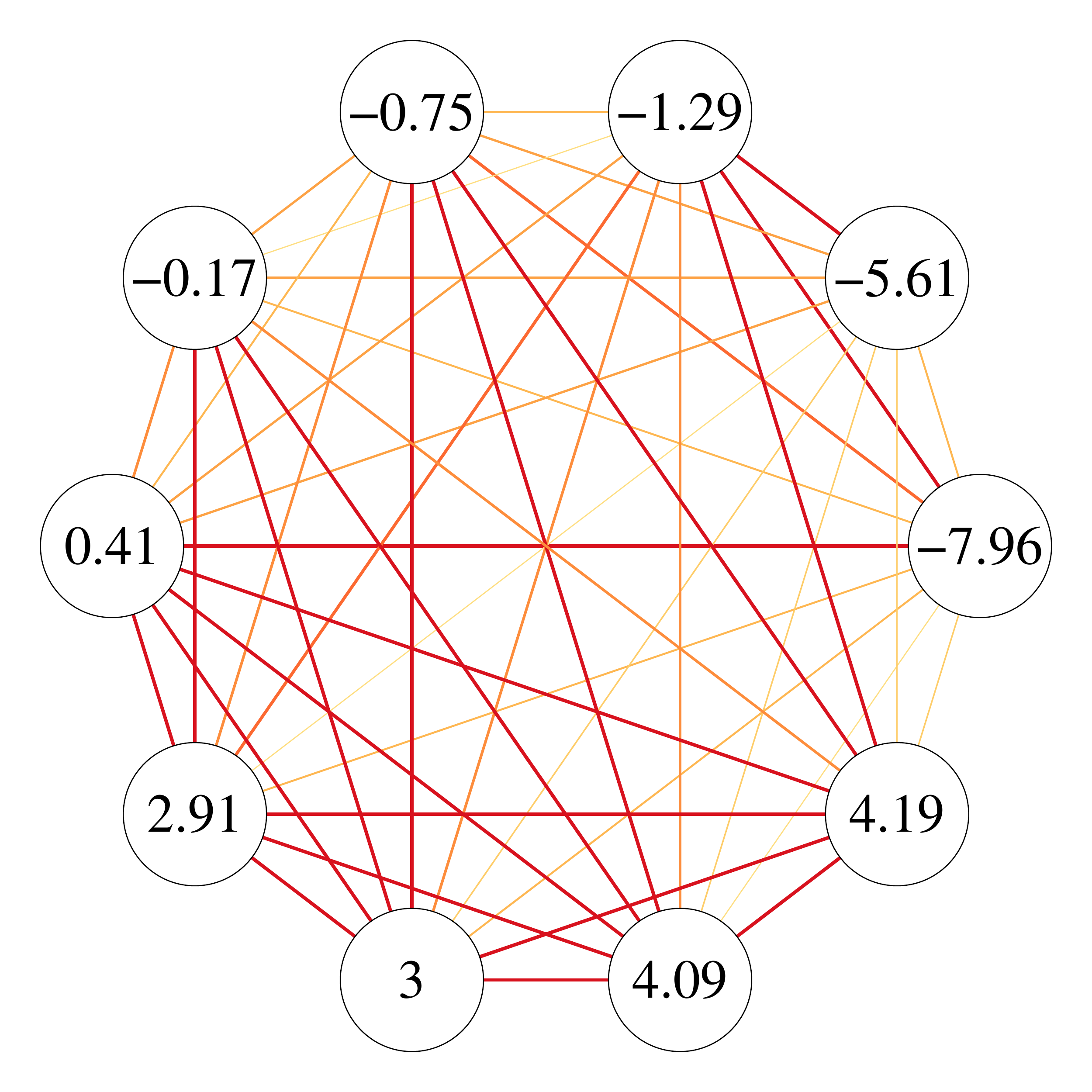}
		\end{minipage}
		\caption{The aggregated networks of 10 selected ants during the fourth period reflect static heterogeneity (Left) and dynamic heterogeneity (Right) respectively.  The thickness of each edge is proportional to the aggregation. The number in the nodes are the fitted $\beta_{i,0}$ (Left) and $\beta_{i,1}$ (Right).
		}
		\label{fig:sub}
	\end{figure}
	On the other hand, we examine how TWHM can capture dynamic heterogeneity. Towards this, we plot a subgraph of the 10 nodes having the smallest fitted $\beta_{i,0}$ values in Figure \ref{fig:sub}(b), where edges represent the magnitude of $\sum_{t=33}^{41} I\left(X_{i,j}^{t}=X_{i,j}^{t-1}\right)/9$ which is a measure of the extent that an edge is preserved across the whole period and hence dynamic heterogeneity. Again, we can see an agreement between the fitted $\bbeta^*_1$ and how likely these nodes will preserve their ties.		 
	
	To evaluate how TWHM performs when it comes to making prediction, we further carry out the following experiments:
	\begin{itemize}
		\item[(i)] From \eqref{Def1}, given the MLE $\{\hat{\beta}_{i,r}, i=1,\ldots, p, r=0,1\}$ and the network  at time $t-1$, we can estimate the conditional expectation of node $i$'s degree as
		\begin{eqnarray*}
			\tilde{d}_{i}^{t}&:=& \sum_{j=1,\:j \neq i}^{p}\E\left( X_{i,j}^{t}\Big|X_{i,j}^{t-1},\hat{\btheta}\right)\\
			& =&\sum_{j=1,\:j \neq i}^{p} \left( \frac{e^{\hat{\beta}_{i,0}+\hat{\beta}_{j,0}}}{1+e^{\hat{\beta}_{i,0} +\hat{\beta}_{j,0}}+e^{\hat{\beta}_{i,1} +\hat{\beta}_{j,1}}}+ \frac{e^{\hat{\beta}_{i,1}+\hat{\beta}_{j,1}}}{1+e^{\hat{\beta}_{i,0} +\hat{\beta}_{j,}}+e^{\hat{\beta}_{i,1} +\hat{\beta}_{j,1}}}X_{i,j}^{t-1}\right).
		\end{eqnarray*} 
		We can then compare the density of the estimated degree sequence $\{\tilde{d}_{i}^{t}, i=1, \ldots, p\}$ with that of the observed degree sequence $\{d_{i}^{t}, i=1,\ldots, p\}$ at time $t$. 	{\color{black}
			To provide a comparison, we treat networks in each period as i.i.d. observations and utilize the classical $\beta$-model to derive the degree sequence estimator $\{\check{\bd}^{t}\}$ for the four periods. The fitted degree distributions are depicted in Figure \ref{fig:ant}, revealing a close resemblance between the estimated and observed densities. This observation suggests that the TWHM demonstrates strong performance in one-step-ahead prediction. 
			
			To further assess the similarity between the estimated degree sequences $\{\tilde{\bd}^{t}\}$, $\{\check{\bd}^{t}\}$, and the true degree sequence $\{\bd^{t}\}$, we compute the Kolmogorov-Smirnov (KS) distance and conduct the KS test for $t = 2, \ldots, 41$. The mean and standard deviation of the KS distances, the p-values of the KS test, and the rejection rate are summarized in Table \ref{tab:KS}. Notably, at a significance level of 0.05, out of the 40 KS tests, we fail to reject the null hypothesis that $\{\tilde{\bd}^{t}\}$ and $\{\bd^{t}\}$ originate from the same distribution in 38 instances, resulting in a rejection rate consistent with the significance level. Conversely, for the degree sequence estimators based on the $\beta$-model $\{\check{\bd}^{t}\}$, 8 out of the 40 tests were rejected. These findings indicate that our model exhibits highly promising performance in recovering the degree sequences.}
		
		\begin{table}[ht]
			\def\arraystretch{1.5}
			\centering
			\caption{The mean and standard deviation  of the KS distances, the  p-values of KS test, and the rejection rate between the true degree sequence $\{\bd^{t}\}$, and the THWM based  estimator	$\tilde{\bd}^{t}$ and the $\beta$-model based estimator $\check{\bd}^{t}$ over the 40 networks ($t=2,\ldots, 41$) in the ant dataset.}
			\begin{tabular}{|c|c|c|c|}
				\hline
				& KS distance & KS test p-value & Rejection rate \\
				\hline
				$\tilde{\bd}^{t}$ vs $\bd^{t}$ & 0.179(0.058) & 0.361(0.267)  & 0.05\\
				\hline
				$\check{\bd}^{t}$ vs $\bd^{t}$ & 0.192(0.061) & 0.298(0.246) & 0.20\\
				\hline
			\end{tabular}%
			\label{tab:KS}%
		\end{table}%
		
		\begin{figure}[!htbp]
			\centering
			\includegraphics[width=0.85\linewidth]{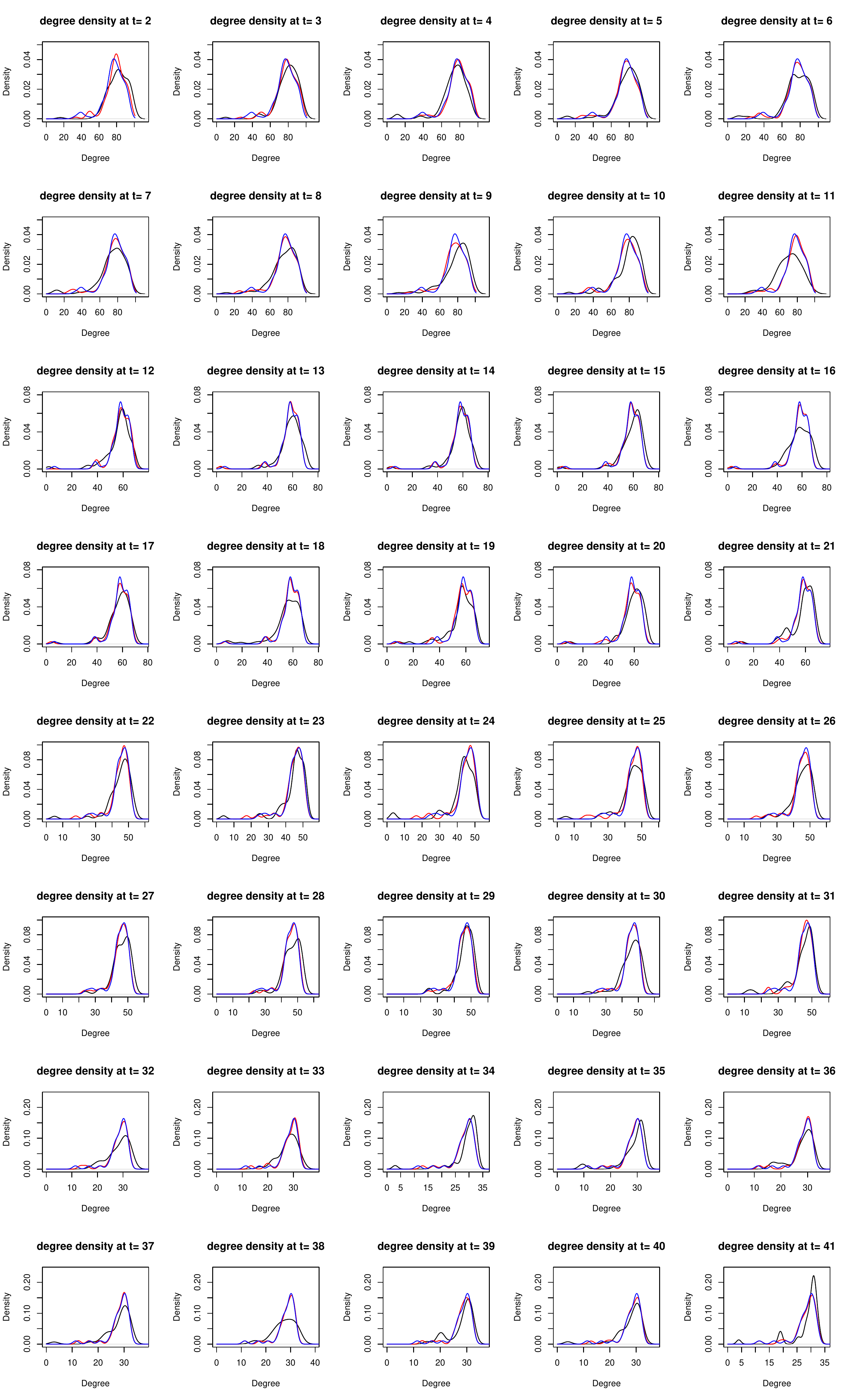}
			\caption{The observed and estimated degree distributions. X-axis: the node degrees; Red curves: the smoothed degree distributions of the estimated degree sequences;  
				Black curves:  the smoothed degree distributions of the observed degree sequences.
			}
		\label{fig:ant}
	\end{figure}
	\item[(ii)] By incorporating network dynamics, TWHM naturally enables one-step-ahead link prediction via
	\begin{equation}\label{linkpre}
		\bP\left( \hat{X}_{i,j}^{t}=1\Big|X_{i,j}^{t-1}\right)=\frac{e^{\hat{\beta}_{i,0}+\hat{\beta}_{j,0}}}{1+ e^{\hat{\beta}_{i,0} +\hat{\beta}_{j,0}}+e^{\hat{\beta}_{i,1} +\hat{\beta}_{j,1}}}+ \frac{e^{\hat{\beta}_{i,1} +\hat{\beta}_{j,1}}}{1+e^{\hat{\beta}_{i,0} +\hat{\beta}_{j,0}}+e^{\hat{\beta}_{i,1} +\hat{\beta}_{j,1}}}X_{i,j}^{t-1}.
	\end{equation}
	To transform these probabilities into links, we threshold them by setting $\hat{X}_{i,j}^{t}=1$ when $\bP\left(\hat{X}_{i,j}^{t}=1\right)$ $\geq c_{i,j}$ and  $\hat{X}_{i,j}^{t}=0$ when $\bP\left( \hat{X}_{i,j}^{t}=1\right)<c_{i,j}$ for some cut-off constants $c_{i,j}$. As an illustration, we  first consider simply setting  $c_{i,j}=0.5$ for all $1\le i<j\le p$ for predicting links. We  shall denote this approach as TWHM$_{0.5}$. 
	
	As an alternative, owing to the fact that networks may change slowly, for a given   parameter $\omega$, we also consider the following adaptive approach for choosing  $c_{i,j}$:
	\begin{equation}\label{es}
		\tilde{X}_{i,j}^{t} :=I\{ \omega\bP\left( \hat{X}_{i,j}^{t}=1\right) + \left(1-\omega \right) X_{i,j}^{t-1}>0.5\}.
	\end{equation}
	It can be shown that the above estimator is equivalent to the prediction rule $I\left\{  \bP\left( \hat{X}_{i,j}^{t}=1\right)  >c_{i,j}\right\}$ with  cut-off values specified as
	\begin{equation*}
		c_{i,j}=\frac{0.5e^{\hat{\beta}_{i,1} +\hat{\beta}_{j,1}}+\left( 1-w\right)e^{\hat{\beta}_{i,0}+\hat{\beta}_{j,0}} }{\left( 1-w\right) +e^{\hat{\beta}_{i,1} +\hat{\beta}_{j,1}}+\left( 1-w\right)e^{\hat{\beta}_{i,0}+\hat{\beta}_{j,0}}},\quad  1\le i<j\le p.
	\end{equation*}
	This method is denoted as TWHM$_{adaptive}$.
	Lastly,  as a benchmark, we have also considered a naive approach that simply predicts $\bX^{t}$ as $\bX^{t-1}$.
	
	
	In this experiment, we set the number of training samples to be $n_{train}=2, 5$ or $8$.  
	For a given training sample size $n_{train}$  and a period with $n$ networks,   we predict the  graph ${\bf X}^{n_{train }+i}$ based on the previous $n_{train}$ networks $\{{\bf X}^{t}, t=i, \ldots,n_{train }+i-1\}$ for  $i=1,\ldots, n-n_{train}$.  That is, over the four periods in the data, we have predicted 33, 21 and 9 networks, with 5151 edges in each network in the first period, 2628 in the second period, 1485 in the third period, and 595 in the fourth period for our choices of $n_{train}$. The $\omega$ parameter employed in TWHM$_{adaptive}$ is selected as follows. 
	{For prediction in each period}, 
	we choose the value in a sequence of $\omega$ values that produces the highest prediction accuracy in predicting $\bX^{n_{train}+i-1}$ for predicting 
	$\bX^{n_{train}+i}$.   
	For example, in the first period with $n=11$ networks, when $n_{train}=8$,  we used  $\{\bX^{t}, t=i,\cdots, i+7\}$ to predict $\bX^{i+8}$ for $i=1, 2, 3$. For each $i$, let $\tilde{\bf X}^{i+7}$ be defined as in \eqref{es}. 
	A set of candidate values for $\omega$ were used to compute $\tilde{\bf X}^{i+7}$, and the one that  returns the smallest misclassification rate (in predicting ${\bf X}^{i+7}$) was used in TWHM$_{adaptive}$ for predicting  $\bX^{i+8}$.   {\color{black} The mean of the chosen $\omega$  is $0.936$ when $n=2$,   $0.895$ when $n=5$, and   $0.905$ when $n=8$.}
	The prediction accuracy of the above-mentioned methods, defined as the percentages of correctly predicted links, are reported in Table \ref{T7}. 
	We can see that TWHM$_{0.5}$ and TWHM$_{adaptive}$ both perform  better than the naive approach in all the cases. {\color{black}On the other hand, TWHM coupled with adaptive cut-off points can improve the prediction accuracy of TWHM with a cur-off value 0.5 in most periods.}

\end{itemize}

\begin{table}[tb]
	\centering
	\caption{The prediction accuracy of TWHM with $0.5$ as a cut-off point, TWHM with adaptive cut-off points, and the naive estimator $\bX^{t-1}$.}
	\begin{tabular}{|c|cccc|}
		\hline
		$n_{train}$ & Period & TWHM$_{0.5}$ & TWHM$_{adaptive}$ & Naive \\
		\hline
		\multirow{5}[2]{*}{ 2 } & One     & 0.773 & 0.800 & 0.749 \\
		& Two     &  0.817 & 0.817 & 0.780 \\
		& Three     & 0.837 & 0.837 & 0.806\\
		& Four     &  0.824 & 0.831 & 0.807 \\
		& Overall & 0.811 &   0.822 &  0.784 \\
		\hline
		\multirow{5}[2]{*}{ 5 } & One     & 0.789 &   0.807 &  0.759 \\
		& Two     & 0.826 &  0.823 & 0.779 \\
		& Three     & 0.846 & 0.849 & 0.805 \\
		& Four     & 0.833 & 0.842 & 0.805 \\
		& Overall &0.822 & 0.829 & 0.786\\
		\hline
		\multirow{5}[2]{*}{ 8 } & One    & 0.795 & 0.800 & 0.759 \\
		& Two     & 0.832 & 0.832 & 0.778 \\
		& Three     & 0.855 &  0.845 & 0.823 \\
		& Four     &  0.831 & 0.863 & 0.779 \\
		& Overall &  0.825 & 0.831&  0.782 \\
		\hline
	\end{tabular}%
	\label{T7}
\end{table}%

\section{Summary and Discussion}\label{sec: conclusion}
We have proposed a novel two-way heterogeneity model that utilizes two sets of parameters  to explicitly capture static heterogeneity and dynamic heterogeneity.  In a high-dimension setup, we have provided the existence and the rate of convergence of its local MLE, and proposed a novel method of moments estimator as an initial value to find this local MLE. To the best of our knowledge, this is the first model in the network literature that the {local} MLE is obtained for a non-convex loss function. The theory of our model is established by developing new uniform upper bounds for the deviation of the loss function.

While we have focused on the estimation of the parameters in this paper, how to conduct statistical inference for the local MLE is a natural next step for research. 		In our setup, we assume that the parameters are time invariant but this need not be the case. A future direction is to allow the static heterogeneity parameter $\bbeta_0$ and/or the dynamic heterogeneity parameter $\bbeta_1$ to depend on time, giving rise to non-stationary network processes. 
{ In case when these parameters change smoothly over time,
	we may consider estimating the parameters $\beta_{i,0}^\tau,\beta_{i,1}^\tau$ at time $\tau$ by kernel smoothing, that is,  by maximizing the following smoothed log-likelihood: 			
	\begin{multline*}
		\tilde{L}(\tau, \bX^n,\bX^{n-1}, \cdots, \bX^1|\bX^0 ) 	\nonumber \\
		=    {\sum_{t=1}^n w_t\sum_{1\leq i<j\leq p  }   \Bigg\{ -\log\Big({1+ e^{\beta_{i,0} +\beta_{j,0 } }+e^{\beta_{i,1} +\beta_{j,1 } }}\Big)+  \left( \beta_{i,0} +\beta_{j,0 }\right)    X_{i,j}^t\left( 1-X_{i,j}^{t-1}\right) }     \nonumber\\
		+ 
		\left( 1-X_{i,j}^t\right) \left( 1-X_{i,j}^{t-1}\right) \log\left(  1+e^{\beta_{i,1} +\beta_{j,1} }\right) 
		+  X_{i,j}^t X_{i,j}^{t-1} \log\big( e^{\beta_{i,0} +\beta_{j,0} }+e^{\beta_{i,1} +\beta_{j,1} }\big) \Bigg\} ,
	\end{multline*}
	with
	$
	w_t=\frac{K(h^{-1}|t-\tau|)}{\sum_{t=1}^nK(h^{-1}|t-\tau|)},
	$
	where $K(\cdot)$ is a kernel function and $h$ is the bandwidth parameter. As another line of research, note that TWHM is formulated as an AR(1) process. 
	We can  extend it by including more time lags. For example, we can extend TWHM to include lag-$k$ dependence by writing 
	\begin{equation*} 
		X^t_{i,j} = I( \ve_{i,j}^t=0)+\sum_{r=1}^{k}X^{t-r}_{i,j} I( \ve_{i,j}^t=r)  ,
	\end{equation*}
	where the innovations
	$\ve_{i,j}^t$ are independent such that
	\begin{equation*} 
		P(\ve_{i,j}^t =r) = \frac{e^{\beta_{i,r}+\beta_{j,r}}}{1+\sum_{s=0}^k e^{\beta_{i,s}+\beta_{j,s}}}~~ {\rm for} ~ r=0,\cdots,k; \quad 	P(\ve_{i,j}^t = -1) =  \frac{1}{1+\sum_{s=0}^k e^{\beta_{i,s}+\beta_{j,s}}},
	\end{equation*}
	with parameter $\bbeta_0=(\beta_{1,0},\ldots, \beta_{p,0})^\top$ denoting node-specific static heterogeneity and $\bbeta= \left( \beta_{i,r}\right)_{1\leq i\leq p;1\leq r\leq k}$ $ \in \mathbb{R}^{p\times k}$ denoting lag-$k$ dynamic fluctuation.
}
Other future lines of research include adding covariates to model the tendency of nodes making connections  \citep{yan2019statistical} and exploring additional structures \citep{chen2021analysis}. 
		
			\newpage
			\bibliographystyle{apalike}
			\bibliography{reference}      
			\newpage
			\begin{appendix}
				\section{Technical proofs}
				For brevity,  we denote
				$\alpha_{0,i,j}:=e^{\beta_{i,0} +\beta_{j,0}}$ and $\alpha_{1,i,j}:=e^{\beta_{i,1} +\beta_{j,1}}$,   and  define	$\alpha^*_{0,i,j}, \alpha^*_{1,i,j}$ and $\hat{\alpha}_{0,i,j}, \hat{\alpha}_{1,i,j}$ similarly based on  the true parameter $\btheta^*$ and the MLE $\hat{\btheta}$.
				
				\subsection{Some technical lemmas}
				Before presenting the proofs for our main results, we first provide some technical lemmas which will be used from time to time in our proofs.
				Lemmas \ref{thm1_lemma} and  \ref{var_bound} below provide further properties about the process $\{\bX^t\}$.
				
				\begin{lem}\label{thm1_lemma}
					Let $\{\bX^t\} \sim P_{\btheta}$
					We have:
					\item[(i)] $\{ \bX^t\circ \bX^{t-1}, t=0, 1, 2, \cdots \}$ where $\circ$  is the Hadamard product operator, is strictly stationary. Furthermore for any $1\leq i<j\leq p, 1\leq \ell<m\leq p  $ and $| t-s| \ge 1$, we have
					\begin{eqnarray*}
						&&\E\left( X_{i,j}^tX_{i,j}^{t-1}\right) =\frac{\alpha_{0,i,j}(\alpha_{0,i,j}+\alpha_{1,i,j})} {(1+\alpha_{0,i,j}) (1+\alpha_{0,i,j} +\alpha_{1,i,j})}, \nonumber\\
						&&\var\left( X_{i,j}^t X_{i,j}^{t-1}\right)  =\frac{\alpha_{0,i,j}(\alpha_{0,i,j}+\alpha_{1,i,j}) (2\alpha_{0,i,j}+ \alpha_{1,i,j}+1)}{(1+\alpha_{0,i,j})^2(1+\alpha_{0,i,j} +\alpha_{1,i,j})^2}, \\
						&&\cov(X_{i,j}^tX_{i,j}^{t-1},X_{l,m}^{s} X_{l,m}^{s-1})= \\
						&&\left\{
						\begin{array}{ccc}
							\left(  \frac{\alpha_{1,i,j}}{1+\alpha_{0,i,j} +\alpha_{1,i,j}} \right) ^{| t-s|-1}\frac{\alpha_{0,i,j}(\alpha_{0,i,j} +\alpha_{1,i,j})^2}{(1+\alpha_{0,i,j})^2(1+\alpha_{0,i,j} +\alpha_{1,i,j})^2} ,  &{(i,j)=(l,m)},\\
							\\
							0, &{(i,j)\neq(l,m)}.
						\end{array}
						\right.
					\end{eqnarray*}
					\item[(ii)] $\{ \left( 1-\bX^t\right) \circ\left( 1- \bX^{t-1}\right) , t=0, 1, 2, \cdots \}$ is strictly stationary. Furthermore for any $1\leq i<j\leq p, 1\leq \ell<m\leq p  $ and $| t-s| \ge 1$, we have
					\begin{eqnarray*}
						&&\E\left(  \Big( 1-X_{i,j}^{t}\Big) \Big( 1-X_{i,j}^{t-1}\Big) \right) =\frac{1+\alpha_{1,i,j}} {(1+\alpha_{0,i,j}) (1+\alpha_{0,i,j} +\alpha_{1,i,j})},\\
						&&\var\left( \Big( 1-X_{i,j}^{t}\Big) \Big( 1-X_{i,j}^{t-1}\Big) \right) =\frac{\alpha_{0,i,j}(1+\alpha_{1,i,j}) (\alpha_{0,i,j}+ \alpha_{1,i,j}+2)}{(1+\alpha_{0,i,j})^2(1+\alpha_{0,i,j} +\alpha_{1,i,j})^2},\\
						&&\cov\left( \Big( 1-X_{i,j}^{t}\Big) \Big( 1-X_{i,j}^{t-1}\Big),(1-X_{l.m}^{s})(1-X_{l,m}^{s-1})   \right) = \\
						&&\left\{
						\begin{array}{ccc}
							\left(  \frac{\alpha_{1,i,j}}{1+\alpha_{0,i,j} +\alpha_{1,i,j}} \right) ^{| t-s|-1}\frac{\alpha_{0,i,j}(1+\alpha_{1,i,j})^2 }{(1+\alpha_{0,i,j})^2 (1+\alpha_{0,i,j} +\alpha_{1,i,j})^2} ,  &{(i,j)=(l,m)},\\
							\\
							0, &{(i,j)\neq(l,m)}.
						\end{array}
						\right.
					\end{eqnarray*}
				\end{lem}
				
				\begin{proof}
					(i)   Denote $\mu_{i,j}=\E\left( X_{i,j}^tX_{i,j}^{t-1}\right) $, $\gamma_{i,j}(k)= \cov\left( X_{i,j}^t X_{i,j}^{t-1},X_{i,j}^{t-k} X_{i,j}^{t-k-1}\right) $ and $\rho_{i,j}(k)=$ $ \gamma_{i,j}(k)/\gamma_{i,j}(1)$ ($k\geq 1$). For every $i< j$, we have
					\begin{eqnarray*}
						&&\E\left( X_{i,j}^tX_{i,j}^{t-1}\right) = P\left(X_{i,j}^t=1\Big|X_{i,j}^{t-1}=1 \right)P\left(X_{i,j}^{t-1}=1\right)=\frac{\alpha_{0,i,j}(\alpha_{0,i,j}+\alpha_{1,i,j})} {(1+\alpha_{0,i,j}) (1+\alpha_{0,i,j} +\alpha_{1,i,j})},
						\\
						&&		\var\left( X_{i,j}^t X_{i,j}^{t-1}\right) = \E\left( \Big(X_{i,j}^tX_{i,j}^{t-1}\Big)^2\right) - \E\left( X_{i,j}^tX_{i,j}^{t-1}\right)^2 \\ &&=\left(1-  \E\left( X_{i,j}^tX_{i,j}^{t-1}\right)\right)  \E\left( X_{i,j}^tX_{i,j}^{t-1}\right) 	=\frac{\alpha_{0,i,j}(\alpha_{0,i,j}+\alpha_{1,i,j}) (2\alpha_{0,i,j}+ \alpha_{1,i,j}+1)}{(1+\alpha_{0,i,j})^2(1+\alpha_{0,i,j} +\alpha_{1,i,j})^2},
					\end{eqnarray*}
					and
					\begin{eqnarray*}
						\gamma_{i,j}(1) &=& \E\left( X_{i,j}^t(X_{i,j}^{t-1})^2X_{i,j}^{t-2}\right) - \mu^2_{i,j}=P\left( X_{i,j}^tX_{i,j}^{t-1}X_{i,j}^{t-2}=1\right) - \mu^2_{i,j}\\
						&=&P\left( X_{i,j}^t\Big| X_{i,j}^{t-1}X_{i,j}^{t-2}=1 \right)P\left(X_{i,j}^{t-1}X_{i,j}^{t-2}=1 \right)   - \mu^2_{i,j}\\
						&=&\E\left( X_{i,j}^t\Big| X_{i,j}^{t-1}=1\right) \mu_{i,j}   -\mu^2_{i,j}\\
						&=& \frac{\alpha_{0,i,j}(\alpha_{0,i,j} +\alpha_{1,i,j})^2}{(1+\alpha_{0,i,j})^2(1+\alpha_{0,i,j} +\alpha_{1,i,j})^2}.
					\end{eqnarray*}
					For $k\geq2$, by Proposition \ref{prop1}, we have,
					\begin{eqnarray*}
						\gamma_{i,j}(k) &=& \E\left( X_{i,j}^tX_{i,j}^{t-1}X_{i,j}^{t-k} X_{i,j}^{t-k-1}\right) - \mu^2_{i,j}\\
						&=&P\left( X_{i,j}^tX_{i,j}^{t-1}=1\Big| X_{i,j}^{t-k}X_{i,j}^{t-k-1}=1\right) P\left(X_{i,j}^{t-k}X_{i,j}^{t-k-1  }=1 \right)  - \mu^2_{i,j}\\
						&=&P\left(X_{i,j}^t=1\Big|X_{i,j}^{t-1}=1 \right)P\left(X_{i,j}^{t-1}=1\Big|X_{i,j}^{t-k}=1 \right)\mu_{i,j}- \mu_{i,j}^2\\
						& =&P\left(X_{i,j}^t=1\Big|X_{i,j}^{t-1}=1 \right)P\left(X_{i,j}^{t-1}X_{i,j}^{t-k}=1 \right)P\left( X_{i,j}^{t-k}=1 \right)^{-1} \mu_{i,j}- \mu^2_{i,j}\\
						&=&P\left(X_{i,j}^t=1\Big|X_{i,j}^{t-1}=1 \right)P\left(X_{i,j}^{t-1}X_{i,j}^{t-k}=1 \right)P\left(X_{i,j}^{t-k+1}=1\Big|X_{i,j}^{t-k}=1 \right)- \mu^2_{i,j}\\
						&=&P\left(X_{i,j}^{t-1}X_{i,j}^{t-k}=1 \right)P\left(X_{i,j}^t=1\Big|X_{i,j}^{t-1}=1 \right)^2- \mu^2_{i,j}\\
						&=&\left(P\left(X_{i,j}^{t-1}X_{i,j}^{t-k}=1 \right)-\E\left(X_{i,j}^{t}\right) ^2 \right)  P\left(X_{i,j}^t=1\Big|X_{i,j}^{t-1}=1 \right)^2\\
						&=&\cov\left( X_{i,j}^{t-1},X_{i,j}^{t-k}\right) \frac{(\alpha_{0,i,j} +\alpha_{1,i,j})^2}{(1+\alpha_{0,i,j} +\alpha_{1,i,j})^2}\\
						&= &\left(  \frac{\alpha_{1,i,j}}{1+\alpha_{0,i,j} +\alpha_{1,i,j}} \right) ^{k-1} \frac{\alpha_{0,i,j}}{(1+\alpha_{0,i,j})^2}\frac{(\alpha_{0,i,j} +\alpha_{1,i,j})^2}{(1+\alpha_{0,i,j} +\alpha_{1,i,j})^2}\\
						&=&\left(  \frac{\alpha_{1,i,j}}{1+\alpha_{0,i,j} +\alpha_{1,i,j}} \right) ^{k-1}\gamma_{i,j}(1).
					\end{eqnarray*}
					This proves (i).
					
					(ii)	Let $\mu_{i,j}'=\E\left(  \Big( 1-X_{i,j}^{t}\Big) \Big( 1-X_{i,j}^{t-1}\Big) \right) $, $\gamma_{i,j}'(k)= \cov\Bigg(  \Big( 1-X_{i,j}^{t}\Big) \Big( 1-X_{i,j}^{t-1}\Big),$ $  \left( 1-X_{i,j}^{t-k}\right) \left( 1-X_{i,j}^{t-k-1}\right) \Bigg) $ and $\rho_{i,j}'(k)=$ $ \gamma_{i,j}'(k)/\gamma_{i,j}'(1)$  ($k\geq 1$). Similarly, for every $i< j$, we have
					\begin{eqnarray*}
						\E\left(  \Big( 1-X_{i,j}^{t}\Big) \Big( 1-X_{i,j}^{t-1}\Big) \right)  &=& P\left(X_{i,j}^t=0\Big|X_{i,j}^{t-1}=0 \right)P\left(X_{i,j}^{t-1}=0\right) \\
						&=&\frac{1+\alpha_{1,i,j}} {(1+\alpha_{0,i,j}) (1+\alpha_{0,i,j} +\alpha_{1,i,j})},\\
						\var\left( \Big( 1-X_{i,j}^{t}\Big) \Big( 1-X_{i,j}^{t-1}\Big) \right)
						&=&\left( 1- \mu'_{i,j}\right) \mu'_{i,j}=\frac{\alpha_{0,i,j}(1+\alpha_{1,i,j}) (\alpha_{0,i,j}+ \alpha_{1,i,j}+2)}{(1+\alpha_{0,i,j})^2(1+\alpha_{0,i,j} +\alpha_{1,i,j})^2},
					\end{eqnarray*}
					and
					\begin{eqnarray*}
						\gamma_{i,j}'(1)&=&\cov\left(  \Big( 1-X_{i,j}^{t}\Big) \Big( 1-X_{i,j}^{t-1}\Big),  \left( 1-X_{i,j}^{t-1}\right) \left( 1-X_{i,j}^{t-2}\right) \right) \\
						&=&\E\left( \Big( 1-X_{i,j}^t\Big) \Big( 1-X_{i,j}^{t-1}\Big)^2 \left( 1-X_{i,j}^{t-2}\right) \right)-  \E\left(\Big( 1-X_{i,j}^{t}\Big) \Big( 1-X_{i,j}^{t-1}\Big) \right)^2\\
						&=&P\left( \Big( 1-X_{i,j}^t\Big) \Big( 1-X_{i,j}^{t-1}\Big)^2 \left( 1-X_{i,j}^{t-2}\right) =1\right)-  \E\left(\Big( 1-X_{i,j}^{t}\Big) \Big( 1-X_{i,j}^{t-1}\Big) \right)^2\\
						&=&P\left((1-X_{i,j}^{t-1})(1-X_{i,j}^{t-2}) =1\right) P\left(X_{i,j}^t=0\Big|X_{i,j}^{t-1}=0 \right)-\left( \mu'_{i,j}\right) ^2\\
						&=&(1+\alpha_{0,i,j})\left( \mu'_{i,j}\right) ^2-\left( \mu_{i,j}'\right) ^2\\
						&=&\alpha_{0,i,j}\left( \mu'_{i,j}\right) ^2\\
						&=&\frac{\alpha_{0,i,j}(1+\alpha_{1,i,j})^2 }{(1+\alpha_{0,i,j})^2 (1+\alpha_{0,i,j} +\alpha_{1,i,j})^2}.
					\end{eqnarray*}
					For   $k\ge 2$ we have,
					\begin{eqnarray*}
						\gamma_{i,j}'(k) &=& \E\left( \Big( 1-X_{i,j}^t\Big) \Big( 1-X_{i,j}^{t-1}\Big) \left( 1-X_{i,j}^{t-k}\right) \left( 1-X_{i,j}^{t-k-1}\right) \right)- \left( \mu'_{i,j}\right) ^2\\
						&=&P\left(\Big( 1-X_{i,j}^t\Big) \Big( 1-X_{i,j}^{t-1}\Big)=1\Big| \left( 1-X_{i,j}^{t-k}\right) \left( 1-X_{i,j}^{t-k-1}\right) =1\right) \mu'_{i,j}  - \left( \mu'_{i,j}\right) ^2\\
						&=&P\left(X_{i,j}^t=0\Big|X_{i,j}^{t-1}=0 \right)P\left(X_{i,j}^{t-1}=0\Big|X_{i,j}^{t-k}=0 \right)\mu'_{i,j}- \left( \mu'_{i,j}\right) ^2\\
						&=&P\left(X_{i,j}^t=0\Big|X_{i,j}^{t-1}=0 \right)P\left(X_{i,j}^{t-1}=0,X_{i,j}^{t-k}=0\right) P\left( X_{i,j}^{t-k}=0 \right)^{-1}\mu'_{i,j}- \left( \mu'_{i,j}\right) ^2,
					\end{eqnarray*}
					with
					\begin{eqnarray*}
						P\left(X_{i,j}^{t-1}=0,X_{i,j}^{t-k}=0\right) &=& \cov\left(  1-X_{i,j}^{t-1},1-X_{i,j}^{t-k}\right) +\E\left(  1-X_{i,j}^{t} \right)^2\\
						&= &\cov\left(  X_{i,j}^{t-1},X_{i,j}^{t-k}\right) + \frac{1}{(1+\alpha_{0,i,j})^2}\\
						&=&\left(  \frac{\alpha_{1,i,j}}{1+\alpha_{0,i,j} +\alpha_{1,i,j}} \right) ^{k-1} \frac{\alpha_{0,i,j}}{(1+\alpha_{0,i,j})^2}+ \frac{1}{(1+\alpha_{0,i,j})^2}\\
						&=&\frac{1}{(1+\alpha_{0,i,j})^2}\left(\left(  \frac{\alpha_{1,i,j}}{1+\alpha_{0,i,j} +\alpha_{1,i,j}} \right) ^{k-1}\alpha_{0,i,j}+1 \right) ,
					\end{eqnarray*}
					and
					\begin{equation*}
						P\left(X_{i,j}^t=0\Big|X_{i,j}^{t-1}=0 \right)\frac{1}{(1+\alpha_{0,i,j})^2} P\left( X_{i,j}^{t-k}=0 \right)^{-1}=\frac{1+\alpha_{1,i,j}} {(1+\alpha_{0,i,j}) (1+\alpha_{0,i,j} +\alpha_{1,i,j})}.
					\end{equation*}
					Thus,
					\begin{eqnarray*}
						&&\gamma_{i,j}'(k)\\
						&=&P\left(X_{i,j}^t=0\Big|X_{i,j}^{t-1}=0 \right)P\left(X_{i,j}^{t-1}=0,X_{i,j}^{t-k}=0\right) P\left( X_{i,j}^{t-k}=0 \right)^{-1}\mu'_{i,j}- \left( \mu'_{i,j}\right) ^2\\
						&=&\left( P\left(X_{i,j}^t=0\Big|X_{i,j}^{t-1}=0 \right)P\left(X_{i,j}^{t-1}=0,X_{i,j}^{t-k}=0\right) P\left( X_{i,j}^{t-k}=0 \right)^{-1}- 1\right) \left( \mu'_{i,j}\right) ^2\\
						&=&	P\left(X_{i,j}^t=0\Big|X_{i,j}^{t-1}=0 \right)\left(  \frac{\alpha_{1,i,j}}{1+\alpha_{0,i,j} +\alpha_{1,i,j}} \right) ^{k-1} \frac{\alpha_{0,i,j}}{(1+\alpha_{0,i,j})^2}P\left( X_{i,j}^{t-k}=0 \right)^{-1}\mu'_{i,j}\\
						&=&\left(  \frac{\alpha_{1,i,j}}{1+\alpha_{0,i,j} +\alpha_{1,i,j}} \right) ^{k-1}\frac{\alpha_{0,i,j}(1+\alpha_{1,i,j})^2 }{(1+\alpha_{0,i,j})^2 (1+\alpha_{0,i,j} +\alpha_{1,i,j})^2}\\
						&=&\left(  \frac{\alpha_{1,i,j}}{1+\alpha_{0,i,j} +\alpha_{1,i,j}} \right) ^{k-1}	\gamma_{i,j}'(1).
					\end{eqnarray*}
					This proves (ii).
				\end{proof}

				\begin{lem}\label{var_bound}
					Let $\{\bX^{t}\} \sim P_{\btheta}$ hold. Under  condition (A1) we have,
					\begin{equation*}
						\sup_{1\leq i<j\leq p} \left\{\var\left( \sum_{t=1}^{n}X_{i,j}^tX_{i,j}^{t-1}\right) ,\quad \var\left( \sum_{t=1}^{n}X_{i,j}^t\right) \right\} =O(n).
					\end{equation*}
				\end{lem}
				\begin{proof}
					Let $Y_1, Y_2, \ldots $ be a sequence of Bernoulli random variables with $\E Y_i =\mu $,  $\var(Y_i)=\sigma^2$ for all $i=1, 2\dots $,  and assume that 
					$\cov\left( Y_i,Y_j\right)  \leq \sigma^2 \rho^{|i-j|}$ for some $0\le \rho <1$. We have
					\begin{eqnarray*}
						\var\left(\sum_{i=1}^{n} Y_i \right)  &=& \left( \sum_{i=1}^{n} \var\left( Y_i\right)   + \sum_{1\leq i\neq j\leq n}  \cov\left( Y_i,Y_j\right)    \right) \\
						&\leq&  {\sigma^2} \left( n +  2\rho(n-1) + 2\rho^2(n-2) +\cdots+2\rho^{n-1}\right)\\
						&\leq&  {2\sigma^2} \left( n +  \rho(n-1) + \rho^2(n-2) +\cdots+\rho^{n-1}\right) \\
						&	\le & \frac{2n\sigma^2}{  1-\rho }.
					\end{eqnarray*}
					Lemma \ref{var_bound} then follows directly from Proposition \ref{prop1}, Lemma \ref{thm1_lemma}, condition (A1), and the above inequality.
				\end{proof}
				
				{ 	
					\begin{lem}\label{Bers}
						Suppose $Z_{i}$, $i=1,\cdots,p $ are independent random variables with
						\begin{equation*}
							\E\left( Z_{i} \right)=0,\quad  \var\left( Z_{i} \right) \leq \sigma^2,
						\end{equation*}
						and $Z_{i} \leq b$ almost surely. We have, for any constant $c>0$, there exists a large enough constant $C>0$ such that, 	with probability greater than $1-(p)^{-c}$,
						\begin{equation*}
							\left| \sum_{i=1}^{p} Z_{i} \right| \le C[\sqrt{p\log(p)}\sigma + b\log(p)] .
						\end{equation*}
					\end{lem}
				}
				\begin{proof}
					This is a direct result of Bernstein's inequality \cite{lin2011probability}.
				\end{proof}
				
				\subsection{Proof of Proposition \ref{prop1} }

				\begin{proof}\mbox{}\\
					$(i)$ follows directly from Proposition 1 of \cite{jiang2020autoregressive}. For  (ii), we have, 	
					\begin{eqnarray*}
						\E \left( d_i^t\right) &=& \sum_{k=1,\:k \neq i}^{p} \E\left( X_{i,k}^t\right)= \sum_{k=1,\:k \neq i}^{p}\frac{e^{\beta_{i,0} +\beta_{k,0} }}{1+e^{\beta_{i,0} +\beta_{k,0}}},\\
						\var\left( d_i^t\right)  &=& \sum_{k=1,\:k \neq i}^{p} \var\left( X_{i,k}^t\right)=\sum_{k=1,\:k \neq i}^{p}\frac{e^{\beta_{i,0} +\beta_{k,0} }}{(1+e^{\beta_{i,0} +\beta_{k,0}})^2}, \\
						\cov\left( d_i^t,d_{i}^{s}\right)&=&\sum_{k=1,\:k\neq i}^{p} \cov\left( X_{i,k}^t,X_{i,k}^{s}\right)\\
						&=&\sum_{k=1,\:k \neq i}^{p}\left( \frac{e^{\beta_{i,1} +\beta_{k,1}} }{ 1+\sum_{r=0}^1e^{\beta_{i,r} +\beta_{k,r }}  }\right)^{ |t-s|}\frac{e^{\beta_{i,0} +\beta_{k,0} }}{(1+e^{\beta_{i,0} +\beta_{k,0}})^2}.
					\end{eqnarray*}
					Thus,
					\begin{eqnarray*}
						\rho^d_{i,j}(|t-s|) &\equiv&\cor(d_{i }^t, d_{j}^{s}) \\
						&=&\begin{cases}
							C_{i,\rho}\sum_{k=1,\:k \neq i}^{p}\left( \frac{e^{\beta_{i,1} +\beta_{k,1}} }{ 1+\sum_{r=0}^1e^{\beta_{i,r} +\beta_{k,r }}  }\right)^{ |t-s|}\frac{e^{\beta_{i,0} +\beta_{k,0} }}{(1+e^{\beta_{i,0} +\beta_{k,0}})^2}\quad & {\rm if} \; i=j,\\
							0 & {\rm if} \; i\neq j,
						\end{cases}
					\end{eqnarray*}
					where $C_{i,\rho}=\left( \sum_{k=1,\:k \neq i}^{p}\frac{e^{\beta_{i,0} +\beta_{k,0} }}{(1+e^{\beta_{i,0} +\beta_{k,0}})^2}\right)^{-1} $.
					
				\end{proof}

				\subsection{Proof of Lemma \ref{prop1} }
				
				\begin{proof}
					By  Proposition 3 of \cite{jiang2020autoregressive} and Lemma \ref{thm1_lemma}, we have that the process $\{\bX^{t},t=1,2,\ldots\}$ is $\alpha$-mixing with exponential decaying rate.	Since the mixing property is hereditary,   the processes   $ \{\left( 1-\bX^t\right) \circ\left( 1- \bX^{t-1}\right) ,t=1,2,\ldots \}$ and $ \{\bX^t \circ\bX^{t-1}  ,t=1,2,\ldots \}$ are also $\alpha$-mixing with exponential decaying rate. From  Theorem 1 in \cite{merlevede2009bernstein},  we obtain the following concentration inequalities: there exists a positive constant $C$ such that,
					\begin{eqnarray*}
						&&	\bP\left(\left|\sum_{t=1}^{n}\left\{ X_{i,j}^{t}- \E\left(X_{i,j}^{t} \right)  \right\}\right| > \epsilon\right)  \leq  \exp\left( \frac{-C \epsilon^2 }{n+\epsilon\log\left(n\right)\log\log\left(n\right) }\right),   \\
						&&	\bP\left(\left|\sum_{t=1}^{n}\left\{ X_{i,j}^{t}X_{i,j}^{t-1} - \E\left(X_{i,j}^{t}X_{i,j}^{t-1}  \right)\right\} \right| > \epsilon\right)  \leq  \exp\left( \frac{-C \epsilon^2 }{n+\epsilon\log\left(n\right)\log\log\left(n\right) }\right),   \\
						&&	\bP\left(\left|\sum_{t=1}^{n}\left\{ (1-X_{i,j}^t) (1-X_{i,j}^{t-1})-\E\left((1-X_{i,j}^t) (1-X_{i,j}^{t-1}) \right) \right\}\right|  > \epsilon\right)\\
						&\leq&  \exp\left( \frac{-C \epsilon^2 }{n+\epsilon\log\left(n\right)\log\log\left(n\right) }\right),
					\end{eqnarray*}
					hold for all $1\leq i< j\leq p$.
					For any positive constant $c>0$, by setting $\epsilon=c_1\sqrt{n\log(np)}+ c_1\log\left(n\right) \log\log\left(n\right)\log\left(np\right)$ with a big enough constant $c_{1}>0$, we have, with probability greater than $1-(np)^{-c}$,
					\begin{eqnarray*}
						&&	\left|\sum_{t=1}^{n}\left\{ X_{i,j}^{t}- \E\left(X_{i,j}^{t} \right) \right\}\right| \leq \epsilon,\\
						&&		\left|\sum_{t=1}^{n}\left\{ X_{i,j}^{t}X_{i,j}^{t-1} - \E\left(X_{i,j}^{t}X_{i,j}^{t-1}  \right)\right\} \right|\leq \epsilon,\\
						&&		\left|\sum_{t=1}^{n}\left\{ (1-X_{i,j}^t) (1-X_{i,j}^{t-1})-\E\left((1-X_{i,j}^t) (1-X_{i,j}^{t-1}) \right) \right\}\right| \leq  \epsilon,
					\end{eqnarray*}
					hold for all $1\leq i< j\leq p$.
				\end{proof}
				\subsection{Proof of Lemma \ref{bound_l2}}
				\begin{proof}
					Note that for $a,b,c,d\in [0,1]$, we have
					\begin{equation*}
						\left|ab-cd\right|\leq \left|ab-cb\right|+\left|cb-cd\right|= \left|a-c\right|b+\left|b-d\right|c\leq \left|a-c\right|+\left|b-d\right|.
					\end{equation*}
					{ With $-\kappa_{0}\leq \theta_{i,0} \leq \kappa_{0} $ and $-\kappa_{1}\leq \theta_{i,1} \leq \kappa_{1} $,  for all $1\leq i\leq p$ and $\kappa_{r}=\max\left( \kappa_{r}, \kappa_{r}\right) $ for $r=0,1$, }
					we then have,  there exist positive constants $C_{1}, C_{2}$ such that, for all $1\leq i\neq j\leq p$ and
					{ $\btheta \in \bB_{\infty}\left( \btheta^*, c_{r}e^{-4\kappa_{0}-4\kappa_{1}}\right)$ }
					where  $c_r>0$ is a small enough constant, 
					\begin{equation*}
						\E(\bV_2(\btheta)+\bV_1(\btheta))_{i,j}= \frac{1}{p}\frac{\alpha_{0,i,j}} {(1+\alpha_{0,i,j} +\alpha_{1,i,j})^2}\geq C_{1} \frac{ e^{-4\kappa_{1}- 2\kappa_{0} }}{p},
					\end{equation*}
					\begin{eqnarray}\label{v2ij}
						&& \quad\quad\E(-\bV_2(\btheta))_{i,j} \\
						&=& \frac{1}{p} \left( \frac{\alpha_{0,i,j}\alpha_{1,i,j}} {(1+\alpha_{0,i,j} +\alpha_{1,i,j})^2} - \E(b_{i,j}) \frac{\alpha_{0,i,j}\alpha_{1,i,j}} {(\alpha_{0,i,j} +\alpha_{1,i,j})^2}\right) \nonumber\\
						&=& \frac{1}{p} \frac{\alpha_{0,i,j}\alpha_{1,i,j}} {(\alpha_{0,i,j} +\alpha_{1,i,j})^2}\left( \frac{(\alpha_{0,i,j} +\alpha_{1,i,j})^2} {(1+\alpha_{0,i,j} +\alpha_{1,i,j})^2} - \E(b_{i,j}) \right) \nonumber\\
						&=&\frac{1}{p}\frac{\alpha_{0,i,j}\alpha_{1,i,j}} {(\alpha_{0,i,j} +\alpha_{1,i,j})^2}\Bigg\{ \left(\frac{(\alpha_{0,i,j} +\alpha_{1,i,j})^2} {(1+\alpha_{0,i,j} +\alpha_{1,i,j})^2}- \frac{\alpha_{0,i,j}(\alpha_{0,i,j}+\alpha_{1,i,j})}{(1+\alpha_{0,i,j})(1+\alpha_{0,i,j} +\alpha_{1,i,j})}\right)\nonumber\\
						&&-	 \left( \frac{\alpha^*_{0,i,j}(\alpha^*_{0,i,j}+ \alpha^*_{1,i,j})} {(1+\alpha^*_{0,i,j})(1+\alpha^*_{0,i,j} +\alpha^*_{1,i,j})} -\frac{\alpha_{0,i,j} (\alpha_{0,i,j} +\alpha_{1,i,j})}{(1+\alpha_{0,i,j})(1+\alpha_{0,i,j} +\alpha_{1,i,j})}\right) \Bigg\}\nonumber\\
						&\geq&\frac{1}{p}\frac{\alpha_{0,i,j}\alpha_{1,i,j}} {(\alpha_{0,i,j} +\alpha_{1,i,j})^2}\Bigg\{\frac{\alpha_{1,i,j}(\alpha_{0,i,j}+\alpha_{1,i,j})}{(1+\alpha_{0,i,j})(1+\alpha_{0,i,j}+\alpha_{1,i,j})^2 }\nonumber\\
						&&- \left( \left| \frac{\alpha^*_{0,i,j}} {1+\alpha^*_{0,i,j}} -\frac{\alpha_{0,i,j} }{1+\alpha_{0,i,j}}\right|+\left| \frac{\alpha^*_{0,i,j}+ \alpha^*_{1,i,j}} {1+\alpha^*_{0,i,j} +\alpha^*_{1,i,j}} -\frac{\alpha_{0,i,j} +\alpha_{1,i,j}}{1+\alpha_{0,i,j} +\alpha_{1,i,j}}\right|\right)\Bigg\}\nonumber \\
						&\geq&\frac{1}{p}\frac{\alpha_{0,i,j}\alpha_{1,i,j}} {(\alpha_{0,i,j} +\alpha_{1,i,j})^2}\Bigg\{\frac{\alpha_{1,i,j}(\alpha_{0,i,j}+\alpha_{1,i,j})}{(1+\alpha_{0,i,j})(1+\alpha_{0,i,j}+\alpha_{1,i,j})^2 }\nonumber\\
						&&-  \left( \left|\frac{\alpha_{0,i,j}\left( \frac{\alpha^*_{0,i,j}}{\alpha_{0,i,j}} -1 \right) }{\big( 1+\alpha_{0,i,j}\big) \big(1+\alpha^*_{0,i,j} \big)  }\right|+\left| \frac{\alpha_{0,i,j}\left( \frac{\alpha^*_{0,i,j}}{\alpha_{0,i,j}} -1 \right)+\alpha_{1,i,j}\left( \frac{\alpha^*_{1,i,j}}{\alpha_{1,i,j}} -1 \right)} {\big(1+\alpha^*_{0,i,j} +\alpha^*_{1,i,j}\big)\big(1+\alpha_{0,i,j} +\alpha_{1,i,j}\big)}\right|\right)\nonumber\Bigg\} \\
						&\geq&\frac{1}{p}\frac{\alpha_{0,i,j}\alpha_{1,i,j}} {(\alpha_{0,i,j} +\alpha_{1,i,j})^2}\Bigg\{\frac{\alpha_{1,i,j}(\alpha_{0,i,j}+\alpha_{1,i,j})}{(1+\alpha_{0,i,j})(1+\alpha_{0,i,j}+\alpha_{1,i,j})^2} -\frac{\left( 2\left| \frac{\alpha^*_{0,i,j}} {\alpha_{0,i,j}} -1\right|+\left| \frac{\alpha^*_{1,i,j}} {\alpha_{1,i,j}} -1\right|\right)}{1+\alpha^*_{0,i,j} +\alpha^*_{1,i,j}}\Bigg\}\nonumber\\
						&\geq& C_{1}  \frac{e^{-6\kappa_{0}-4\kappa_{1}}}{p}\nonumber,
					\end{eqnarray}
					and
					\begin{eqnarray}\label{v23ij}
						&& \E(\bV_2(\btheta)+\bV_3(\btheta))_{i,j}\\
						&= & \frac{1}{p} \left( \frac{\alpha_{1,i,j}} {(1+\alpha_{0,i,j}+\alpha_{1,i,j})^2} - \E(d_{i,j}) \frac{\alpha_{1,i,j}}{(1+\alpha_{1,i,j})^2} \right)\nonumber \\
						&= & \frac{1}{p} \frac{\alpha_{1,i,j}}{(1+\alpha_{1,i,j})^2}\left( \frac{(1+\alpha_{1,i,j})^2} {(1+\alpha_{0,i,j}+\alpha_{1,i,j})^2} - \E(d_{i,j})  \right)\nonumber \\
						&	= &\frac{1}{p} \frac{\alpha_{1,i,j}}{(1+\alpha_{1,i,j})^2}\Bigg\{ \frac{(1+\alpha_{1,i,j})^2} {(1+\alpha_{0,i,j}+\alpha_{1,i,j})^2}  - \frac{1+\alpha_{1,i,j}} {(1+\alpha_{0,i,j})(1+\alpha_{0,i,j}+\alpha_{1,i,j})}\nonumber\\
						&	&- \left(   \frac{1+\alpha^*_{1,i,j}} {(1+\alpha^*_{0,i,j}) (1+\alpha^*_{0,i,j}+\alpha^*_{1,i,j})} -\frac{1+\alpha_{1,i,j}} {(1+\alpha_{0,i,j}) (1+\alpha_{0,i,j}+\alpha_{1,i,j})}\right)\Bigg\} \nonumber\\
						&=&\frac{1}{p} \frac{\alpha_{1,i,j}}{(1+\alpha_{1,i,j})^2}\Bigg\{\frac{\alpha_{0,i,j} \alpha_{1,i,j}(1+\alpha_{1,i,j})}{(1+\alpha_{0,i,j})(1+\alpha_{0,i,j}+\alpha_{1,i,j})^2 }\nonumber\\
						&	&-	 \left( \left| \frac{1} {1+\alpha^*_{0,i,j}} -\frac{1}{1+\alpha_{0,i,j}}\right| +\left| \frac{1+ \alpha^*_{1,i,j}} {1+\alpha^*_{0,i,j} +\alpha^*_{1,i,j}} -\frac{1+\alpha_{1,i,j}} {1+\alpha_{0,i,j} +\alpha_{1,i,j}}\right|\right)\Bigg\} \nonumber\\
						&	=&\frac{1}{p} \frac{\alpha_{1,i,j}}{(1+\alpha_{1,i,j})^2}\Bigg\{\frac{\alpha_{0,i,j} \alpha_{1,i,j}(1+\alpha_{1,i,j})}{(1+\alpha_{0,i,j})(1+\alpha_{0,i,j}+\alpha_{1,i,j})^2 }\nonumber\\
						&&-  \left( \left|\frac{\alpha_{0,i,j}-\alpha^*_{0,i,j} }{\big( 1+\alpha_{0,i,j}\big) \big(1+\alpha^*_{0,i,j} \big)  }\right|+\left| \frac{\alpha_{0,i,j}-\alpha^*_{0,i,j} +\alpha_{0,i,j}\alpha_{1,i,j}^*-\alpha^*_{0,i,j}\alpha_{1,i,j}} {\big(1+\alpha^*_{0,i,j} +\alpha^*_{1,i,j}\big)\big(1+\alpha_{0,i,j} +\alpha_{1,i,j}\big)}\right|\right)\Bigg\}\nonumber \\
						&=&\frac{1}{p} \frac{\alpha_{1,i,j}}{(1+\alpha_{1,i,j})^2}\Bigg\{\frac{\alpha_{0,i,j} \alpha_{1,i,j}(1+\alpha_{1,i,j})}{(1+\alpha_{0,i,j})(1+\alpha_{0,i,j}+\alpha_{1,i,j})^2 }-  \left|\frac{\alpha_{0,i,j}-\alpha^*_{0,i,j} }{\big( 1+\alpha_{0,i,j}\big) \big(1+\alpha^*_{0,i,j} \big)  }\right|\nonumber\\
						&	&-\left| \frac{\alpha_{0,i,j}-\alpha^*_{0,i,j} +\alpha_{0,i,j}\alpha_{1,i,j}^*- \alpha_{0,i,j}^*\alpha_{1,i,j}^*+ \alpha_{0,i,j}^*\alpha_{1,i,j}^* -\alpha^*_{0,i,j}\alpha_{1,i,j}} {\big(1+\alpha^*_{0,i,j} +\alpha^*_{1,i,j}\big)\big(1+\alpha_{0,i,j} +\alpha_{1,i,j}\big)}\right|\nonumber \Bigg\}\\
						&=&\frac{1}{p} \frac{\alpha_{1,i,j}}{(1+\alpha_{1,i,j})^2}\Bigg\{\frac{\alpha_{0,i,j} \alpha_{1,i,j}(1+\alpha_{1,i,j})}{(1+\alpha_{0,i,j})(1+\alpha_{0,i,j}+\alpha_{1,i,j})^2 }-  \left|\frac{\alpha_{0,i,j}\left( \frac{\alpha^*_{0,i,j}} {\alpha_{0,i,j}} -1 \right) }{\big( 1+\alpha_{0,i,j}\big) \big(1+\alpha^*_{0,i,j} \big)  }\right|\nonumber\\
						&&-\left| \frac{\alpha_{0,i,j}\left( \frac{\alpha^*_{0,i,j}} {\alpha_{0,i,j}} -1 \right)-\alpha_{0,i,j}\alpha^*_{1,i,j}\left( \frac{\alpha^*_{0,i,j}} {\alpha_{0,i,j}} -1 \right)+ \alpha^*_{0,i,j}\alpha_{1,i,j}\left( \frac{\alpha^*_{1,i,j}} {\alpha_{1,i,j}} -1 \right)} {\big(1+\alpha^*_{0,i,j} +\alpha^*_{1,i,j}\big)\big(1+\alpha_{0,i,j} +\alpha_{1,i,j}\big)}\right|\nonumber\Bigg\} \\
						&\geq& C_{1}\frac{1}{p} \frac{\alpha_{1,i,j}}{(1+\alpha_{1,i,j})^2} \Bigg\{\frac{\alpha_{0,i,j} \alpha_{1,i,j}(1+\alpha_{1,i,j})}{(1+\alpha_{0,i,j}) (1+\alpha_{0,i,j}+\alpha_{1,i,j})^2 } \nonumber\\
						&&-\left( 2\left| \frac{\alpha^*_{0,i,j}} {\alpha_{0,i,j}} -1\right|+\left| \frac{\alpha^*_{1,i,j}} {\alpha_{1,i,j}} -1\right|\right) \Bigg\}\nonumber\\
						&\geq& C_{1}  \frac{e^{-4\kappa_{0}-4\kappa_{1}}}{p}\nonumber.
					\end{eqnarray}
					Notice that the elements in $-\E \bV_2(\btheta),\E (\bV_2(\btheta)+ \bV_3(\btheta))$ and $\E (\bV_2(\btheta)+\bV_1(\btheta))$ are all positive.
					Denote  $\bz=(\bz_1^\top,\bz_2^\top)^\top$ with $\bz_1=(z_{1,1},\ldots, z_{1,p})^\top \in \mathbb{R}^p$ and $\bz_2=(z_{2,1},\ldots, z_{2,p})^\top\in \mathbb{R}^p$. Then there exists a constant $C>0$ such that, 
					\begin{eqnarray*}
						&&\left\|\E(\bV(\btheta))\right\|_{2}  \\
						& \geq&  \inf_{\|\bz\|_{2}=1}\Bigg( \sum_{1\leq i<j\leq p }\Big(\E\left(  \bV_{1}(\btheta) + \bV_{2}(\btheta) \right)_{i,j} (z_{1,i}+z_{1,j})^2\\
						&&+\E\left(\bV_{3}(\btheta) + \bV_{2}(\btheta) \right)_{i,j} (z_{2,i}+z_{2,j})^2-\E\left(  \bV_{2}(\btheta) \right)_{i,j} \left( z_{1,i}+z_{1,j}-z_{2,i}-z_{2,j}\right) ^2\Big)\Bigg)\\
						&\geq &\inf_{\|\bz\|_{2}=1}\Bigg( \sum_{1\leq i<j\leq p }\E\left(  \bV_{1} (\btheta)+ \bV_{2}(\btheta) \right)_{i,j}  (z_{1,i}+z_{1,j})^2\\
						&&+\E\left(\bV_{3}(\btheta) + \bV_{2}(\btheta) \right)_{i,j}  (z_{2,i}+z_{2,j})^2\Bigg) \\
						&\geq&  C_{2}  \frac{e^{-4\kappa_{0}-4\kappa_{1}}}{p}	\inf_{\|\bz\|_{2}=1}\sum_{1\leq i<j\leq p }\left(  (z_{1,i}+z_{1,j})^2+ (z_{2,i}+z_{2,j})^2 \right)\\
						&\geq&  C e^{-4\kappa_{0}-4\kappa_{1}}.
					\end{eqnarray*}
					Here in the last step we have used the fact that for any $\ba=(a_1,\ldots, a_p)^\top\in  \mathbb{R}^p$,
					$\sum_{1\leq i<j\leq p }  (a_{i}+a_{j})^2= \ba^\top \bC \ba$ where
					$\bC= (p-2){\bI}_{p}+{\bf 1}_p {\bf 1}_p^\top$, and the fact that the eigenvalues of 	$\bC$ is greater or equal to $p-2$. 
				\end{proof}

				\subsection{Proof of Lemma \ref{BersM}}
				
				\begin{proof}
					Define a series of matrices $\{\bY_{i,j}\}$ $(1\leq i< j \leq p)$. For $\bY_{i,j}$, the $(i,j)$, $(j,i)$, $(i,i)$, $(j,j)$ elements are $Z_{i,j}$ while other elements are set to be zero. Then all the $\bY_{i,j}$ matrices are independent and
					\begin{eqnarray*}
						\sum_{1\leq i<j\leq p } \bY_{i,j} = \bZ.
					\end{eqnarray*}
					Since $\bY_{i,j}$ are symmetric and centered random matrices, we have $\var\left( \bY_{i,j}\right)=\E\left( \bY_{i,j}\bY_{i,j}\right) $. Further, by the definition of $\bY_{i,j}$, we know that the $(i,j)$th, $(j,i)$th, $(i,i)$th and $(j,j)$th elements of $\var(\bY_{i,j})$ are all equal  to $2\var\left( Z_{i,j} \right)$, while all other elements are zero. Consequently, 
					\begin{eqnarray*}
						\|\bY_{i,j}\|_2 &&\leq b\sup_{\|a\|_2=1} \left(  (a_i + a_j )^2 \right)\leq2 b,\\
						\left\|\sum_{1\leq i<j\leq p }  \var(\bY_{i,j})\right\|_{2} &=& \sup_{\|a\|_2=1} \left( \sum_{1\leq i<j\leq p } 2\var\left( Z_{i,j} \right)(a_i + a_j )^2 \right) \\
						& \leq& \max_{i,j} \left(2 \var\Big(Z_{i,j}\Big)\right) \sup_{\|a\|_2=1} \left( \sum_{1\leq i<j\leq p } (a_i + a_j )^2 \right)\\
						&\leq& 4\sigma^2(p-1).
					\end{eqnarray*}
					Using the Matrix Bernstein inequality (c.f. Theorem 6.17 of \cite{wainwright2019high}), we have
					\begin{equation*}
						P\left( \left\| \bZ \right\|_{2} > \epsilon\right) = P\left( \left\|\sum_{1\leq i<j\leq p } \bY_{i,j}\right\|_{2}>\epsilon \right) \leq 2p \ \exp\left( -\frac{\epsilon^2}{ 4\sigma^2(p-1) + 4b\epsilon } \right).
					\end{equation*}
				\end{proof}
				
				\subsection{Proof of Theorem \ref{thm1}}
				
				\begin{proof}
					Note that for any $\btheta$ and $i\neq j$, we have
					\begin{eqnarray*}
						\left( \bV_2(\btheta)+\bV_1(\btheta)\right) _{i,j} &=&\frac{1}{p}\frac{\alpha_{0,i,j}} {(1+\alpha_{0,i,j} +\alpha_{1,i,j})^2} >0,\\
						\left( \bV_2(\btheta)+\bV_1(\btheta)\right) _{i,i}&=& \sum_{k=1,\:k\neq i}^{p} \left( \bV_2(\btheta)+\bV_1(\btheta)\right) _{i,k}.
					\end{eqnarray*}
					Therefore, \ $\bV_2(\btheta)+\bV_1(\btheta)$ is always positive definite.      	
					Next we prove that all   $\bU(\btheta)\in\{-\bV_2(\btheta) ,\bV_2(\btheta)+\bV_3(\btheta)\}$
					are positive definite (in probability) by showing that with probability tending to 1,
					\begin{equation*}
						\inf_{\|\ba\|_2=1} \left( \ba^\top (\E(\bU(\btheta))) \ba\right) > \sup_{\|\ba\|_2=1} \left( \ba^\top (\bU(\btheta)-\E(\bU(\btheta))) \ba\right),
					\end{equation*}
					holds  uniformly for all  $\btheta\in \bB_{\infty}\left(\btheta^*,r \right)$.
					Note that
					\begin{eqnarray*}
						&&\left\| \bU(\btheta)-\E(\bU(\btheta)) \right\|_{2} 
						\\ &\leq&\left\|\bU(\btheta^*)- \E(\bU(\btheta^*))\right\|_{2}+  \left\|\bU(\btheta)- \bU(\btheta^*) +\E(\bU(\btheta^*))- \E(\bU(\btheta)) \right\|_{2}.
					\end{eqnarray*}
					We  consider $\left\|\bU(\btheta^*)- \E(\bU(\btheta^*))\right\|_{2}$ first. 
					By setting $\epsilon = c_1 \big( \sqrt{2\sigma^2(p-1)\log(np)}+ b\log(np)\big)$ with some big enough constant $c_1>0$ in Lemma  \ref{BersM}, we have
					\begin{eqnarray*}
						&&P\left( \left\| \bZ \right\|_{2} > \epsilon \right)\\
						&\leq&2p\exp\left( -c_1^2\frac{ 2\sigma^2(p-1)\log(np)  +2b\sqrt{2\sigma^2(p-1)}\log^{3/2}(np) +b^2\log^2(np)}{2\sigma^2(p-1)  +4c_1 b\sqrt{2\sigma^2(p-1)\log(np)}+ 4c_1 b^2\log(np)} \right)\\
						&\leq& 2p\exp\left( -c_1\log(np)/4 \right)\\
						&= &2p\left( np\right) ^{-c_1/4}.
					\end{eqnarray*}
					As $np\to \infty$,  we have	with probability greater than $1-2p\left( np\right) ^{-c_1/4}$,
					\begin{equation}\label{Lem3}
						\left\| \bZ \right\|_{2}\leq  c_{1}\left( \sigma\sqrt{2p\log(np)} + b\log(np) \right).
					\end{equation}
					By Lemma \ref{mixing} and Lemma \ref{var_bound}, we have, uniformly for all $\btheta $  and $1\le i\neq j\le p$, there exist positive constants $C_{1}$, $c_{2}$  such that,  with probability greater than $1-(np)^{-c_2}$,
					\begin{eqnarray*}
						|\bU(\btheta)_{i,j}- \E(\bU(\btheta))_{i,j}|&\leq&C_{1}\left( \sqrt{\frac{\log(np)}{np^2}}+ \frac{\log\left(n\right) \log\log\left(n\right)\log\left(np\right)}{np}\right)  ;\\
						\var\left(\bU(\btheta)_{i,j} \right) &\leq&  \frac{	\sup_{i\neq j} \left\{\var\left( \sum_{t=1}^{n}X_{i,j}^tX_{i,j}^{t-1}\right) , \var\left( \sum_{t=1}^{n}X_{i,j}^t\right) \right\} }{n^2p^2} \\
						&\leq&\frac{C_{1}}{np^2} .
					\end{eqnarray*}
					Consequently, from \eqref{Lem3} we have 
					\begin{eqnarray*}
						\left\|\bU(\btheta^*)- \E(\bU(\btheta^*))\right\|_{2}
						&=& O_{p}\left(\sqrt{\frac{\log(np)}{np}}+  \sqrt{\frac{\log^3(np)}{np^2}}+ \frac{\log\left(n\right) \log\log\left(n\right)\log^2\left(np\right)}{np} \right) .
					\end{eqnarray*}
					
					Next we derive uniform upper bounds for $\left\|\bU(\btheta)- \bU(\btheta^*) +\E(\bU(\btheta^*))- \E(\bU(\btheta)) \right\|_{2}$.		
					When $\bU(\btheta)=-\bV_2(\btheta) $, by Lemma \ref{mixing}, we have, there exist positive constants $C_{2}$ and $c_{3}$  such that with probability greater than $1-(np)^{-c_3}$,
					\begin{eqnarray*}
						&&\left\|\bU(\btheta)- \bU(\btheta^*) +\E(\bU(\btheta^*))- \E(\bU(\btheta)) \right\|_{2}\\
						&=& \sup_{\|\ba\|_2=1}\sum_{1\leq i,j\leq p} \left[ \bU(\btheta)- \bU(\btheta^*) +\E(\bU(\btheta^*))- \E(\bU(\btheta)) \right]_{i,j}  a_ia_j  \\
						&=& \sup_{\|\ba\|_2=1}\sum_{1\leq i<j\leq p}  \left[\bU(\btheta)- \bU(\btheta^*) +\E(\bU(\btheta^*))- \E(\bU(\btheta)) \right]_{i,j}  \left(a_i+a_j \right)^2 \\
						&\leq&\max_{1\le i< j\le p}\left|\left[ \bU(\btheta)- \bU(\btheta^*) +\E(\bU(\btheta^*))- \E(\bU(\btheta)) \right]_{i,j}\right| \sup_{\|\ba\|_2=1}\sum_{1\leq i<j\leq p} \left(a_i+a_j \right)^2\\
						&\leq& \max_{1\le i< j\le p}\Bigg|\bV_2(\btheta)_{i,j} - \bV_2(\btheta)^*_{i,j}+ \E\left( \bV_2(\btheta)^*_{i,j}\right) -\E\left( \bV_2(\btheta)_{i,j}\right)  \Bigg|2\left(p-1 \right) \\
						&= &\frac{2\left(p-1 \right)}{np}\max_{1\le i< j\le p}\Bigg|\left( b_{i,j}-\E(b_{i,j})\right)  \left( \frac{\alpha_{0,i,j}\alpha_{1,i,j}} {(\alpha_{0,i,j} +\alpha_{1,i,j})^2}-\frac{\alpha_{0,i,j}^*\alpha_{1,i,j}^*} {(\alpha_{0,i,j}^* +\alpha_{1,i,j}^*)^2}\right)  \Bigg|  \\
						&\leq&\max_{1\le i< j\le p}\frac{2\left| b_{i,j}-\E(b_{i,j})\right| }{n}\max_{1\le i< j\le p}\Bigg|  \frac{\left( \alpha_{0,i,j}^*\alpha_{1,i,j}-\alpha_{0,i,j} \alpha_{1,i,j}^* \right)\left( \alpha_{0,i,j}^*\alpha_{0,i,j}-\alpha_{1,i,j}\alpha_{1,i,j}^*\right)  } {(\alpha_{0,i,j} +\alpha_{1,i,j})^2(\alpha_{0,i,j}^* +\alpha_{1,i,j}^*)^2}\Bigg|\\
						&=&\max_{1\le i< j\le p}\frac{2\left| b_{i,j}-\E(b_{i,j})\right| }{n}\max_{i\neq j}\Bigg| \left( \frac{\alpha_{0,i,j}^*}{\alpha_{0,i,j}}-\frac{\alpha_{1,i,j}^*}{\alpha_{1,i,j}}  \right) \frac{\alpha_{0,i,j}\alpha_{1,i,j}\left( \alpha_{0,i,j}^*\alpha_{0,i,j}-\alpha_{1,i,j}\alpha_{1,i,j}^*\right)  } {(\alpha_{0,i,j} +\alpha_{1,i,j})^2(\alpha_{0,i,j}^* +\alpha_{1,i,j}^*)^2}\Bigg|\\
						&\leq&C_{2}\min\left\{1,\sqrt{\frac{\log(np)}{ {n}}}
						+\frac{\log\left(n\right)\log\log\left(n\right)\log\left(np\right)}{n}
						\right\}\max_{1\le i< j\le p} \left| \frac{\alpha_{0,i,j}^*}{\alpha_{0,i,j}}-\frac{\alpha_{1,i,j}^*}{\alpha_{1,i,j}} \right|,
					\end{eqnarray*}	
					holds	uniformly for all $\btheta$. 	
					Similarly,	when $\bU(\btheta)=\bV_2(\btheta) + \bV_3(\btheta)$, by Lemma \ref{mixing},  we have there exist positive constants $C_{3}$ and $c_{4}$ such that,   with probability greater than $1-(np)^{-c_4}$,
					\begin{eqnarray*}
						&&\left\|\bU(\btheta)- \bU(\btheta^*) +\E(\bU(\btheta^*))- \E(\bU(\btheta)) \right\|_{2}\\
						&\leq&\max_{1\le i< j\le p}\left|\left[ \bU(\btheta)- \bU(\btheta^*) +\E(\bU(\btheta^*))- \E(\bU(\btheta)) \right]_{i,j}\right| \sup_{\|\ba\|_2=1}\sum_{1\leq i<j\leq p} \left(a_i+a_j \right)^2\\
						&\leq&\max_{1\le i< j\le p}\Bigg|\Bigg\{\bV_2(\btheta)_{i,j} + \bV_3(\btheta)_{i,j}- \bV_2(\btheta)^*_{i,j} - \bV_3(\btheta)^*_{i,j}+ \E\left( \bV_2(\btheta)^*_{i,j}\right)  + \E\left( \bV_3(\btheta)^*_{i,j}\right)\\ &&-\E\left( \bV_2(\btheta)_{i,j}\right)  - \E\left( \bV_3(\btheta)_{i,j}\right)\Bigg\}  \Bigg|2\left(p-1 \right) \\
						&= &\frac{2\left(p-1 \right)}{np}\max_{1\le i< j\le p}\Bigg|\left( b_{i,j}-\E(b_{i,j})\right)  \left( \frac{\alpha_{0,i,j}\alpha_{1,i,j}} {(\alpha_{0,i,j} +\alpha_{1,i,j})^2}-\frac{\alpha_{0,i,j}^*\alpha_{1,i,j}^*} {(\alpha_{0,i,j}^* +\alpha_{1,i,j}^*)^2}\right) \Bigg| \\
						&&+\frac{2\left(p-1 \right)}{np}\max_{1\le i< j\le p}\Bigg|\left( d_{i,j}-\E(d_{i,j})\right) \left( \frac{\alpha_{1,i,j}}{(1+\alpha_{1,i,j})^2}- \frac{\alpha_{1,i,j}^*}{(1+\alpha_{1,i,j}^*)^2} \right) \Bigg|  \\
						&\leq &2\max_{1\le i< j\le p}\frac{\left|b_{i,j}-\E(b_{i,j})\right| }{n}\max_{1\le i< j\le p} \left| \left( \frac{\alpha_{1,i,j}^*} {\alpha_{1,i,j}}-\frac{\alpha_{0,i,j}^*}{\alpha_{0,i,j}} \right)\frac{ \alpha_{0,i,j}\alpha_{1,i,j} \left( \alpha_{0,i,j}\alpha_{0,i,j}^*-\alpha_{1,i,j}\alpha_{1,i,j}^* \right) }{(\alpha_{0,i,j} +\alpha_{1,i,j})^2(\alpha_{0,i,j}^* +\alpha_{1,i,j}^*)^2} \right|\\
						&&+ 2\max_{1\le i< j\le p}\frac{\left| d_{i,j}-\E(d_{i,j})\right| }{n}\max_{1\le i< j\le p}\Bigg|\left( 1-  \frac{\alpha_{1,i,j}^*} {\alpha_{1,i,j}} \right) \frac{\alpha_{1,i,j}\left(1-  \alpha_{1,i,j}\alpha_{1,i,j}^*\right)  }{(1+\alpha_{1,i,j})^2 (1+\alpha_{1,i,j}^*)^2} \Bigg|\\
						&\leq &C_{3}\min\left\{1,\sqrt{\frac{\log(np)}{ {n}}}
						+\frac{\log\left(n\right)\log\log\left(n\right)\log\left(np\right)}{n}
						\right\}\\
						&&\times \left(  \max_{1\le i< j\le p}\left| \frac{\alpha_{0,i,j}^*}{\alpha_{0,i,j}}-\frac{\alpha_{1,i,j}^*} {\alpha_{1,i,j}} \right| +			\max_{1\le i< j\le p}\left| 1-  \frac{\alpha_{1,i,j}^*}{\alpha_{1,i,j}} \right|\right),
					\end{eqnarray*}
					holds uniformly for all $\btheta$.
					Consequently, there exist positive constants $C_{4}$ and $c_{5}$ such that,  with probability greater than $1-(np)^{-c_5}$,
					\begin{eqnarray*}
						&&\sup_{\btheta\in \bB_{\infty}\left(\btheta^*,r \right)}\left\|\bU(\btheta)- \bU(\btheta^*) +\E(\bU(\btheta^*))- \E(\bU(\btheta)) \right\|_{2}\\
						&\leq& \max\left\{ C_{3},C_{2}\right\} \min\left\{1,\sqrt{\frac{\log(np)}{ {n}}}
						+\frac{\log\left(n\right)\log\log\left(n\right)\log\left(np\right)}{n}
						\right\}\\
						&&\times\sup_{\btheta\in \bB_{\infty}\left(\btheta^*,r \right)}\left( \max_{1\le i< j\le p}\left| \frac{\alpha_{0,i,j}^*}{\alpha_{0,i,j}}-\frac{\alpha_{1,i,j}^*}{\alpha_{1,i,j}} \right| +			\max_{1\le i< j\le p}\left| 1-  \frac{\alpha_{1,i,j}^*}{\alpha_{1,i,j}} \right|\right) \\
						&\leq&  \max\left\{ C_{3},C_{2}\right\}\min\left\{1,\sqrt{\frac{\log(np)}{ {n}}}+\frac{\log\left(n\right)\log\log\left(n\right) \log\left(np\right)}{n}
						\right\}\\
						&&\times \sup_{\btheta\in \bB_{\infty}\left(\btheta^*,r \right)}\left( \max_{1\le i< j\le p}\left| \frac{\alpha_{0,i,j}^*} {\alpha_{0,i,j}} -1\right| +2\max_{1\le i< j\le p} \left| \frac{\alpha_{1,i,j}^*} {\alpha_{1,i,j}}-1 \right| \right) \\
						&\leq& C_{4} \min\left\{1,\sqrt{\frac{\log(np)}{ {n}}}
						+\frac{\log\left(n\right)\log\log\left(n\right)\log\left(np\right)}{n}
						\right\} r.
					\end{eqnarray*}
					On the other hand, 	by inequalities \eqref{v2ij} and \eqref{v23ij} we have, there exists a  positive constant $C_{5}$ such that   
					\begin{eqnarray*}
						\inf_{\|\ba\|_2=1} \left( \ba^\top \E(\bU(\btheta)) \ba\right)&=&\inf_{\|\ba\|_2=1} \sum_{1\leq i<j\leq p}\E(\bU(\btheta))_{i,j} \left(a_{i}+a_{j} \right)^2 \\
						&\geq& C_5  e^{-4\kappa_{0}-4\kappa_{1}}.
					\end{eqnarray*}
					holds  uniformly for any $\btheta \in \bB_{\infty}\left( \btheta^*, c_{r}e^{-4\kappa_{0}-4\kappa_{1}}\right)$   where  $c_r>0$ is a small enough constant.
					Thus, as $np\to \infty$, 
					{ 
						$\kappa_{0}+\kappa_{1} \le c\log\frac{np}{\log (np)}$ for some small enough constant $c >0$
					} and  	$r =c_r e^{-4\kappa_{0}-4\kappa_{1}}$  for some small enough constant $0<c_r<C_5/(2C_4)$, we have, there exist big enough positive constants $C_{6},c_{6}$ such that  with probability greater than $1-(np)^{-c_6}$,
					\begin{eqnarray*}
						&&\left\| \bU(\btheta)-\E(\bU(\btheta)) \right\|_{2}\\
						&\leq& \left\|\bU(\btheta^*)- \E(\bU(\btheta^*))\right\|_{2}+  \left\|\bU(\btheta)- \bU(\btheta^*) +\E(\bU(\btheta^*))- \E(\bU(\btheta)) \right\|_{2}\\
						&\leq&C_{6} \left( \sqrt{\frac{\log(np)}{np}}+  \sqrt{\frac{\log^3(np)}{np^2}}+ \frac{\log\left(n\right) \log\log\left(n\right)\log^2\left(np\right)}{np}\right)  +c_rC_{4}e^{-4\kappa_{0}-4\kappa_{1}}\\
						&\leq&2c_rC_{4}e^{-4\kappa_{0}-4\kappa_{1}}\\
						&<& C_5 e^{-4\kappa_{0}-4\kappa_{1}}\\
						&\leq&\inf_{\|\ba\|_2=1} \left( \ba^\top (\E(\bU(\btheta))) \ba\right),
					\end{eqnarray*}
					holds uniformly for all $\btheta\in \bB_\infty(\btheta^*, c_{r}e^{-4\kappa_{0}-4\kappa_{1}})$.
					The statements in Theorem \ref{thm1} can then be concluded.
				\end{proof}
				\subsection{Proof of Theorem \ref{thm2}}
				\begin{proof}
					Recall that 
					\begin{multline*}
						l(\btheta) :=  \frac{1}{p}\sum_{1\leq i<j\leq p  }\log \Big({1+ e^{\beta_{i,0} +\beta_{j,0 } }+e^{\beta_{i,1} +\beta_{j,1 } }}\Big)  -\frac{1}{np} { \sum_{1\leq i<j\leq p  }   \Bigg\{   \left( \beta_{i,0} +\beta_{j,0 }\right)    \sum_{t=1}^nX_{i,j}^t}  +   \nonumber\\
						\log\left(  1+e^{\beta_{i,1} +\beta_{j,1} }\right)\sum_{t=1}^n\left( 1-X_{i,j}^t\right) \left( 1-X_{i,j}^{t-1}\right) 
						+  \log\big( 1+e^{\beta_{i,1} +\beta_{j,1}-\beta_{i,0} -\beta_{j,0 } }\big)\sum_{t=1}^n X_{i,j}^t X_{i,j}^{t-1} \Bigg\},
					\end{multline*}
					and write $l_{E}(\btheta)= \E l(\btheta)$; that is,
					\begin{align*}
						l_{E}(\btheta) := & \frac{1}{p}\sum_{1\leq i<j\leq p  }\log \Big({1+ e^{\beta_{i,0} +\beta_{j,0 } }+e^{\beta_{i,1} +\beta_{j,1 } }}\Big)  -\frac{1}{np} \sum_{1\leq i<j\leq p  }   \Bigg\{   \left( \beta_{i,0} +\beta_{j,0 }\right)  \sum_{t=1}^n\E \left(X_{i,j}^t\right)      \nonumber\\
						&+\log\left(  1+e^{\beta_{i,1} +\beta_{j,1} }\right)\sum_{t=1}^n\E\left( 1-X_{i,j}^t\right) \left( 1-X_{i,j}^{t-1}\right)\\
						&+  \log\big( 1+e^{\beta_{i,1} +\beta_{j,1}-\beta_{i,0} -\beta_{j,0 } }\big)\sum_{t=1}^n\E\left(  X_{i,j}^t X_{i,j}^{t-1}\right)  \Bigg\},
					\end{align*}
					Define $\bD_{n,p}:=\left\{ \check{\btheta}:l\left( \check{\btheta}\right) -l(\btheta^*)\leq 0\ \text{and} \  \|\check{\btheta}-\btheta^*\|_{\infty}\leq  c_{r}e^{-4\kappa_{0}-4\kappa_{1}}\right\}$. 
					Note that when $c_r$ is small enough, $c_{r}e^{-4\kappa_{0}-4\kappa_{1}} <\alpha_{0}$.
					By Corollary \ref{Cor1}, there exist big enough positive constants $C_{1},c_1$ such that,  with probability greater than $1-(np)^{-c_1}$, 
					\begin{eqnarray*}
						\left| \big(l(\btheta^*)-l\left( \check{\btheta}\right) \big) - \big( l_{E}(\btheta^*)-l_{E}\left( \check{\btheta}\right) \big)\right| 
						\leq C_1\left(1+ \frac{\log(np)}{\sqrt{ p}} \right)\sqrt{\frac{\log(np)}{n}} \left\| \check{\btheta}-\btheta^*\right\|_{2},
					\end{eqnarray*}
					holds  uniformly for all random variable $\check{\btheta}\in \bD_{n,p}$.
					\\
					Note that   $\btheta^*$ is the minimizer of $l_{E}(\cdot)$.  By Lemma \ref{bound_l2}, there exists a constant $C_{2}>0$, such that for all $\check{\btheta}\in \bD_{n,p}$, we have,
					\begin{eqnarray*}
						&&\big(l(\btheta^*)-l\left( \check{\btheta}\right)\big) - \big( l_{E}(\btheta^*)-l_{E}\left( \check{\btheta}\right)\big)\\
						&\geq&l_{E}\left( \check{\btheta}\right)-l_{E}(\btheta^*)\\
						&= &\left( \check{\btheta}- \btheta^*\right) ^{\top}\nabla^2l_{E} (c_{\btheta}\btheta^*+(1-c_{\btheta})\check{\btheta})\left( \check{\btheta}- \btheta^*\right)\\
						&\geq & \left\|\check{\btheta}- \btheta^*\right\|_{2}^{2}\inf_{\|\ba\|_{2}\leq 1}\left( \ba^{\top}\nabla^2l_{E} (c_{\btheta}\check{\btheta}+(1-c_{\btheta})\btheta^*)\ba\right) \\
						&\geq&  \frac{C_{2}}{e^{4\kappa_{0}+4\kappa_{1}}}\left\|\check{\btheta}- \btheta^*\right\|_{2}^{2}.
					\end{eqnarray*}
					Here $0\leq c_{\btheta} \leq 1$ is a random scalar that may depend on $\check{\btheta}$.
					Consequently, as $np\to \infty$, with probability greater than $1-(np)^{-c_1}$,  there exists a constant $C_{3}>0$ such that
					\begin{equation}\label{thm2_c1}
						\sup_{\check{\btheta} \in \bD_{n,p}}\left\|\check{\btheta}-\btheta^*\right\|_{2}\leq C_{3} e^{4\kappa_{0}+4\kappa_{1}}\sqrt{\frac{\log(np)}{n}}\left(1+ \frac{\log(np)}{\sqrt{ p}} \right).
					\end{equation}
					Note that the MLE $\hat{\btheta}$   satisfies that $ \hat{\btheta}\in \bB_{\infty}\left( \btheta^*, c_{r}e^{-4\kappa_{0}-4\kappa_{1}}\right)$ and $l(\hat{\btheta})\leq l(\btheta^*)$. Therefore we have $\hat{\btheta}\in \bD_{n,p}$, and from \eqref{thm2_c1} we can conclude that 
					as $np\to \infty$,  with probability tending to 1,
					\begin{equation*}
						\frac{1}{\sqrt{p}}\left\|\hat{\btheta}-\btheta^*\right\|_{2}\leq C_{3} e^{4\kappa_{0}+4\kappa_{1}}\sqrt{\frac{\log(np)}{np}}\left(1+ \frac{\log(np)}{\sqrt{ p}} \right).
					\end{equation*}
				\end{proof}
				
				\subsection{Proof of Theorem \ref{thm3}}
				Before presenting the proof, we first introduce a technical lemma.
				\begin{lem}\label{thm3_1}
					For all $\ba,\bb,\bc,\bd \in \mathbb{R}^{K}$ s.t. $\max\{\|\ba-\bd\|_{\infty},\|\bb-\bc\|_{\infty}\}\leq 1$ and any positive $z_1,z_2$ s.t. $z_1z_2\geq 1/4$, we have:
					\begin{equation*}
						\frac{\left( \sum_{k=1}^{K}e^{a_{k}+b_{k}}\right) \left(\sum_{k=1}^{K}e^{c_{k}+d_{k}} \right) }{\left(\sum_{k=1}^{K}e^{a_{k}+c_{k}}\right) \left(\sum_{k=1}^{K}e^{b_{k}+d_{k}}  \right)}\leq (e-1)^2 k^2 \prod_{1\leq j,s\leq K}e^{ z_{1}\left| a_{k}-d_{k}\right|^2+z_{2}\left|b_{s}-c_{s}\right|^2}.
					\end{equation*}
					Here $a_k, b_k, c_k$ and $d_k$ are the $k$th element of $\ba,\bb,\bc$ and $\bd$, respectively. 
				\end{lem}
				\begin{proof}
					Note that for all $0<x\leq  1$,  we have $1\leq \left( e^{x}-1\right) /x\leq e-1$ and for all $y$, we have $|e^{y}-1|\leq e^{|y|}-1$. Consequently,  for all $\ba,\bb,\bc,\bd \in \mathbb{R}^{K}$ s.t. $\max\{\|\ba-\bd\|_{\infty},\|\bb-\bc\|_{\infty}\}\leq 1$, we have,
					\begin{eqnarray*}
						&&\frac{\left( \sum_{k=1}^{K}e^{a_{k}+b_{k}}\right) \left(\sum_{k=1}^{K}e^{c_{k}+d_{k}} \right) }{\left(\sum_{k=1}^{K}e^{a_{k}+c_{k}}\right) \left(\sum_{k=1}^{K}e^{b_{k}+d_{k}}  \right)}\nonumber\\
						&=&1+\sum_{1\leq j<s\leq K}\frac{ e^{a_{j}+b_{j}+c_{s}+d_{s}} -e^{a_{j}+b_{s}+c_{j}+d_{s}}  +e^{a_{s}+b_{s}+c_{j}+d_{j}} -e^{a_{s}+b_{j}+c_{s}+d_{j}}  }{\left(\sum_{k=1}^{K}e^{a_{k}+c_{k}}\right) \left(\sum_{k=1}^{K}e^{b_{k}+d_{k}}  \right)}\nonumber\\
						&=&1+\sum_{1\leq j<s\leq K}\frac{ e^{a_{s}+b_{s}+c_{j}+d_{j}}  }{\left(\sum_{k=1}^{K}e^{a_{k}+c_{k}}\right) \left(\sum_{k=1}^{K}e^{b_{k}+d_{k}}  \right)}\left( e^{a_{j}+d_{s}-a_{s}-d_{j}} -1\right) \left( e^{b_{j}+c_{s}-b_{s}-c_{j}}-1 \right)\nonumber\\
						&\leq&1+\sum_{1\leq j<s\leq K}\frac{ e^{a_{s}+b_{s}+c_{j}+d_{j}} }{\left(\sum_{k=1}^{K}e^{a_{k}+c_{k}}\right) \left(\sum_{k=1}^{K}e^{b_{k}+d_{k}}  \right)}\left(e^{\left| a_{j}-d_{j}\right| +\left| a_{s}-d_{s}\right| } -1\right) \left( e^{\left| b_{j}-c_{j}\right| +\left|b_{s}-c_{s}\right| }-1 \right)   \nonumber\\
						&\leq&1+(e-1)^2\sum_{1\leq j<s\leq K}\frac{ e^{a_{s}+b_{s}+c_{j}+d_{j}}\left( \left| a_{j}-d_{j}\right| +\left| a_{s}-d_{s}\right| \right)\left(\left| b_{j}-c_{j}\right| +\left|b_{s}-c_{s}\right| \right) }{\left(\sum_{k=1}^{K}e^{a_{k}+c_{k}}\right) \left(\sum_{k=1}^{K}e^{b_{k}+d_{k}}  \right)}   \nonumber \\
						&\leq& 1+(e-1)^2\sum_{1\leq j<s\leq K} \frac{ e^{a_{s}+b_{s}+c_{j}+d_{j}}\left( e^{\left( \left| a_{j}-d_{j}\right| +\left| a_{s}-d_{s}\right| \right)\left(\left| b_{j}-c_{j}\right| +\left|b_{s}-c_{s}\right| \right)} -1\right) }{\left(\sum_{k=1}^{K}e^{a_{k}+c_{k}}\right) \left(\sum_{k=1}^{K}e^{b_{k}+d_{k}}  \right)}
						\nonumber \\
						&\leq& (e-1)^2 \sum_{1\leq j,s\leq k} e^{\left| a_{j}-d_{j}\right|\left|b_{s}-c_{s}\right|}\nonumber \\
						&\leq& (e-1)^2 \sum_{1\leq j,s\leq k} e^{z_{1}\left| a_{j}-d_{j}\right|^2+z_{2}\left|b_{s}-c_{s}\right|^2}\nonumber \\
						&\leq& (e-1)^2 k^2 \prod_{1\leq j,s\leq k}e^{ z_{1}\left| a_{j}-d_{j}\right|^2+z_{2}\left|b_{s}-c_{s}\right|^2}
					\end{eqnarray*}
					holds	for any positive $z_1,z_2$ s.t. $z_1z_2\geq 1/4$.
				\end{proof}
				Next we proceed to 
				the proof of Theorem \ref{thm3}.
				
				\noindent
				{\bf Proof of Theorem \ref{thm3}}
				\begin{proof}
					Recall the functions $l_{E}\left(\btheta_{(i)},\btheta_{(-i)}\right)$ and $l\left(\btheta_{(i)}, \btheta_{(-i)} \right)$ defined in the proof of Theorem \ref{thm2}. 
					We denote the MLE studied in Theorem \ref{thm2} as $\hat{\btheta}$ $\big($the local minimizer of $l(\btheta)$ in $\bB_{\infty}\left( \btheta^*, c_{r}e^{-4\kappa_{0}-4\kappa_{1}} \right)\big)$. Also, let $\hat{\btheta}^{(s)}=(\hat{\beta}_{1,0}^{(s)},\ldots, \hat{\beta}_{p,0}^{(s)},\hat{\beta}_{1,1}^{(s)},\ldots, \hat{\beta}_{p,1}^{(s)})^\top$ be the local minimizer of $l(\btheta)$ in the $\ell_\infty$ ball $\bB_{\infty}\left( \btheta^*, r_s\right)$ where  
					\begin{eqnarray}\label{thm3_con}
						r_s&=&  e^{8\kappa_{0}+8\kappa_{1}}\log\log(np)\frac{ \sqrt{\log(np)} }{\left( np\right) ^{s}}\left(1+  \frac{\log(np)}{\sqrt{p}} \right),
					\end{eqnarray}
					for some  {given} constants $s\geq0$. 
					Let
					\begin{align}\label{thm3_s}
						s_{0}&= &\frac{12\left( \kappa_{0}+\kappa_{1}\right) +\log\log(np)/2 +\log\log\log(np)+\log\left(1+  \frac{\log(np)}{\sqrt{p}} \right)-\log(c_{r})}{\log(np)}.
					\end{align}
					We then  have $\hat{\btheta}^{(s_{0})}=\hat{\btheta}$. 	Under the condition that $\kappa_{0}+\kappa_{1}\leq c\log(np)$ for some positive constant $c$, we have $s_{0}<1/2$ when $np$ is large enough.
					
					Next we will show that with probability tending to 1, uniformly for all $s\in [s_0,1/2]$, $\hat{\btheta}^{(s)}$, the local MLE for $\bB_{\infty}\left( \btheta^*, r_s\right)$,  is also  the local MLE in a smaller ball $\bB_{\infty}\left( \btheta^*, r_{s'}\right)$ where
					\begin{equation*}
						r_{s'}= (np)^{\frac{2s-1}{4}}{r_{s}}= e^{8\kappa_{0}+8\kappa_{1}}\log\log(np)\frac{ \sqrt{\log(np)} }{\left( np\right) ^{s'}}\left(1+  \frac{\log(np)}{\sqrt{p}} \right),
					\end{equation*}
					and $s'=\frac{2s+1}{4}<1/2$.
					Note  that $\btheta_{(i)}^{*}$ is the minimizer of $l_{E}\left( \cdot,\btheta_{(-i)}^{*}\right) $. Given $\btheta_{(-i)}^{*}$,  the Hessian matrix of $l_{E}\left(\cdot, \btheta_{(-i)}^{*}\right) $ evaluated at $\btheta_{(i)}$ is a  $2\times 2$ matrix given as: 
					\begin{equation*}
						\E\bV^{(i)}\left(\btheta_{(i)} \right) := \left[
						\begin{array}{ccc}
							\E\bV^{(i)}_{1}\left(\btheta_{(i)} \right) & \E\bV^{(i)}_{2} \left(\btheta_{(i)} \right) \\
							\E\bV^{(i)}_{2}\left(\btheta_{(i)} \right) & \E\bV^{(i)}_{3}\left(\btheta_{(i)} \right)
						\end{array}
						\right].
					\end{equation*}
					Following the calculations in Lemma \ref{bound_l2},  there exists a constant $C_1>0$ which is independent of $\btheta_{(i)}$, such that, for all $\btheta_{(i)} \in \bB_{\infty}\left( \btheta^*_{(i)}, c_{r}e^{-4\kappa_{0}-4\kappa_{1}} \right)$, we have
					\begin{equation*}
						\E\bV^{(i)}_2\left(\btheta_{(i)} \right)+\E\bV^{(i)}_1\left(\btheta_{(i)} \right)= \sum_{j=1,\:j \neq i}^{p}\frac{1}{p}\frac{e^{\beta_{i,0}+\beta^{*}_{j,0}}} {(1+e^{\beta_{i,0}+\beta^{*}_{j,0}} +e^{\beta_{i,1}+\beta^{*}_{j,1}})^2}\geq C_1e^{-2\kappa_{0}-4\kappa_{1}},
					\end{equation*}
					and
					\begin{equation*}
						-\E\bV^{(i)}_2\left(\btheta_{(i)} \right)
						\geq C_{1}e^{-4\kappa_{0}-4\kappa_{1}}; \quad \quad 	\E\bV^{(i)}_2\left(\btheta_{(i)} \right)+\E\bV^{(i)}_3\left(\btheta_{(i)} \right)
						\geq C_{1}e^{-4\kappa_{0}-4\kappa_{1}}.
					\end{equation*}
					Then we have for all $\btheta_{(i)} \in \bB_{\infty}\left( \btheta^*_{(i)}, c_{r}e^{-4\kappa_{0}-4\kappa_{1}}\right)$,
					\begin{eqnarray*}
						&& \left\|\E\bV^{(i)}\left(\btheta_{(i)} \right)\right\|_{2} \\
						&=& \inf_{\|\bz\|_{2}=1}\Bigg( \left( \E\bV^{(i)}_2\left(\btheta_{(i)} \right)+\E\bV^{(i)}_1\left(\btheta_{(i)} \right) \right) z_{1}^2+	\left( \E\bV^{(i)}_2\left(\btheta_{(i)} \right)+\E\bV^{(i)}_3\left(\btheta_{(i)} \right) \right)  z_{2}^2\\
						&&- \E\bV^{(i)}_2\left(\btheta_{(i)} \right) \left( z_{1}-z_{2}\right) ^2\Bigg) \\
						&\geq& \inf_{\|\bz\|_{2}=1}\Bigg\{  \left( \E\bV^{(i)}_2\left(\btheta_{(i)} \right)+\E\bV^{(i)}_1\left(\btheta_{(i)} \right) \right) z_{1}^2+	\left( \E\bV^{(i)}_2\left(\btheta_{(i)} \right)+\E\bV^{(i)}_3\left(\btheta_{(i)} \right) \right)  z_{2}^2 \Bigg\} \\
						&\geq& \frac{C_1}{e^{4\kappa_{0}+4\kappa_{1}}}\inf_{\|\bz\|_{2}=1}\Bigg\{	  z_{1}^2+  z_{2}^2 \Bigg\}\\
						&=&\frac{C_1}{e^{4\kappa_{0}+4\kappa_{1}}}.
					\end{eqnarray*}
					Consequently, for any $i$ and $s\in[s_0,1/2]$, there exists a random scalar $\btheta_{(i)}^{(\xi)}$, s.t.
						\begin{eqnarray}\label{Li_low}
							&& l\left(\btheta_{(i)}^*, \hat{\btheta}^{(s)}_{(-i)}\right) - l\left(\hat{\btheta}^{(s)}_{(i)},\hat{\btheta}^{(s)}_{(-i)} \right)   -\left[ l_{E}\left(\btheta_{(i)}^{*},\btheta_{(-i)}^{*}\right) - l_{E}\left(\hat{\btheta}^{(s)}_{(i)},\btheta_{(-i)}^{*}\right) \right]\\
							&\geq&  l_{E}\left(\hat{\btheta}^{(s)}_{(i)},\btheta_{(-i)}^{*}\right)- l_{E}\left(\btheta_{(i)}^{*},\btheta_{(-i)}^{*}\right)\nonumber\\
							&=& \frac{1}{2}\left( \btheta_{(i)}^*-\hat{\btheta}^{(s)}_{(i)}\right) ^{\top}  l_{E}''\left( \btheta_{(i)}^{(\xi)},\btheta_{(-i)}^{*}\right)  \left( \btheta_{(i)}^*-\hat{\btheta}^{(s)}_{(i)}\right)\nonumber\\
							&\geq& \|\btheta_{(i)}^*-\hat{\btheta}^{(s)}_{(i)}\|_2^{2}  \left(\inf_{\|\btheta_{(i)}\|_{\infty}<\kappa}\left\|\E\bV^{(i)}\left( \btheta_{(i)}^{(\xi)},\btheta_{(-i)}^{*}\right) \right\|_{2} \right)\nonumber \\
							&\geq& \frac{C_1\|\btheta_{(i)}^*-\hat{\btheta}^{(s)}_{(i)}\|_{\infty}^2}{e^{-4\kappa_{0}-4\kappa_{1}}}.\nonumber
						\end{eqnarray}
					On the other hand, notice that 
					\begin{eqnarray*}
						&&l\left(\btheta_{(i)}^*, \hat{\btheta}^{(s)}_{(-i)}\right) -l\left(\btheta_{(i)},\hat{\btheta}^{(s)}_{(-i)} \right)   -\left[ l_{E}\left(\btheta_{(i)}^{*},\btheta_{(-i)}^{*}\right)- l_{E}\left(\btheta_{(i)},\btheta_{(-i)}^{*}\right)  \right]\\
						&\leq&\left|l\left(\btheta_{(i)}^*, \hat{\btheta}^{(s)}_{(-i)}\right) -l\left(\btheta_{(i)},\hat{\btheta}^{(s)}_{(-i)} \right)  - \left[ l\left(\btheta_{(i)}^{*},\btheta_{(-i)}^{*}\right) - l\left(\btheta_{(i)},\btheta_{(-i)}^{*}\right) \right]\right|\\
						&&+ \left|l\left(\btheta_{(i)}^{*},\btheta_{(-i)}^{*} \right) -l\left(\btheta_{(i)}, \btheta_{(-i)}^{*}\right)-\left[ l_{E}\left(\btheta_{(i)}^{*},\btheta_{(-i)}^{*}\right)- l_{E}\left(\btheta_{(i)},\btheta_{(-i)}^{*}\right)  \right]\right|.
					\end{eqnarray*}
					By Corollary \ref{Cor1}, there exist large enough positive constants $C_2$ and $c_1$ which are independent of $\btheta_{(i)}$  such that, with probability greater than $1-(np)^{-c_1}$, 
					\begin{eqnarray}\label{thm3_e1}
						&& \left|l\left(\btheta_{(i)}^{*},\btheta_{(-i)}^{*} \right) -l\left(\btheta_{(i)}, \btheta_{(-i)}^{*}\right)-\left[ l_{E}\left(\btheta_{(i)}^{*},\btheta_{(-i)}^{*}\right)- l_{E}\left(\btheta_{(i)},\btheta_{(-i)}^{*}\right)  \right]\right|\\
						&\leq& C_{2}\left(1+ \frac{\log(np)}{\sqrt{ p}} \right)\sqrt{\frac{\log(np)}{np}} \left\|\btheta_{(i)}-\btheta^*_{(i)}\right\|_{2}\nonumber\\
						&\leq& \sqrt{2}C_{2} \left(1+ \frac{\log(np)}{\sqrt{ p}} \right) \sqrt{ \frac{\log(np)}{np}}\left\|\btheta_{(i)}-\btheta^*_{(i)}\right\|_{\infty},\nonumber
					\end{eqnarray}
					holds 
					uniformly for any  $\btheta_{(i)} \in \bB_{\infty}\left( \btheta^*_{(i)},\alpha_{0}\right)$ where $\alpha_{0}<1/4$. 
					By Lemma \ref{thm3_1}, as $np\to \infty$,  there exist big enough positive constants $C_3$  and $c_2$ which are independent of $\btheta_{(i)}$  such that, with probability greater than $1-(np)^{-c_2}$,
					\begin{eqnarray*}
						&&  \left|l\left(\btheta_{(i)}^*, \hat{\btheta}^{(s)}_{(-i)}\right) -l\left(\btheta_{(i)},\hat{\btheta}^{(s)}_{(-i)} \right)  - \left[ l\left(\btheta_{(i)}^{*},\btheta_{(-i)}^{*}\right) - l\left(\btheta_{(i)},\btheta_{(-i)}^{*}\right) \right]\right|\nonumber\\
						&=&     \Bigg|\frac{1}{np}\sum_{j=1,\:j \neq i}^{p}   \Bigg[- n\log \left( \frac{  1+ e^{\beta^{*}_{i,0} +\beta^{*}_{j,0 } }+e^{\beta^{*}_{i,1} +\beta^{*}_{j,1 }   }}{ 1+ e^{\beta^{*}_{i,0} +\hat{\beta}^{(s)}_{j,0 } }+e^{\beta^{*}_{i,1} +\hat{\beta}^{(s)}_{j,1 } }} \times \frac{ 1+ e^{\beta_{i,0} +\hat{\beta}^{(s)}_{j,0 } }+e^{\beta_{i,1} +\hat{\beta}^{(s)}_{j,1 } }}{  1+ e^{\beta_{i,0} +\beta^{*}_{j,0 } }+e^{\beta_{i,1} +\beta^{*}_{j,1 }   }}\right)  \nonumber\\
						&& +   \log\left( \frac{1+e^{\beta^{*}_{i,1} +\beta^{*}_{j,1 }}}{1+e^{\beta^{*}_{i,1} + \hat{\beta}^{(s)}_{j,1 }}}\times\frac{1+e^{\beta_{i,1} + \hat{\beta}^{(s)}_{j,1 }}}{1+e^{\beta_{i,1} +\beta^{*}_{j,1 }}} \right)  d_{i,j}\nonumber\\
						&& + \log\left(\frac{ e^{\beta^{*}_{i,0}+\beta^{*}_{j,0}}+e^{\beta^{*}_{i,1} +\beta^{*}_{j,1 }}} {e^{\beta^{*}_{i,0}+\hat{\beta}^{(s)}_{j,0}}+e^{\beta^{*}_{i,1} + \hat{\beta}^{(s)}_{j,1 }}} \times\frac{ e^{\beta_{i,0}+\hat{\beta}^{(s)}_{j,0}}+e^{\beta_{i,1} + \hat{\beta}^{(s)}_{j,1 }}} { e^{\beta_{i,0}+\beta^{*}_{j,0}}+e^{\beta_{i,1} +\beta^{*}_{j,1 }}   } \right)  b_{i,j} \Bigg] \Bigg|\nonumber\\
						&\leq&\frac{C_3}{p}\sum_{j=1,\:j \neq i}^{p}\Bigg(z_1\left|\beta^{*}_{i,0}- \beta_{i,0}  \right|^2 +z_1\left|\beta^{*}_{i,1}- \beta_{i,1}  \right|^2+z_2\left|\beta^{*}_{j,0}-\hat{\beta}^{(s)}_{j,0}  \right|^2 +z_2\left|\beta^{*}_{j,1}- \hat{\beta}^{(s)}_{j,1}  \right|^2 \Bigg)\nonumber\\
						&\leq& 2C_3z_1\|\btheta_{(i)}-\btheta^*_{(i)}\|^2_{\infty}+ \frac{C_3z_2}{p}\|\hat{\btheta}^{(s)}_{(-i)}-\btheta^*_{(-i)}\|^2_{2} \nonumber\\
						&\leq& 2C_3z_1\frac{ e^{16\kappa_{0}+16\kappa_{1}}[\log\log(np)]^2\log(np)}{\left( np\right) ^{2s}}\left(1+  \frac{\log(np)}{\sqrt{p}} \right)^2\\
						&&+ \frac{C_3z_2}{p}\frac{e^{8\kappa_{0}+8\kappa_{1}}\log(np)}{n}\left(1+ \frac{\log(np)}{\sqrt{ p}} \right)^2,
					\end{eqnarray*}
					holds uniformly for all for any positive $z_1,z_2$ s.t. $z_1z_2\geq 1/4$, $s\in[s_0,1/2]$ and $\btheta_{(i)}\in \bB_{\infty}\left( \btheta^*_{(i)},r_s\right)$.	Here in the last step we have used inequality \eqref{thm2_c1} and the fact that for all $s$, $\hat{\btheta}^{(s)}\in \bD_{n,p}$.
					Let $z_{1}=0.5e^{-4\kappa_{0}-4\kappa_{1}}[\log\log(np)]^{-1}$ $(np)^{s-1/2}$ and $z_{2}=0.5e^{4\kappa_{0}+4\kappa_{1}}\log\log(np) (np)^{1/2-s}$, we have, there exists a big enough constant $C_{4}>0$, with probability greater than $1-(np)^{-c_2}$,
					\begin{eqnarray}\label{thm3_e2}
						&&
						\left|l\left(\btheta_{(i)}^*, \hat{\btheta}^{(s)}_{(-i)}\right) -l\left(\btheta_{(i)},\hat{\btheta}^{(s)}_{(-i)} \right)  - \left[ l\left(\btheta_{(i)}^{*},\btheta_{(-i)}^{*}\right) - l\left(\btheta_{(i)},\btheta_{(-i)}^{*}\right) \right]\right|\\
						&\leq& C_{4} \frac{e^{12\kappa_{0}+12\kappa_{1}} \log\log(np)\log(np)}{(np)^{s+1/2}}\left(1+  \frac{\log(np)}{\sqrt{p}} \right)^2,\nonumber
					\end{eqnarray}
					holds uniformly for all $s\in[s_0,1/2]$ and $\btheta_{(i)}\in \bB_{\infty}\left( \btheta^*_{(i)},r_s\right)$.
					Combining the inequalities \eqref{thm3_e1} and \eqref{thm3_e2}, we have, with probability greater than $1-(np)^{-c_3}$ for a large enough constant  $c_3>0$,  
					\begin{align}\label{Li_up}
						&\left|l\left(\btheta_{(i)}^*, \hat{\btheta}^{(s)}_{(-i)}\right) -l\left(\btheta_{(i)},\hat{\btheta}^{(s)}_{(-i)} \right)   -\left[ l_{E}\left(\btheta_{(i)}^{*},\btheta_{(-i)}^{*}\right)- l_{E}\left(\btheta_{(i)},\btheta_{(-i)}^{*}\right)  \right]\right|\\
						\leq& \sqrt{2}C_{2} \left(1+ \frac{\log(np)}{\sqrt{ p}} \right) \sqrt{ \frac{\log(np)}{np}}\left\|\btheta_{(i)}-\btheta^*_{(i)}\right\|_{\infty}\nonumber\\
						&+C_{4} \frac{e^{12\kappa_{0}+12\kappa_{1}}\log\log(np)\log(np)}{(np)^{s+1/2}}\left(1+  \frac{\log(np)}{\sqrt{p}} \right)^2,\nonumber
					\end{align}
					holds uniformly for all $s\in[s_0,1/2]$, all $i=1,\ldots,p$ and  $\btheta_{(i)}\in \bB_{\infty}\left( \btheta^*_{(i)},r_{s}\right)$. 
					
					Combining the inequalities \eqref{Li_low} and \eqref{Li_up},  we have, with probability greater than $1-(np)^{-c_3}$,
					\begin{multline*}
						\frac{C_1}{e^{4\kappa_{0}+4\kappa_{1}}}\|\hat{\btheta}^{(s)}_{(i)}-\btheta_{(i)}^*\|_{\infty}^2   \leq \sqrt{2}C_{2} \left(1+ \frac{\log(np)}{\sqrt{ p}} \right) \sqrt{ \frac{\log(np)}{np}}\left\|\hat{\btheta}^{(s)}_{(i)}-\btheta^*_{(i)}\right\|_{\infty}\nonumber\\
						+C_{4} \frac{e^{12\kappa_{0}+12\kappa_{1}}\log\log(np)\log(np)}{(np)^{s+1/2}}\left(1+  \frac{\log(np)}{\sqrt{p}} \right)^2,
					\end{multline*}
					holds uniformly for all $s\in[s_0,1/2]$ and all $i=1,\ldots,p$. 
					Notice that the constants $C_1$, $C_2$, $C_3$, $C_4$, $c_2$ and $c_3$ are all independent of $r_s$ and $s$. This 
					indicates that   there exists a big enough constant $C_5>0$ which is independent of $r_s$ and $s$ s.t.
					\begin{align}\label{thm3_e3}
						&\|\hat{\btheta}^{(s)}-\btheta^*\|_{\infty}\\
						\leq& C_{5}   e^{4\kappa_{0}+4\kappa_{1}}\sqrt{ \frac{\log(np)}{np}}\left(1+ \frac{\log(np)} {\sqrt{ p}} \right)+ C_{5}\frac{e^{8\kappa_{0}+8\kappa_{1}} \sqrt{\log(np) \log\log(np)}} {(np)^{\frac{2s+1}{4}}}\left(1+ \frac{\log(np)}{\sqrt{p}} \right)\nonumber\\
						\leq&2C_{5}\frac{e^{8\kappa_{0}+8\kappa_{1}} \sqrt{\log(np) \log\log(np)}} {(np)^{\frac{2s+1}{4}}}\left(1+ \frac{\log(np)}{\sqrt{p}} \right)\nonumber\\
						\leq&e^{8\kappa_{0}+8\kappa_{1}}\log\log(np)\frac{\sqrt{\log(np)}}{(np)^{\frac{2s+1}{4}}}\left(1+ \frac{\log(np)}{\sqrt{p}} \right).\nonumber
					\end{align}
					Recall that   $\hat{\btheta}^{(s)}$ is the   local maximizer of $l(\btheta)$ in  $\bB_{\infty}\left( \btheta^*, r_{s}\right)$   with
					\begin{eqnarray*}
						r_s&=&  e^{8\kappa_{0}+8\kappa_{1}}\log\log(np)\frac{ \sqrt{\log(np)} }{\left( np\right) ^{s}}\left(1+  \frac{\log(np)}{\sqrt{p}} \right).
					\end{eqnarray*}
					Thus far we have proved that: with probability greater than $1-(np)^{-c_3}$ for some large   enough constant $c_3>0$,  uniformly for all $s\in[s_0,1/2]$, $\hat{\btheta}^{(s)}$ is also within the ball $\bB_{\infty}\left( \btheta^*, r_{s'}\right)$ with 
					\begin{equation}\label{s_seq1}
						{r_{s'}}= (np)^{\frac{2s-1}{4}}{r_{s}}=e^{8\kappa_{0}+8\kappa_{1}}\log\log(np)\frac{\sqrt{\log(np)}}{(np)^{\frac{2s+1}{4}}}\left(1+ \frac{\log(np)}{\sqrt{p}} \right).
					\end{equation}

						Now define a series $\{s_{i};i=0,1,\cdots\}$ s.t. $s_{0}$ is defined as in equation \eqref{thm3_s} and $s_{k}=s_{k-1}/2+1/4$. We have: 
						\begin{equation}\label{s_seq2}
							s_{k}-\frac{1}{2}=\frac{1}{2}\left( s_{k-1}-\frac{1}{2}\right) = \frac{1}{2^{k}}\left( s_{0}-\frac{1}{2}\right).
						\end{equation}
						Then, we have $s_{k-1}<s_{k}<1/2$ for $k>1$ and $\lim_{k\to\infty}s_{k}\to 1/2$.
						
						Let  $K=\lfloor\log_{2}\left( \log(np)\right)\rfloor+1$ where $\lfloor\cdot\rfloor$ is the smallest integer function. Beginning with $\hat{\btheta}^{(s_0)}=\hat{\btheta}$, and repeatedly using the result in \eqref{thm3_e3}   for  $K$ times, we have, with probability greater than $1-(np)^{-c_3}$,  
						we can sequentially reduce the upper bound   from 
						$\|  \hat{\btheta}-\btheta^*\|_{\infty}  \leq  e^{8\kappa_{0}+8\kappa_{1}}\log\log(np) \frac{\sqrt{\log(np)}}{(np)^{\frac{1}{2}+\frac{2s_{0}-1}{2}}}\left(1+ \frac{\log(np)}{\sqrt{p}} \right)$ to:
						\begin{eqnarray*}
							\|  \hat{\btheta}-\btheta^*\|_{\infty} 
							&\leq& e^{8\kappa_{0}+8\kappa_{1}}\log\log(np) \frac{\sqrt{\log(np)}}{(np)^{\frac{1}{2}+ \frac{2s_{0}-1}{2^{K}}}}\left(1+ \frac{\log(np)}{\sqrt{p}} \right)\\
							&=& (np)^{\frac{1-2s_{0}}{2^{K}}}e^{8\kappa_{0}+8\kappa_{1}}\log\log(np) \sqrt{\frac{\log(np)}{np}}\left(1+ \frac{\log(np)}{\sqrt{p}} \right)\\
							&\le& e^{1-2s_{0}} e^{8\kappa_{0}+8\kappa_{1}}\log\log(np)\sqrt{\frac{\log(np)}{np}}\left(1+ \frac{\log(np)}{\sqrt{p}} \right).
						\end{eqnarray*}
						Here in the last step of above inequality, we have used the fact that when $K=\lfloor\log_{2}\left( \log(np)\right)\rfloor+1> \log_{2} \left( \log(np)\right)$,
						\begin{equation*}
							(np)^{\frac{1-2s_{0}}{2^{K}}}\le \left( (np)^{\frac{\log(2)}{\log(np)}}\right) ^{\frac{1-2s_{0}}{\log(2)}} = 2^{\frac{1-2s_{0}}{\log(2)}}=e^{1-2s_{0}}.
						\end{equation*}
						Thus, assuming that $np\rightarrow \infty, n\ge 2$,      $\kappa_{0}+\kappa_{1}  \le c\log(np)$ for some positive constant $c$, we have, with probability tending to 1,
						\begin{equation*}
							\left\|\hat{\btheta}-\btheta^*\right\|_{\infty}\lesssim e^{8\kappa_{0}+8\kappa_{1}}\log\log(np) \sqrt{\frac{\log(np)}{np}}\left(1+ \frac{\log(np)}{\sqrt{p}} \right).
						\end{equation*}
					\end{proof}
					\subsection{Proof of Theorem \ref{thm4}}
					For brevity, we denote $\tilde{\alpha}_{r,i,j}:=e^{\tilde{\beta}_{i,r}+\tilde{\beta}_{i,r}}$, with $i,j=1,\cdots,p$ and $r=0,1$. We use the interior mapping theorem \citep{gragg1974optimal} to establish the existence and uniform consistency of the moment estimator. The interior mapping theorem is presented in Lemma \ref{Lemma_IMT}.
					
					\begin{lem}\label{Lemma_IMT}
						{\rm(Interior mapping theorem).}
						Let $\bF\left( \bx \right)=\left(\bF_{1}\left( \bx\right),\cdots,\bF_{p}\left( \bx \right) \right)^{\top}$ be a vector function on an open convex subset $\calD $ of $\mathbb{R}^p$ and $\left\|\bF'\left( \bx \right)-\bF'\left( \by \right)\right\|_{\infty}\lesssim \gamma\left\|\bx-\by \right\|$. Assume that $x_{0}\in \calD$ s.t.
						\begin{eqnarray*}
							&&\left\|\bF'\left( \bx_{0} \right)^{-1}\right\|_{\infty}\leq N ,\quad 
							\left\|\bF'\left(  \bx_{0} \right)^{-1}\bF\left(  \bx_{0} \right)\right\|_{\infty}\leq\delta,\quad h=2N\gamma\delta\leq1,\\
							&&t^*\equiv\frac{2}{h}\left( 1-\sqrt{1-h}\right)\delta,\quad \bB_{\infty}\left(\bx_{0},t^* \right) \subset\calD,
						\end{eqnarray*}
						where N and $\delta$ are positive constants that may depend on $x_{0}$ and $p$. Then the Newton iterates  $\bx_{n+1}\equiv \bx_{n} - \bF'\left(  \bx_{n}\right)^{-1}\bF\left(  \bx_{n} \right)$ exists and $\bx_n\in \bB_{\infty}\left(x_{0},t^* \right)\subset\calD $ for all $n\geq 0$. Moreover, $\bx^*=\lim_{n\to\infty} \bx_n$ exists, $x^*\in \overline{\bB_{\infty}\left(x_{0},t^* \right)}\subset\calD $,  where $\overline{A}$ denotes the closure of set $A$ and $\bF\left( \bx^* \right)=0$. 
					\end{lem}
					
					\noindent
					{\bf{Proof of Theorem \ref{thm4}.}}
					\begin{proof}
						Recall that the $\tilde{\btheta}_{(0)}$ is defined as the solution of
						\begin{equation*}
							\sum_{t=1}^{n} \sum_{j=1,\:j \neq i}^{p} X_{i,j}^t = n\sum_{j=1,\:j \neq i}^{p} \frac{ e^{\beta_{i,0} +\beta_{j,0}} }{1+ e^{\beta_{i,0} +\beta_{j,0 }}},\quad i=1,\cdots,p.
						\end{equation*}
						For any  $\bx\in \mathbb{R}^{p}$,  define a system of random functions $\bG\left( \bx \right)$:
						\begin{eqnarray*}
							&&\bG_{i}\left( \bx \right):=-\frac{1}{np}\sum_{t=1}^{n} \sum_{j=1,\:j \neq i}^{p} \left( X_{i,j}^t - \frac{ e^{x_i+x_j} }{1+ e^{x_i+x_j}}\right) ,\quad i=1,\cdots,p;\\
							&&\bG\left( \bx \right):=\left(\bG_{1}\left( \bx \right),\cdots,\bG_{p}\left( \bx \right) \right)^{\top}.
						\end{eqnarray*}
						As $np\to \infty$ and $\kappa=c\log(np) $  with a small enough  constant $c>0$, there exist big enough constants $C_1>0, c_1>0$ such that  with probability greater than $1-(np)^{c_1}$,  
						\begin{eqnarray}\label{MME_con0}
							\left\|\bG'\left( \bx \right)-\bG'\left( \by \right)\right\|_{\infty}&\leq &C_{1}\|\bx-\by\|_{\infty},\\
							\left\|\bG'\left( \btheta_{(0)}^* \right)^{-1}\right\|_{\infty}&\leq& C_{1}e^{2\kappa_{0}},\nonumber\\   
							\left\|\bG'\left(  \btheta_{(0)}^* \right)^{-1}\bG\left(  \btheta_{(0)}^* \right)\right\|_{\infty}&\leq& 	C_{1}e^{2\kappa_{0}}\sqrt{\frac{\log(n)\log(p)}{np}},\nonumber
						\end{eqnarray}
						hold for all $\bx,\by\in \bB_{\infty}\left(0,\kappa_{0} \right)$.    For brevity, the proof of inequalities in  \eqref{MME_con0} is provided independently in Section \ref{p_WMME_con0}.	
						Let $x_{0}=\btheta_{(0)}^*$ and $\calD =\text{Int }( \bB_{\infty}\left(0,\kappa_{0} \right))$. Here the notation $\text{Int}(A) $ denotes the interior of a given  set $A$. We then have:
						\begin{eqnarray*}
							&&	N=C_{1}e^{2\kappa_{0}},\quad \gamma=C_{1},\quad \delta=C_{1} e^{2\kappa_{0}} \sqrt{\frac{\log(n) \log(p)} {np}},\\
							&& 	h= 2N\gamma\delta= 2C_{1}^3 e^{4\kappa_{0}}\sqrt{\frac{\log(n)\log(p)}{np}} =o\left( 1\right) ,\\
							&&	t^*\equiv\frac{2C_{1}}{h}\left( 1-\sqrt{1-h}\right)e^{2\kappa_{0}} \sqrt{\frac{\log(n) \log(p)} {np}},\\
							&&	\bB_{\infty}\left(\btheta_{(1)}^*,t^* \right)  \subset  \calD.
						\end{eqnarray*}
						Note that  $\forall h\in (0,1)$, $1-\sqrt{1-h}<1-(1-h)=h$, we have
						\begin{equation*}
							t^*\equiv\frac{2C_{1}}{h}\left( 1-\sqrt{1-h}\right)e^{2\kappa_{0}} \sqrt{\frac{\log(n) \log(p)} {np}}<4e^{4\kappa_{0}}\sqrt{\frac{\log(n)\log(p)}{np}}.
						\end{equation*}
						Consequently, by Lemma \ref{Lemma_IMT}, we have that, with probability tending to 1,
						\begin{equation}\label{MME_b1}
							\left\|\tilde{\btheta}_{(0)}-\btheta_{(0)}\right\|_{\infty} \lesssim \sqrt{\frac{\log(n)\log(p)e^{4\kappa_{0}}}{np}}.
						\end{equation}

						Next, we  derive the error bound of $\tilde{\btheta}_{(1)}$ based on \eqref{MME_b1}. For $\bx\in \mathbb{R}^{p}$, define a system of random functions $\bF\left( \bx ;\tilde{\btheta}_{(0)}\right)$:
						\begin{eqnarray*}
							&&\bF_{i}\left( \bx ;\tilde{\btheta}_{(0)}\right):=-\frac{1}{np}\sum_{t=1}^{n} \sum_{j=1,\:j \neq i}^{p} \left\{X_{i,j}^tX_{i,j}^{t-1} -  \frac{\tilde{\alpha}_{0,i,j}}{1+\tilde{\alpha}_{0,i,j}}\left(1- \frac{1}{1+\tilde{\alpha}_{0,i,j}+e^{x_i+x_j}}\right)\right\},\\
							&&\bF\left( \bx ;\tilde{\btheta}_{(0)}\right):=\left(\bF_{1}\left( \bx ;\tilde{\btheta}_{(0)}\right),\cdots,\bF_{p}\left( \bx ;\tilde{\btheta}_{(0)}\right) \right)^{\top}.
						\end{eqnarray*}
						As $np\to \infty$ and $\kappa_{0}=c\log(np) $  with a small enough  constant $c>0$, there exist big enough constants $C_2>0, c_2>0$ such that  with probability greater than $1-(np)^{c_2}$,  
						\begin{eqnarray}\label{MME_con}  
							&&\left\|\bF'\left( \bx ;\tilde{\btheta}_{(0)}\right)-\bF'\left( \by ;\tilde{\btheta}_{(0)}\right)\right\|_{\infty}\leq C_{2} \left\|\bx-\by \right\|_{\infty},\\
							&&\left\|\bF'\left( \btheta_{(1)}^*;\tilde{\btheta}_{(0)} \right)^{-1}\right\|_{\infty}\leq C_{2} e^{12\kappa_{0}+6\kappa_{1}},\nonumber\\ 
							&&\left\|\bF'\left(  \btheta_{(1)}^*;\tilde{\btheta}_{(0)} \right)^{-1}\bF\left(  \btheta_{(1)}^* ;\tilde{\btheta}_{(0)}\right)\right\|_{\infty}\leq C_{2} e^{14\kappa_{0}+6\kappa_{1}}\sqrt{\frac{\log(n)\log(p)}{np}},\nonumber
						\end{eqnarray}
						hold for all $x,y\in \bB_{\infty}\left(0,\kappa_{1} \right)$.		For brevity, the proof of inequalities in  \eqref{MME_con} is provided independently in Section \ref{p_WMME_con}.	
						Let $x_{0}=\btheta_{(1)}^*$ and $\calD =\text{Int } (\bB_{\infty}\left(0,\kappa \right))$. 
						We then have:
						\begin{eqnarray*}
							&&N=C_{2}e^{12\kappa_{0}+6\kappa_{1}},\quad \gamma=C_{2},\quad \delta=C_{2} e^{14\kappa_{0}+6\kappa_{1}} \sqrt{\frac{\log(n) \log(p)} {np}} ,\\
							&&	h= 2N\gamma\delta= 2C_{2}^3 e^{26\kappa_{0}+12\kappa_{1}} \sqrt{\frac{\log(n)\log(p)}{np}} =o\left( 1\right) ,\\
							&&t^*\equiv\frac{2C_{2}}{h}\left( 1-\sqrt{1-h}\right)e^{14\kappa_{0}+6\kappa_{1}} \sqrt{\frac{\log(n)\log(p)} {np}},\\
							&&\bB_{\infty}\left(\btheta_{(1)}^*,t^* \right) \subset \calD.
						\end{eqnarray*}
						Note that $\forall h\in (0,1)$, $1-\sqrt{1-h}<1-(1-h)=h$, we have
						\begin{equation*}
							t^*\equiv\frac{2C_{2}}{h}\left( 1-\sqrt{1-h}\right)e^{14\kappa_{0}+6\kappa_{1}}\sqrt{\frac{\log(n)\log(p)}{np}} <4C_{2}e^{14\kappa_{0}+6\kappa_{1}}\sqrt{\frac{\log(n)\log(p)}{np}}.
						\end{equation*}
						Consequently, by Lemma \ref{Lemma_IMT}, we have that, with probability tending to 1,
						\begin{eqnarray*}
							\left\|\tilde{\btheta}_{(1)}-\btheta_{(1)}^{*}\right\|_{\infty}&\leq t^{*}\leq &4C_{2} e^{14\kappa_{0}+6\kappa_{1}}\sqrt{\frac{\log(n)\log(p)}{np}}.
						\end{eqnarray*}
						Combining with \eqref{MME_b1}, the theorem is proved.
					\end{proof}
					\subsection{Proof of   \eqref{MME_con0}}\label{p_WMME_con0}
					\begin{proof}
						Note that $\bG'\left( \bx \right)$ is a balanced symmetric matrix s.t.
						\begin{equation*}
							\bG'\left( \bx \right)_{i,j}  = \frac{1}{p} \frac{\alpha_{0,i,j}}{\left( 1+\alpha_{0,i,j}\right)^2 }, \quad\quad\quad 
							\bG'\left( \bx \right)_{i,i}  =\sum_{j=1,\:j \neq i}^{p}\bG'\left( \bx \right)_{i,j},
						\end{equation*}
						for $i=1,\cdots,p$ and $i\neq j$. Following \cite{yan2012approximating}, we construct a matrix $\bS = (S_{i,j})$ to approximate the inverse of $ \bG'\left( \bx \right)$. Specifically, for any $i\neq j$, we set
						\begin{equation*}
							S_{i,j}= -\frac{1}{\calT},\quad 
							S_{i,i}= \left( \sum_{j=1,\:j \neq i}^{p}\frac{1}{p} \frac{\alpha_{0,i,j}}{\left( 1+\alpha_{0,i,j}\right)^2 }\right)^{-1}  -\frac{1}{\calT},\quad 
							\calT =  \sum_{1\leq i\neq j\leq p}  \frac{1}{p} \frac{\alpha_{0,i,j}}{\left( 1+\alpha_{0,i,j}\right)^2 }.
						\end{equation*}
						Note that 
						\begin{equation*}
							\frac{e^{-2\kappa_{0}}}{4p}\leq \frac{1}{p} \frac{\alpha_{0,i,j}}{\left( 1+\alpha_{0,i,j}\right)^2 }\leq \frac{1}{4p},
						\end{equation*}
						we have 
						\begin{equation*}
							\calT\in \left( \frac{(p-1)e^{-2\kappa_{0}}}{4},\frac{(p-1)}{4}\right).
						\end{equation*}
						By Theorem 1 in \cite{yan2012approximating}, with $m= e^{-2\kappa_{0}}/\left( 4p\right) $ and $M=1/\left( 4p\right) $,   there exists a big enough constant $C_{1}>0$, we have that 		
						\begin{eqnarray*}
							\left\| \bG'\left( \bx \right)^{-1}- \bS\right\|_{\max}&\leq& \frac{M}{m^2} \frac{pM+(p-2)m}{2(p-2)m}  \frac{1}{(p-1)^2} + \frac{1}{2m(p-1)^2}\\
							&=&\frac{1}{2m(p-1)^2}\left(1+ \frac{M}{m} +\frac{M^2}{m^2} \frac{p}{p-2}\right). \\
							&\leq &C_{1}\frac{e^{6\kappa_{0}}}{p^2},
						\end{eqnarray*}
						where $\|\bA\|_{\max}=\max_{i,j}A_{i,j}$. Then, there exists a big enough constant $C_{2}>0$,
						\begin{eqnarray*}
							\left\| \bG'\left( \bx \right)^{-1}\right\|_{\infty}&\leq&\left\| \bG'\left( \bx \right)- \bS\right\|_{\infty}+\left\|  \bS\right\|_{\infty}\\
							&\leq &p\left\| \bG'\left( \bx \right)^{-1}- \bS\right\|_{\max}+\frac{p}{\calT}+ \max_{i}\left( \sum_{j=1,\:j \neq i}^{p}\frac{1}{p} \frac{\alpha_{0,i,j}}{\left( 1+\alpha_{0,i,j}\right)^2 }\right)^{-1}\\
							&\leq &pC_{1}\frac{e^{6\kappa_{0}}}{p^2}+ 8e^{2\kappa_{0}}\\
							&< &C_{2}e^{2\kappa_{0}}.
						\end{eqnarray*}
						By Lemma \ref{mixing} and Lemma \ref{Bers}, there exist  big positive constants $C_{3}, c_1$ such that, with probability greater than $1-(np)^{-c_{1}}$, 
						\begin{align}\label{thm4_e4}
							&\left\| \bG\left( \btheta_{(0)}^*  \right)\right\|_{\infty}\\
							=&\max_{i}\left|\frac{1}{np}\sum_{t=1}^{n} \sum_{j=1,\:j \neq i}^{p} \left( X_{i,j}^t - \frac{ e^{\beta^*_{i,0}+\beta^*_{j,0}} }{1+ e^{\beta^*_{i,0}+\beta^*_{j,0}}}\right)\right|\nonumber\\
							=&\max_{i}\left|\frac{1}{np}\sum_{t=1}^{n} \sum_{j=1,\:j \neq i}^{p} \left( X_{i,j}^t - \E\left( X_{i,j}^t\right)  \right)\right|\nonumber\\
							\leq &\frac{1}{np}\max_{i}\left|\sum_{j=1,\:j \neq i}^{p}\left( \sum_{t=1}^{n}  \left( X_{i,j}^t - \E\left( X_{i,j}^t\right)  \right)\right) \right|\nonumber\\
							\leq&\frac{C_{3}}{np} \left( \sqrt{np\log(p)}+\sqrt{n\log(np)}+\log\left(n\right) \log\log\left(n\right) \log\left(np\right)\right) \nonumber\\
							<& 3C_{3}\sqrt{\frac{\log(n)\log(p)}{np}}.\nonumber
						\end{align}
						Then, for any $\by\in \mathbb{R}^p$ with $\|\by\|_{\infty}<\kappa_0$, we have that
						\begin{equation}\label{thm4_e5}
							\left\|\bG'\left( \btheta_{(0)}^* \right)^{-1}\bG\left( \btheta_{(0)}^* \right)\right\|_{\infty}\leq\left\|\bG'\left( \btheta_{(0)}^* \right)^{-1}\right\|_{\infty}\left\| \bG\left( \btheta_{(0)}^* \right)\right\|_{\infty}
							\leq 2C_{2}C_{3}e^{2\kappa_{0}}\sqrt{\frac{\log(n)\log(p)}{np}}.
						\end{equation}
						There exists a  big enough constant $C_{4}>0$, we have, for every $x,y\in \bB_{\infty}\left(0,\kappa \right)$, 
						\begin{eqnarray}\label{thm4_e6}
							&&\left\|\bG'\left( \bx \right)-\bG'\left( \by \right)\right\|_{\infty}\\
							&\leq& \max_{i}\left|\frac{1}{np}\sum_{t=1}^{n} \sum_{j=1,\:j \neq i}^{p} \left( X_{i,j}^t - \frac{ e^{x_{i}+x_{j}} }{1+ e^{x_{i}+x_{j}}}\right) - \frac{1}{np}\sum_{t=1}^{n} \sum_{j=1,\:j \neq i}^{p} \left( X_{i,j}^t - \frac{ e^{y_{i}+y_{j}} }{1+ e^{y_{i}+y_{j}}}\right) \right|\nonumber\\
							&\leq&\frac{1}{p}\max_{i}\left|\sum_{j=1,\:j \neq i}^{p}\frac{ e^{x_{i}+x_{j}} }{1+ e^{x_{i}+x_{j}}}-\frac{ e^{y_{i}+y_{j}} }{1+ e^{y_{i}+y_{j}}} \right|\nonumber\\
							&\leq&\frac{1}{p}\max_{i}\left|\sum_{j=1,\:j \neq i}^{p}\frac{ e^{z_{i}+z_{j}} }{1+ e^{z_{i}+z_{j}}}\left(x_{i}+x_{j} -y_{i}-y_{j}\right)  \right|\nonumber\\
							&\leq& C_{4}\left\|\bx-\by \right\|_{\infty}\nonumber
						\end{eqnarray}
						where $z_{i,j}:=\left( 1-c_{i,j}\right) \left( x_i+x_j\right) + c_{i,j}\left( y_i+y_j \right) $ with a series of constants $c_{i,j}\in (0,1)$.
						Combining the inequalities \eqref{thm4_e4}, \eqref{thm4_e5} and \eqref{thm4_e6}, we finish the proof of  \eqref{MME_con0}.
					\end{proof}
					\subsection{Proof of   \eqref{MME_con}}\label{p_WMME_con}
					\begin{proof}
						For brevity, $\bF\left(  \bx;\tilde{\btheta}_{(0)} \right)$ is denoted by $\bF\left(  \bx \right)$. Moreover, as all the conclusions hold uniformly for all $\tilde{\btheta}_{(0)}$, the argument ``uniformly for all $\tilde{\btheta}_{(0)}$” are also omitted in what follows.
						
						Note  that $\bF'\left( \bx \right)$ is a balanced symmetric matrix s.t.
						\begin{equation*}
							\bF'\left( \bx \right)_{i,j}  = \frac{1}{p}\frac{\tilde{\alpha}_{0,i,j}}{1+\tilde{\alpha}_{0,i,j}} \frac{e^{x_i+x_j}}{\left( 1+\tilde{\alpha}_{0,i,j}+e^{x_i+x_j}\right)^2 },\quad 
							\bF'\left( \bx \right)_{i,i}  =\sum_{j=1,\:j \neq i}^{p} \bF'\left( \bx \right)_{i,j},
						\end{equation*}
						for $i=1,\cdots,p$ and $i\neq j$.
						Following \cite{yan2012approximating}, we construct a matrix $\bS = (S_{i,j})$ to approximate the inverse of $ \bF'\left( \bx \right) $. Specifically, for all $i\neq j$, we set
						\begin{align*}
							S_{i,j}&= -\frac{1}{\calT},\quad 			S_{i,i}= \left( \sum_{j=1,\:j \neq i}^{p}\frac{1}{p}\frac{\tilde{\alpha}_{0,i,j}}{1+\tilde{\alpha}_{0,i,j}} \frac{e^{x_i+x_j}}{\left( 1+\tilde{\alpha}_{0,i,j}+e^{x_i+x_j}\right)^2 }\right)^{-1}  -\frac{1}{\calT},\\
							\calT &=  \sum_{1\leq i\neq j\leq p}  \frac{1}{p}\frac{\tilde{\alpha}_{0,i,j}}{1+\tilde{\alpha}_{0,i,j}} \frac{e^{x_i+x_j}}{\left( 1+\tilde{\alpha}_{0,i,j}+e^{x_i+x_j}\right)^2 }.
						\end{align*}
						Note that there exists a big enough constant $C_{1}>0$ s.t., for any $i \neq j$,
						\begin{equation*}
							\frac{C_{1}}{pe^{4\kappa_{0}+2\kappa_{1}}}\leq \frac{1}{p}\frac{\tilde{\alpha}_{0,i,j}}{1+\tilde{\alpha}_{0,i,j}} \frac{e^{x_i+x_j}}{\left( 1+\tilde{\alpha}_{0,i,j}+e^{x_i+x_j}\right)^2 } <\frac{1}{4p},
						\end{equation*}
						we have
						\begin{equation}\label{thm4_T}
							\calT =  \sum_{1\leq i\neq j\leq p} \frac{1}{p}\frac{\tilde{\alpha}_{0,i,j}}{1+\tilde{\alpha}_{0,i,j}} \frac{e^{x_i+x_j}}{\left( 1+\tilde{\alpha}_{0,i,j}+e^{x_i+x_j}\right)^2 }\quad \in \left(C_1pe^{-4\kappa_{0}-2\kappa_{1}}, \frac{p}{4}\right) .
						\end{equation}
						By Theorem 1 in \cite{yan2012approximating}, with $m= C_{1}/\left( pe^{4\kappa_{0}+2\kappa_{1}}\right) $ and $M=C_{1}/\left( 4p\right) $, there exists big enough constant $C_{2}>0$,  we have that 		
						\begin{eqnarray*}
							\left\|\bF'\left( \bx \right)^{-1}- \bS\right\|_{\max}&\leq& \frac{M}{m^2} \frac{pM+(p-2)m}{2(p-2)m}  \frac{1}{(p-1)^2} + \frac{1}{2m(p-1)^2}\\
							&=&\frac{1}{2m(p-1)^2}\left(1+ \frac{M}{m} +\frac{M^2}{m^2} \frac{p}{p-2}\right). \\
							&\leq &C_{2}\frac{e^{12\kappa_{0}+6\kappa_{1}}}{p^2},
						\end{eqnarray*}
						where $\|\bA\|_{\max}=\max_{i,j}A_{i,j}$.  Then we have that
						\begin{eqnarray}\label{thm4_e1}
							&&\left\|  \bF'\left( \bx \right)^{-1}\right\|_{\infty}\\
							&\leq&   \left\|\bS\right\|_{\infty}+ \left\|  \bF'\left( \bx \right)^{-1}- \bS\right\|_{\infty}\nonumber\\
							&\leq&\max_{i}\left( \sum_{j=1,\:j \neq i}^{p}\frac{1}{p}\frac{\tilde{\alpha}_{0,i,j}}{1+\tilde{\alpha}_{0,i,j}} \frac{e^{x_i+x_j}}{\left( 1+\tilde{\alpha}_{0,i,j}+e^{x_i+x_j}\right)^2 }\right)^{-1}  +\frac{p}{\calT} + p\left\|\bF'\left( \bx \right)^{-1}- \bS\right\|_{\max}\nonumber\\
							&<&2C_{2}e^{12\kappa_{0}+6\kappa_{1}}\nonumber.
						\end{eqnarray}
						Moreover, we have that
						\begin{eqnarray*}
							&&\left\|\bF\left( \btheta_{(1)}^* \right)\right\|_{\infty} \\
							&=&\max_{i}\left|-\frac{1}{np} \sum_{t=1}^{n} \sum_{j=1,\:j \neq i}^{p} \left\{X_{i,j}^tX_{i,j}^{t-1} -  \frac{\tilde{\alpha}_{0,i,j}}{1+\tilde{\alpha}_{0,i,j}}\left(1- \frac{1}{1+\tilde{\alpha}_{0,i,j}+\alpha_{1,i,j}^*}\right)\right\}\right|\\
							&\leq&\max_{i}\left|\frac{1}{np}\sum_{t=1}^{n} \sum_{j=1,\:j \neq i}^{p} \left\{X_{i,j}^tX_{i,j}^{t-1} -  \frac{\alpha_{0,i,j}^*}{1+\alpha_{0,i,j}^*}\left(1- \frac{1}{1+\alpha_{0,i,j}^*+\alpha_{1,i,j}^*}\right)\right\}\right|\\
							&&+\max_{i}\left|\frac{1}{np} \sum_{t=1}^{n} \sum_{j=1,\:j \neq i}^{p} \left\{\frac{\alpha_{0,i,j}^*}{1+\alpha_{0,i,j}^*} \frac{\alpha_{0,i,j}^*+\alpha_{1,i,j}^*}{1+\alpha_{0,i,j}^*+\alpha_{1,i,j}^*} -  \frac{\tilde{\alpha}_{0,i,j}}{1+\tilde{\alpha}_{0,i,j}} \frac{\tilde{\alpha}_{0,i,j}+\alpha_{1,i,j}^*}{1+\tilde{\alpha}_{0,i,j}+\alpha_{1,i,j}^*}\right\}\right|\\
							&=&L_{1}+L_{2}.
						\end{eqnarray*}
						By Lemma \ref{mixing} and Lemma \ref{Bers}, there exist  big positive constants $C_{3}, c_2$ s.t., with probability greater than $1-(np)^{-c_{2}}$, 
						\begin{eqnarray*}
							L_{1}&=&\frac{1}{np}\max_{i}\left|  \sum_{j=1,\:j \neq i}^{p} \sum_{t=1}^{n} X_{i,j}^tX_{i,j}^{t-1} - \E\left( X_{i,j}^t X_{i,j}^{t-1}\right)  \right|\\
							&\leq &\frac{1}{np}\max_{i}\left|\sum_{j=1,\:j \neq i}^{p}\left\{ \sum_{t=1}^{n}   X_{i,j}^tX_{i,j}^{t-1} - \E\left(X_{i,j}^tX_{i,j}^{t-1} \right)  \right\} \right|\nonumber\\
							&\leq&\frac{C_{3}}{np} \left( \sqrt{np\log(p)}+\sqrt{n\log(np)}+\log\left(n\right) \log\log\left(n\right) \log\left(np\right)\right) \nonumber\\
							&<& 3C_{3}\sqrt{\frac{\log(n)\log(p)}{np}}.
						\end{eqnarray*}
						Moreover, we have
						\begin{eqnarray*}
							&&L_{2}\\
							&=&\frac{1}{np}\max_{i}\Bigg| \sum_{t=1}^{n} \sum_{j=1,\:j \neq i}^{p} \Bigg\{\frac{\alpha_{0,i,j}^*}{1+\alpha_{0,i,j}^*}\left(1- \frac{1}{1+\alpha_{0,i,j}^*+\alpha_{1,i,j}^*}\right)\\
							&& -  \frac{\tilde{\alpha}_{0,i,j}}{1+\tilde{\alpha}_{0,i,j}}\left(1- \frac{1}{1+\tilde{\alpha}_{0,i,j}+\alpha_{1,i,j}^*}\right)\Bigg\}\Bigg|\\
							&\leq& \max_{i,j,\beta_{\xi}}\Bigg|  \left( \beta_{i,0}^*+\beta_{j,0}^* -\tilde{\beta}_{i,0}-\tilde{\beta}_{j,0}\right)  \frac{e^{\beta_{\xi,i,j}}}{\left( 1+e^{\beta_{\xi,i,j}}\right) ^2}\\
							&&+\frac{e^{\beta_{\xi,i,j}}\left( 1+e^{\beta_{\xi,i,j}}\right) \left( 1+e^{\beta_{\xi,i,j}}+\alpha_{1,i,j}^*\right)-e^{\beta_{\xi,i,j}}\left( \left( 2+\alpha_{1,i,j}^*\right) e^{\beta_{\xi,i,j}}+2e^{2\beta_{\xi,i,j}}\right) }{\left( 1+e^{\beta_{\xi,i,j}}\right) ^2\left( 1+e^{\beta_{\xi,i,j}}+\alpha_{1,i,j}^*\right)^2}\Bigg|\\
							&\leq& \max_{i,j,\beta_{\xi}}\Bigg|  \left( \beta_{i,0}^*+\beta_{j,0}^* -\tilde{\beta}_{i,0}-\tilde{\beta}_{j,0}\right) \left(  \frac{e^{\beta_{\xi,i,j}}} {\left( 1+e^{\beta_{\xi,i,j}}\right) ^2} +\frac{e^{\beta_{\xi,i,j}}\left( 1+\alpha_{1,i,j}^*-e^{2\beta_{\xi,i,j}} \right) }{\left( 1+e^{\beta_{\xi,i,j}}\right) ^2\left( 1+e^{\beta_{\xi,i,j}}+\alpha_{1,i,j}^*\right)^2}\right) \Bigg|\\
							&\leq& 2 	 \left\|\tilde{\btheta}_{(0)}-\btheta_{(0)}^*\right\|_{\infty},
						\end{eqnarray*}
						where $\beta_{\xi,i,j}:=\left( 1-c_{i,j}\right) \left( \beta_{i,0}^*+\beta_{j,0}^* \right) + c_{i,j}\left( \tilde{\beta}_{i,0}+\tilde{\beta}_{j,0} \right) $ with a series of constants $c_{i,j}\in (0,1)$. Then, by equation \eqref{MME_b1}, there exists a big enough constant $C_4>0$, we have, with probability tending to 1,
						\begin{eqnarray*}
							L_{2}	&\leq 2 \left\|\tilde{\btheta}_{(0)}-\btheta_{(0)}^*\right\|_{\infty}	    \leq& C_{4} \sqrt{\frac{\log(n)\log(p)e^{4\kappa_{0}}}{np}}.
						\end{eqnarray*}
						We can conclude that, with probability tending to 1,
						\begin{eqnarray*}
							\left|\bF\left( \btheta_{(1)}^* \right)\right|_{\infty}&\leq& \frac{1}{np} \left( L_{1}+L_{2}\right)\\
							&\leq&3C_{3}\sqrt{\frac{\log(n)\log(p)}{np}}+C_{4}\sqrt{\frac{\log(n)\log(p)e^{4\kappa_{0}}}{np}}\\
							&\leq&C_{5}\sqrt{\frac{\log(n)\log(p)e^{4\kappa_{0}}}{np}},
						\end{eqnarray*}
						hold uniformly for any $i$ at the same time.
						Consequently, we have
						\begin{eqnarray}\label{thm4_e2}
							\left\|\bF'\left(  \btheta_{(1)}^* \right)^{-1}\bF\left(  \btheta_{(1)}^* \right)\right\|_{\infty}&\leq&\left\|\bF'\left(  \btheta_{(1)}^* \right)^{-1}\right\|\left\|\bF\left(  \btheta_{(1)}^* \right)\right\|_{\infty}\\
							&\leq&2C_{2}C_{5}e^{12\kappa_{0}+6\kappa_{1}}\sqrt{\frac{\log(n)\log(p)e^{4\kappa_{0}}}{np}}.\nonumber
						\end{eqnarray}
						There exists an big enough constant $C_{6}>0$, we have, for every $x,y\in \bB_{\infty}\left(0,\kappa_{1} \right)$, 
						\begin{align}\label{thm4_e3}
							&\left\|\bF'\left( \bx \right)-\bF'\left( \by \right)\right\|_{\infty}\\
							\leq& \max_{i}\left|2\sum_{j=1,\:j \neq i}^{p}\left( \frac{1}{p}\frac{\tilde{\alpha}_{0,i,j}} {1+\tilde{\alpha}_{0,i,j}} \frac{e^{x_i+x_j}}{\left( 1+\tilde{\alpha}_{0,i,j}+e^{x_i+x_j}\right)^2 }-\frac{1}{p}\frac{\tilde{\alpha}_{0,i,j}}{1+\tilde{\alpha}_{0,i,j}} \frac{e^{y_i+y_j}}{\left( 1+\tilde{\alpha}_{0,i,j}+e^{y_i+y_j}\right)^2 }\right) \right|\nonumber\\
							=& \frac{2}{p}\max_{i}\left|\sum_{j=1,\:j \neq i}^{p} \frac{\tilde{\alpha}_{0,i,j}}{1+\tilde{\alpha}_{0,i,j}}\left(\frac{e^{x_i+x_j}}{\left( 1+\tilde{\alpha}_{0,i,j}+e^{x_i+x_j}\right)^2}- \frac{e^{y_i+y_j}}{\left( 1+\tilde{\alpha}_{0,i,j}+e^{y_i+y_j}\right)^2 }\right) \right|\nonumber\\
							=& \frac{1}{p}\max_{i}\left|\sum_{j=1,\:j \neq i}^{p} \frac{\tilde{\alpha}_{0,i,j}}{1+\tilde{\alpha}_{0,i,j}} \frac{e^{z_{i,j}}}{\left( 1+\tilde{\alpha}_{0,i,j}+e^{z_{i,j}}\right)^3 }\left|x_i+x_j -y_i-y_j  \right|\right|\nonumber\\
							\leq& C_{6} \left\|\bx-\by \right\|_{\infty},\nonumber
						\end{align}
						where $z_{i,j}:=\left( 1-d_{i,j}\right) \left( x_i+x_j\right) + d_{i,j}\left( y_i+y_j \right) $ with a series of constants $d_{i,j}\in (0,1)$.
						Combining the inequalities \eqref{thm4_e1}, \eqref{thm4_e2} and \eqref{thm4_e3}, we finish the proof of \eqref{MME_con}.
					\end{proof}
						\subsection{Proof of  Corollary \ref{cor1}}
						\begin{proof}
							When the condition $-\kappa_{0}\leq \beta_{i,0}\leq C_{\kappa}$ and $\|\bbeta_{1}\|_{\infty}\le \kappa_{1} $,  where $C_{\kappa}=O(1)$, is satisfied, there exists a constant $C_{1}>0$ 
							such that, for all $1\leq i\neq j\leq p$ and $\btheta \in \bB_{\infty}\left( \btheta^*,r \right)$ with $r = c_r e^{-2\kappa_{0}-4\kappa_{1}}$ for  a small enough constant $c_r>0$, it holds
							\begin{align*}
								\E(\bV_2(\btheta)+\bV_1(\btheta))_{i,j}\geq C_{1} \frac{ e^{-4\kappa_{1} }}{p},\quad\E(\bV_2(\btheta)+\bV_3(\btheta))_{i,j}\geq C_{1} \frac{ e^{-4\kappa_{1}- 2\kappa_{0}}}{p},
							\end{align*}
							and
							\begin{align*}
								\E(\bV_2(\btheta))_{i,j}\geq C_{1} \frac{ e^{-4\kappa_{1}- 2\kappa_{0}}}{p}.
							\end{align*}
							Then there exists a constant $C_{2}>0$, for all $\btheta \in \bB_{\infty}\left( \btheta^*,r \right)$ such that we have
							\begin{align*}
								\left\|\E(\bV(\btheta))\right\|_{2}\geq C_{2}e^{-4\kappa_{1}- 2\kappa_{0}}.
							\end{align*}
							With similar arguments as in the proofs of Lemma \ref{bound_l2} and  Theorems \ref{thm1} and \ref{thm2}, we can prove Corollary \ref{cor1}. 
						\end{proof}
						\subsection{Proof of  Corollary \ref{cor2}}
						\begin{proof}
							The first inequality can be proved using the results in  Corollary \ref{cor1},  along with analogous reasoning employed in the proof of Theorem \ref{thm1}.			Replace equations \eqref{thm3_con} and \eqref{thm3_s} by
							\begin{eqnarray*}
								r_s&=&  e^{4\kappa_{0}+8\kappa_{1}}\log\log(np)\frac{ \sqrt{\log(np)} }{\left( np\right) ^{s}}\left(1+  \frac{\log(np)}{\sqrt{p}} \right);\\
								s_{0}&= &\frac{6 \kappa_{0}+12\kappa_{1} +\log\log(np)/2 +\log\log\log(np)+\log\left(1+  \frac{\log(np)}{\sqrt{p}} \right)-\log(c_{r})}{\log(np)}.
							\end{eqnarray*}
							Similar to the proof of Theorem \ref{thm3}, let $z_{1}=0.5e^{-2\kappa_{0}-4\kappa_{1}}[\log\log(np)]^{-1}$ $(np)^{s-1/2}$ and $z_{2}=0.5e^{2\kappa_{0}+4\kappa_{1}}$ $ \log\log(np) (np)^{1/2-s}$. We can assert the existence of a sufficiently large constant  $C_{1}>0$ such that  with probability greater than $1-(np)^{-c_2}$,
							\begin{eqnarray*}
								&&
								\left|l\left(\btheta_{(i)}^*, \hat{\btheta}^{(s)}_{(-i)}\right) -l\left(\btheta_{(i)},\hat{\btheta}^{(s)}_{(-i)} \right)  - \left[ l\left(\btheta_{(i)}^{*},\btheta_{(-i)}^{*}\right) - l\left(\btheta_{(i)},\btheta_{(-i)}^{*}\right) \right]\right|\\
								&\leq& C_{1} \frac{e^{6\kappa_{0}+12\kappa_{1}} \log\log(np)\log(np)}{(np)^{s+1/2}}\left(1+  \frac{\log(np)}{\sqrt{p}} \right)^2\nonumber
							\end{eqnarray*}
							holds uniformly for all $s\in[s_0,1/2]$ and $\btheta_{(i)}\in \bB_{\infty}\left( \btheta^*_{(i)},r_s\right)$. Following similar steps, we can ascertain that as $np\rightarrow \infty$ with $n\ge 2$, with probability converging  to one it holds that
							\begin{equation*}
								\left\|\hat{\btheta}-\btheta^*\right\|_{\infty}\leq C e^{4\kappa_{0}+8\kappa_{1}} \log\log(np) \sqrt{\frac{\log(np)}{np}}\left(1+ \frac{\log(np)}{\sqrt{p}} \right).
							\end{equation*}
						\end{proof}
						\subsection{Proof of  Corollary \ref{cor3}}
						\begin{proof}
							There exist positive constants $C_{1}$ and $C_{2}$ such that
							\begin{equation*}
								\sum_{1\leq i\neq j\leq p} \frac{1}{p}\frac{\tilde{\alpha}_{0,i,j}}{1+\tilde{\alpha}_{0,i,j}} \frac{e^{x_i+x_j}}{\left( 1+\tilde{\alpha}_{0,i,j}+e^{x_i+x_j}\right)^2 }\quad \in \left(C_1pe^{-2\kappa_{0} -2\kappa_{1}}, \frac{p}{4}\right),
							\end{equation*}
							for all $\|\bx\|_{\infty}\leq \kappa_{1}$. Then the inequalities in \eqref{MME_con} can be updated to
							\begin{eqnarray*}
								&&\left\|\bF'\left( \bx ;\tilde{\btheta}_{(0)}\right)-\bF'\left( \by ;\tilde{\btheta}_{(0)}\right)\right\|_{\infty}\leq C_{2} \left\|\bx-\by \right\|_{\infty},\\
								&&\left\|\bF'\left( \btheta_{(1)}^*;\tilde{\btheta}_{(0)} \right)^{-1}\right\|_{\infty}\leq C_{2} e^{2\kappa_{0}+6\kappa_{1}},\nonumber\\ 
								&&\left\|\bF'\left(  \btheta_{(1)}^*;\tilde{\btheta}_{(0)} \right)^{-1}\bF\left(  \btheta_{(1)}^* ;\tilde{\btheta}_{(0)}\right)\right\|_{\infty}\leq C_{2} e^{4\kappa_{0}+6\kappa_{1}}\sqrt{\frac{\log(n)\log(p)}{np}},\nonumber
							\end{eqnarray*}
							for all  $\|\bx\|_{\infty}\leq \kappa_{1}$, $\|\by\|_{\infty}\leq \kappa_{1}$. 
							Consequently,  with Lemma \ref{Lemma_IMT}, similar to the proof of Theorem \ref{thm4},  we can prove Corollary \ref{cor3}.
						\end{proof}
					
					\subsection{Proof of the Theorem \ref{thm5}}
					\begin{proof}
						To simplify the notations, in what follows we shall denote $l_{i,j}(\theta_i,\theta_j)$ as $l_{i,j}(\btheta)$ instead.	Let $K=\lfloor c_0\log(np)\rfloor+1 $  for some big enough constant $c_0>0$, where $\lfloor\cdot\rfloor$ is the smallest integer function. By the Taylor expansion with Lagrange remainder we have, for any $\btheta\in \bB_\infty(\btheta', \alpha_0)$, there exists a $\btheta^{\xi}\in \bB_\infty(\btheta', \alpha_0)$ dependent on $\btheta$ s.t.,
						\begin{eqnarray*}
							&&\bL\left( \btheta\right)-\bL\left( \btheta'\right)\\
							&=&\sum_{I_1=1}^{p} \left( \theta_{I_1}-\theta_{I_1}'\right) \frac{\partial \bL\left( \btheta'\right) }{\partial \theta_{I_1}}\\
							&&+\frac{1}{2!}\sum_{I_1=1}^{p}\sum_{I_2=1}^{p} \left( \theta_{I_1}-\theta_{I_1}'\right)\left( \theta_{I_2}-\theta_{I_2}' \right)   \frac{\partial^2 \bL\left( \btheta'\right) }{\partial \theta_{I_2}\partial \theta_{I_2}}\\
							&&+\cdots\\
							&&+ \frac{1}{K!}\sum_{I_1=1}^{p}\sum_{I_2=1}^{p}\cdots\sum_{I_K=1}^{p} \left( \prod_{\ell=1}^{K}\left( \theta_{I_\ell}-\theta_{I_\ell}'\right) \right)   \frac{\partial^K \bL\left( \btheta'\right) }{\partial \theta_{I_1}\partial \theta_{I_2}\cdots\partial \theta_{I_K}}\\
							&&+ \frac{1}{\left( K+1\right) !}\sum_{I_1=1}^{p}\sum_{I_2=1}^{p}\cdots\sum_{I_{K+1}=1}^{p} \left( \prod_{\ell=1}^{K+1}\left( \theta_{I_\ell}-\theta_{I_\ell}'\right) \right)   \frac{\partial^{K+1} \bL\left( \btheta^{\xi}\right) }{\partial \theta_{I_1}\partial \theta_{I_2}\cdots\partial \theta_{I_{K+1}}}\\
							&=&\frac{1}{p}\sum_{k=1}^{K}\left\{\frac{1}{k!} \sum_{1\leq i\neq j\leq p  }\left( Y_{i,j} \sum_{s=0}^{k} \tbinom{k}{s}\left( \theta_{i}-\theta_{i}'\right)^{s}\left( \theta_{j}-\theta_{j}' \right)^{k-s} \frac{\partial^k l_{i,j}\left( \btheta'\right) }{\partial \theta_{i}^{s} \partial \theta_{j}^{k-s} }   \right)   \right\}\\
							&&+\frac{1}{p}\frac{1}{\left( K+1\right) !} \sum_{1\leq i\neq j\leq p  }\left( Y_{i,j} \sum_{s=0}^{K+1} \tbinom{K+1}{s}\left( \theta_{i}-\theta_{i}'\right)^{s}\left( \theta_{j}-\theta_{j}' \right)^{K+1-s} \frac{\partial^{K+1} l_{i,j}\left( \btheta^{\xi}\right) }{ \partial \theta_{i}^{s} \partial \theta_{j}^{K+1-s} }   \right)\\
							&\leq&\left|\sum_{k=1}^{K} \frac{1}{p}\frac{1}{k!} \sum_{1\leq i\neq j\leq p  } Y_{i,j}\left(  \left( \theta_{i}-\theta_{i}'\right)^{k}\frac{\partial^k l_{i,j}\left( \btheta'\right) }{\partial \theta_{i}^{k} } + \left( \theta_{j}-\theta_{j}' \right)^{k} \frac{\partial^k l_{i,j}\left( \btheta'\right) }{ \partial \theta_{j}^{k}} \right)      \right|\\
							&&+ \left|\sum_{k=2}^{K}\frac{1}{p}\frac{1}{k!} \sum_{1\leq i\neq j\leq p  }\left( Y_{i,j} \sum_{s=1}^{k-1} \tbinom{k}{s}\left( \theta_{i}-\theta_{i}'\right)^{s}\left( \theta_{j}-\theta_{j}' \right)^{k-s} \frac{\partial^k l_{i,j}\left( \btheta'\right) }{\partial \theta_{i}^{s} \partial \theta_{j}^{k-s} }   \right)   \right|\\
							&&+\frac{1}{p}\left|\frac{1}{\left( K+1\right) !} \sum_{1\leq i\neq j\leq p  }\left( Y_{i,j} \sum_{s=0}^{K+1} \tbinom{K+1}{s}\left( \theta_{i}-\theta_{i}'\right)^{s}\left( \theta_{j}-\theta_{j}' \right)^{K+1-s} \frac{\partial^{K+1} l_{i,j}\left( \btheta^{\xi}\right) }{ \partial \theta_{i}^{s} \partial \theta_{j}^{K+1-s} }   \right)\right|\\
							&=&S^{(1)} + S^{(2)} + S^{(3)}.
						\end{eqnarray*}
						We first consider $S^{(1)} $. By  Lemma \ref{Bers}, there exist  big enough constants $C_1>0, c_1>0$ such that, uniformly for all  $\btheta\in\bB_{\infty}\left(\btheta',\alpha_{0} \right)$  and all $k=1,\cdots,K$, we have,   with probability greater than $1-(np)^{-c_1}$,
						\begin{eqnarray*}
							&&\frac{1}{p}\frac{1}{k!}\left| \sum_{1\leq i\neq j\leq p  } Y_{i,j}\left(  \left( \theta_{i}-\theta_{i}'\right)^{k}\frac{\partial^k l_{i,j}\left( \btheta'\right) }{\partial \theta_{i}^{k} }+ \left( \theta_{j}-\theta_{j}' \right)^{k}\frac{\partial^k l_{i,j}\left( \btheta'\right) }{\partial \theta_{j}^{k} } \right)    \right|\\
							&\leq &\frac{2}{p}\frac{1}{k!}  \sum_{i=1}^{p}  \left| \frac{\theta_i-\theta_i'}{\alpha}\right|^{k} \left| \sum_{j\neq i,j=1}^{p} \alpha^kY_{i,j}\frac{\partial^k l_{i,j}\left( \btheta'\right) }{\partial \theta_{j}^{k} } \right| \\
							&\leq &\frac{2}{p}\frac{1}{k!} \max_{i}\left|\sum_{j\neq i,j=1}^{p} \alpha^k Y_{i,j} \frac{\partial^k l_{i,j}\left( \btheta'\right) }{\partial \theta_{j}^{k} }  \right|\sum_{i=1}^{p}  \left| \frac{\theta_i-\theta_i'}{\alpha}\right|^{k}\\
							&\leq& \frac{C_1}{p}\frac{1}{k!}  \left(b_{(p)}\log(np)+ \sigma_{(p)}\sqrt{p\log(np)} \right)\left(k-1 \right)!  \left\|\frac{\btheta-\btheta'} {\alpha}\right\|_{k}^{k}\\
							&=&C_1 \left\|\frac{\btheta-\btheta'} {\alpha}\right\|_{k}^{k} \frac{b_{(p)} \log(np)+ \sigma_{(p)}\sqrt{p\log(np)}}{kp}.
						\end{eqnarray*}
						Consequently, with probability greater than $1-(np)^{-c_1}$ we have:
						\begin{eqnarray*}
							S^{(1)}&\leq&\sum_{k=1}^{K}\frac{1}{p}\frac{1}{k!}\left| \sum_{1\leq i\neq j\leq p  } Y_{i,j}\left(  \left( \theta_{i}-\theta_{i}'\right)^{k}\frac{\partial^k l_{i,j}\left( \btheta'\right) }{\partial \theta_{i}^{k} }+ \left( \theta_{j}-\theta_{j}' \right)^{k}\frac{\partial^k l_{i,j}\left( \btheta'\right) }{\partial \theta_{j}^{k} } \right)    \right|\\
							&\leq&C_1 \sum_{k=1}^{K}\left\|\frac{\btheta-\btheta'} {\alpha}\right\|_{k}^{k}\frac{b_{(p)} \log(np)+ \sigma_{(p)}\sqrt{p\log(np)}}{kp}\\
							&\leq& C_1 \frac{b_{(p)} \log(np)+ \sigma_{(p)}\sqrt{p\log(np)}}{p} \sum_{k=1}^{K}\frac{1}{k}\left\|\frac{\btheta-\btheta'} {\alpha}\right\|_{k}^{k},
						\end{eqnarray*}
						holds uniformly for all  $\btheta\in\bB_{\infty}\left(\btheta',\alpha_{0} \right)$.
						Note that for any $\bx=\left( x_1,\cdots,x_{p}\right)^{\top} $ s.t. $\|\bx\|_{\infty}\leq a<1$ for some constant $a>0$, we have
						\begin{equation}\label{thm5_1}
							\sum_{k=1}^{K}\frac{1}{k}\left\|\bx \right\|_{k}^{k}=\sum_{k=1}^{K}  \sum_{i=1}^{p} \frac{1}{k}x_{i}^{k}=  \sum_{i=1}^{p}x_{i}\sum_{k=1}^{K}   \frac{x_{i}^{k-1}}{k}
							\leq \|\bx\|_{1}\sum_{k=1}^{K} \frac{a^{k-1}}{k} \leq \left\|\bx \right\|_{1}\left(-\frac{\log(1-a)}{a} \right).
						\end{equation}
						Here in last step we have used $-\log(1-a)=\sum_{k=1}^{\infty}a^{k}/k$.
						With the fact that $\left\|\frac{\btheta-\btheta'}{\alpha}\right\|_{\infty}<\alpha_0/\alpha<1/2$ holds for any  $\btheta\in\bB_{\infty}\left(\btheta',\alpha_{0} \right)$, we have, with probability greater than $1-(np)^{-c_1}$,
						\begin{eqnarray*}
							S^{(1)}	&\leq &-\frac{\log(1/2)}{1/2}  C_1\frac{b_{(p)} \log(np)+ \sigma_{(p)}\sqrt{p\log(np)}}{p}\left\|\btheta-\btheta'\right\|_{1}\\
							&\leq &2  C_1\frac{b_{(p)} \log(np)+ \sigma_{(p)}\sqrt{p\log(np)}}{p}\left\|\btheta-\btheta'\right\|_{1}.
						\end{eqnarray*}
						
						Next we derive an upper bound for $S^{(2)}$. Define a series of random $p\times p$ matrices $\{\bY^{s}_{k}: k=2,\ldots, K, s=1,\cdots,k-1\}$ with the $(i,j)$-th elements of $\bY^{s}_{k}$ given as
						\begin{equation*}
							\left( \bY^{s}_{k}\right) _{i,j} = Y_{i,j}\alpha^{k} \frac{\partial^k l_{i,j}\left( \btheta'\right) }{\partial \theta_{i}^{s}\partial \theta_{j}^{k-s}},\quad 1\leq i\neq j \leq p.
						\end{equation*}
						Further for any $k=2,\ldots, K$, define a random $(k-1)p\times (k-1)p$ matrix $\bW_{k}$ as: 
						\begin{equation*}
							\bW_{k} = \left[
							\begin{array}{ccccc}
								0 &   &   &   &\bY^{1}_{k}\\
								&   &   &\bY^{2}_{k}&   \\
								&   &\iddots&   &    \\
								&\bY^{k-2}_{k}&   &   & \\
								\bY^{k-1}_{k} &   &   &   &    0   \\
							\end{array}
							\right].
						\end{equation*}
						Also we define a series of $p\times 1$ vectors, $\{ \bz^{(s)}_{k} :  k=2,\ldots, K, s=1,\cdots,k-1\}$  with $i$-th element of $\bz^{(s)}_{k} $ given as:
						\begin{equation*}
							\left( \bz^{(s)}_{k}\right)_i =\left(0.5\alpha\right)^{-s} \left(\theta_i-\theta_i'\right)^{s}\sqrt{\tbinom{k}{s}}.
						\end{equation*}
						For any $k=2,\ldots, K$, by denoting $\tilde{\bz}_{k}$ as:
						\begin{equation*}
							\tilde{\bz}_{k} = \left[
							\begin{array}{c}
								\bz_{k}^{(1)} \\
								\vdots  \\
								\bz_{k}^{(k-1)}   \\
							\end{array}
							\right],
						\end{equation*}
						we have:
						\begin{eqnarray}\label{S2_bound}
							&&\frac{1}{p}\frac{1}{k!}\left| \sum_{1\leq i\neq j\leq p  }\left( Y_{i,j} \sum_{s=1}^{k-1} \tbinom{k}{s}\left( \theta_{i}-\theta_{i}'\right)^{s}\left( \theta_{j}-\theta_{j}' \right)^{k-s} \frac{\partial^k l_{i,j}\left( \btheta'\right) }{\partial \theta_{i}^{s}\partial \theta_{j}^{k-s}}   \right)   \right|\\
							&=&\frac{1}{p}\frac{1}{2^{k}k!}\left|\sum_{s=1}^{k-1}\left( \tilde{\bz}_{k}^{(k-s)}\right) ^{\top}\bY_{k}^{s}\tilde{\bz}_{k}^{(s)}\right|\nonumber\\
							&=&\frac{1}{p}\frac{1}{2^{k}k!}\left|\tilde{\bz}_{k}^{\top}\bW_{k}\tilde{\bz}_{k} \right| \nonumber\\
							&\leq& \frac{1}{p}\frac{1}{2^{k}k!}\|\tilde{\bz}_{k} \|_2^2\| \bW_{k} \|_2.\nonumber
						\end{eqnarray}
						We remark that by formulating the confounding terms in $S^{(2)}$ via $\{\bz^{(s)}_{k},\bW_{k} \}$, we have established in  \eqref{S2_bound} an upper bound that depends on the parameters in $\{\bz^{(s)}_{k},k=2,\ldots, K\}$ and the random matrices $\{\bW_k,2,\ldots,K\}$ separately. 	
						Using the fact that $\sum_{l=0}^{\infty}\tbinom{l+s}{s}0.5^{l+s+1}=1$, we have
						\begin{equation*}
							\sum_{k=s+1}^{K}\tbinom{k}{s}0.5^{k}< \sum_{l=0}^{\infty}\tbinom{l+s}{s}0.5^{l+s}= 2.
						\end{equation*}
						Consequently,  there exists a big enough constant $C_{2}>0$ such that,  for all $\btheta\in\bB_{\infty}\left(\btheta',\alpha_{0} \right)$, we have
						\begin{eqnarray}\label{z_bound}
							\sum_{k=2}^{K}\frac{0.5^{k}}{k}\|\tilde{\bz}_{k} \|_2^2&=&\sum_{k=2}^{K}\frac{0.5^{k}}{k}\sum_{i=1}^{p}\sum_{s=1}^{k-1}\left(0.5\alpha\right)^{-2s} \left(\theta_i-\theta_i'\right)^{2s}\tbinom{k}{s}\\
							&=&\sum_{k=2}^{K}\frac{0.5^{k}}{k}\sum_{s=1}^{k-1}\tbinom{k}{s}\left(0.5\alpha\right)^{-2s}\|\btheta-\btheta'\|_{2s}^{2s}\nonumber \\
							&\leq&\sum_{s=1}^{K-1}\frac{1}{s}\left\|\frac{\btheta-\btheta'} {\alpha/2}\right\|_{2s}^{2s}\left(  \sum_{k=s+1}^{K}\tbinom{k}{s}0.5^{k}\right)\nonumber\\
							&<&2\sum_{s=1}^{K-1}\frac{1}{s}\left\|\frac{\btheta-\btheta'} {\alpha/2}\right\|_{2s}^{2s} \nonumber\\
							&< &2\sum_{s=1}^{K-1}\frac{1}{s}\left\|\frac{\btheta-\btheta'} {\alpha/2}\right\|_{s}^{s} \nonumber\\
							&\leq&C_{2}\left\|\btheta-\btheta'\right\|_{1}.\nonumber
						\end{eqnarray}
						Here the last step follows from \eqref{thm5_1}.
						With similar arguments as in the proof of the matrix Bernstein inequality (c.f. Lemma \ref{BersM}), we can show that uniformly for all $k\leq K=\lfloor c_0\log(np)\rfloor+1 $,  there exist big enough constants $C_3>0$ and $c_2>0$, such that with probability greater than $1-(np)^{c_2}$, 
						\begin{equation}\label{W_bound}
							\| \bW_{k} \|_2 \leq C_3\left( k-1\right)! \left( b_{(p)} \log(np)+ \sigma_{(p)}\sqrt{p\log(np)}\right).
						\end{equation}
						For brevity, the proof of inequality \eqref{W_bound} is provided independently in Section \eqref{p_W_bound}. Consequently, combining \eqref{S2_bound}, \eqref{z_bound} and \eqref{W_bound} and  $K=\lfloor c_0\log(np)\rfloor+1$, we  conclude that  with probability greater than $1-(np)^{-c_2}$,
						\begin{eqnarray*}
							S^{(2)}&\leq&\sum_{k=2}^{K}\frac{1}{p}\frac{1}{k!}\left| \sum_{1\leq i\neq j\leq p  }\left( Y_{i,j} \sum_{s=1}^{k-1} \tbinom{k}{s}\left( \theta_{i}-\theta_{i}'\right)^{s}\left( \theta_{j}-\theta_{j}' \right)^{k-s} \frac{\partial^k l_{i,j}\left( \btheta'\right) }{\partial \theta_{i}^{k} }   \right)   \right|\\
							&\leq &\sum_{k=2}^{K} \frac{1}{p}\frac{1}{2^kk!}\|\tilde{\bz}_{k} \|_2^2\| \bW_{k} \|_2 \\
							&\leq &C_3\frac{b_{(p)} \log(np)+ \sigma_{(p)}\sqrt{p\log(np)}}{p}\sum_{k=2}^{K}  \frac{1}{k2^k}\|\tilde{\bz}_{k}\|_2^2\\
							&\leq& C_{2}C_3\frac{b_{(p)} \log(np)+ \sigma_{(p)}\sqrt{p\log(np)}}{p} \left\|\btheta-\btheta'\right\|_{1},
						\end{eqnarray*}
						uniformly for all  $\btheta\in\bB_{\infty}\left(\btheta',\alpha_{0} \right)$.	
						Finally, we derive an upper bound for $S^{(3)}$.	By condition (L-A1),  when $c_0$ is chosen to be big enough, there exists a big enough constant $c_3>1$, such that, uniformly for all $\btheta\in \bB_\infty(\btheta', \alpha_0)$ and $\btheta^{\xi}$, we have
						\begin{eqnarray*}
							&&S^{(3)}\\
							&=&\frac{1}{p}\left|\frac{1}{\left( K+1\right) !} \sum_{1\leq i\neq j\leq p  }\left( Y_{i,j} \sum_{s=0}^{K+1} \tbinom{K+1}{s}\left( \theta_{i}-\theta_{i}'\right)^{s}\left( \theta_{j}-\theta_{j}' \right)^{K+1-s} \frac{\partial^{K+1} l_{i,j}\left( \btheta^{\xi}\right) }{\partial \theta_{i}^{s} \partial \theta_{j}^{K+1-s} }   \right) \right|\\
							&\leq &\frac{1}{p}\frac{1}{K+1}\sum_{1\leq i\neq j\leq p} \left| Y_{i,j} \right|\left(\frac{\left| \theta_{i}-\theta_{i}' \right|}{\alpha} + \frac{\left| \theta_{j}-\theta_{j}' \right|}{\alpha}\right)^{K+1} \\
							&\leq & \frac{pb_{(p)}}{K+1}\left(\frac{2\left\| \btheta-\btheta' \right\|_{\infty}}{\alpha} \right)^{K+1}\\
							&\leq& \frac{pb_{(p)}}{K+1} \left( \frac{2\alpha_{0}}{\alpha}\right) ^{c_0\log\left(np \right) }\\
							&\leq& b_{(p)} \left( np\right)^{-c_3}.
						\end{eqnarray*}
						Here the last step will hold if we choose $c_0$ to be large enough such that $(2\alpha_{0}/\alpha)^{c_0/(c_3+1)}$$<e^{-1}$. When $np\rightarrow \infty$, and $c_0, c_3$ are chosen to be large enough,  this bound will be dominated by the upper bounds for $S^{(1)}$ and $S^{(2)}$.
						
						Combining the upper bound on $S^{(1)},S^{(2)}$ and $S^{(3)}$, we conclude that, for any given $\btheta'$, there exist large enough constants $C>0$, $c>0$ which are independent of $\btheta'$ such that for any $\btheta\in \bB_{\infty}\left( \btheta',\alpha_{0}\right)$ where $\alpha_{0}\in (0,\alpha/2)$ at the same time with probability greater than $1-(np)^{-c}$,
						\begin{equation*}
							\left|\bL\left( \btheta\right)-\bL\left( \btheta'\right)\right|
							\le 	C\frac{b_{(p)} \log(np)+ \sigma_{(p)}\sqrt{p\log(np)}}{p} \left\|\btheta-\btheta'\right\|_{1}.
						\end{equation*}
					\end{proof}
					\subsection{Proof of inequality  \eqref{W_bound}}\label{p_W_bound}
					\begin{proof}
						For any $k=2,\ldots, K$ and $1\leq i\neq j \leq p$, let  $\bW_{k,i,j}$ be defined by  keeping the $(i,j)$th element of all the $\bY^{s}_{k}, s=1, \cdots,k-1$  in $\bW_k$ unchanged, and setting all other   elements  to be  zero. Then the random matrices $\bW_{k,i,j}$, $1\leq i\neq j \leq p$   are independent, and
						\begin{equation*}
							\sum_{1\leq i\neq j\leq p } \bW_{k,i,j} = \bW_k.
						\end{equation*}
						For any $\ba\in \mathbb{R}^{(k-1)p}$, we can write it as
						\begin{equation*}
							\ba = \left[
							\begin{array}{c}
								\ba^{(1)} \\
								\vdots  \\
								\ba^{(k-1)}   \\
							\end{array}
							\right],
						\end{equation*}
						with $\ba^{(s)}=(a^{(s)}_1, \ldots, a^{(s)}_p)^\top\in \mathbb{R}^{p}$, $r=1,2,\cdots,k-1$.	Then   for any $k=2,\ldots, K$ and $1\leq i\neq j \leq p$, we have,
						\begin{eqnarray*}
							\|\bW_{k,i,j}\|_2 &=& \sup_{\|\ba\|_2=1} \|\bW_{k,i,j}\ba\|_2\\
							&=& \sup_{\|\ba\|_2=1} \left(  \sum_{s=1}^{k-1}\left( \alpha^kY_{i,j}\frac{\partial^k l_{i,j}\left( \btheta'\right) }{\partial \theta_{i}^{s}\partial \theta_{j}^{k-s}}  a^{(k- s)}_j\right) ^2  \right)^{1/2}\\
							&\leq& \max_{1\leq s\leq k}\left|\alpha^kY_{i,j}\frac{\partial^k l_{i,j}\left( \btheta'\right) }{\partial \theta_{i}^{s}\partial \theta_{j}^{k-s}} \right|\sup_{\|\ba\|_2=1} \left( \sum_{s=1}^{k-1}  \left( a^{(k- s)}_j\right) ^2 \right)^{1/2}\\
							&\leq& \left( k-1\right)!b_{(p)}.
						\end{eqnarray*}
						On the other hand,
						\begin{eqnarray*}
							&& \max\left\{\left\|\sum_{1\leq i\neq j\leq p }\E\left( \bW_{k,i,j}\bW_{k,i,j}^{\top}\right) \right\|_2 ,\left\|\sum_{1\leq i\neq j\leq p }\E\left( \bW_{k,i,j}^{\top}\bW_{k,i,j}\right) \right\|_2 \right\}	\\
							&=&\max\Bigg\{ \sup_{\|\ba\|_2=1}\left| \ba^{\top}\left( \sum_{1\leq i\neq j\leq p } \E\left( \bW_{k,i,j}\bW_{k,i,j}^{\top}\right) \right) \ba\right|,\\
							&&\sup_{\|\ba\|_2=1} \left|\ba^{\top}\left( \sum_{1\leq i\neq j\leq p } \E\left( \bW_{k,i,j}^{\top}\bW_{k,i,j}\right) \right) \ba\right|    \Bigg\}\\
							&=&\max\Bigg\{ \sup_{\|\ba\|_2=1} \left| \sum_{1\leq i\neq j\leq p } \sum_{s=1}^{k-1}\left(\var\left( Y_{i,j} \right)  \alpha^{2k}\left|\frac{\partial^k l_{i,j}\left( \btheta'\right) }{\partial \theta_{i}^{s}\partial \theta_{j}^{k-s}} \right|^2\left( a^{(k-s)}_j \right) ^2 \right)\right|,\\
							&&\sup_{\|\ba\|_2=1} \left| \sum_{1\leq i\neq j\leq p } \sum_{s=1}^{k-1}\left(\var\left( Y_{i,j} \right)  \alpha^{2k}\left|\frac{\partial^k l_{i,j}\left( \btheta'\right) }{\partial \theta_{i}^{s}\partial \theta_{j}^{k-s}} \right|^2\left( a^{(s)}_i \right) ^2\right)\right|   \Bigg\}\\
							&\leq& \max_{i,j,s}\left( \var\left(Y_{i,j} \right)\alpha^{2k}\left|\frac{\partial^k l_{i,j}\left( \btheta'\right) }{\partial \theta_{i}^{s}\partial \theta_{j}^{k-s}} \right|^2\right)\sup_{\|\ba\|_2=1} \left|\sum_{1\leq i\neq j\leq p } \sum_{s=1}^{k-1} \left( a^{(k-s)}_i \right) ^2 +\left( a^{(s)}_j \right) ^2 \right| \\
							&\leq &2p\left( \left( k-1\right)!\right)^2\sigma_{(p)}^2.
						\end{eqnarray*}
						Using the general Matrix Bernstein inequality (c.f. Theorem 6.17 and equation (6.43) of \cite{wainwright2019high}), we have,
						\begin{align*}
							P\left( \left\| \bW_{k} \right\|_{2} > \epsilon\right) &= P\left( \left\|\sum_{1\leq i<j\leq p } \bW_{k,i,j}\right\|_{2}>\epsilon \right)\\
							&\leq 2p \ \exp\left( -\frac{\epsilon^2}{\left(p-1 \right)  \left( \left( k-1\right)!\right)^2\sigma_{(p)}^2+ 2\left( k-1\right)!b_{(p)}\epsilon   } \right).
						\end{align*}
						Consequently, there exist  big enough constants $C_2>0$ and $c_ 2>0$  s.t. by choosing 
						\begin{align*}
							\epsilon=C_{2} \left( k-1\right)! \left( b_{(p)} \log(np)+ \sigma_{(p)}\sqrt{p\log(np)}\right),
						\end{align*}
						we have, with probability greater than $1-(np)^{c_2}$, 
						\begin{equation*}
							\| \bW_{k} \|_2 \leq C_2\left( k-1\right)! \left( b_{(p)} \log(np)+ \sigma_{(p)}\sqrt{p\log(np)}\right),
						\end{equation*}
						holds  for all  $k\leq K=c_{0}\log(np)$ where   $c_{0}>0$ is big enough constant.  
					\end{proof}
					\subsection{Proof of Theorem \ref{thm6}}
					\begin{proof}
						For any vector $\bx\in \mathbb{R}^{p}$, we define  $\bx_{-i}:=(x_{1},\cdots,x_{i-1}, x_{i+1}, \cdots, x_{p})^{\top}$.
						With similar arguments as in the proof of Theorem \ref{thm5}, we have that there exists  a $\btheta^{\xi}$ such that
						\begin{eqnarray*}
							&&\bL_{i}\left( \btheta\right)-\bL_{i}\left( \btheta'\right)\\
							& =&\frac{1}{p}\sum_{k=1}^{\infty}\left(\frac{1}{k!} \sum_{j=1,\:j \neq i}^{p} \left( Y_{i,j} \sum_{s=0}^{k} \tbinom{k}{s}\left( \theta_{i}-\theta_{i}'\right)^{s}\left( \theta_{j}-\theta_{j}' \right)^{k-s} \frac{\partial^k l_{i,j}\left( \btheta'\right) }{\partial \theta_{i}^{s}\partial \theta_{j}^{k-s}}   \right)   \right)\\
							&\leq&\left|\sum_{k=1}^{K} \frac{1}{p}\frac{1}{k!}\sum_{j=1,\:j \neq i}^{p} Y_{i,j}\left(  \left( \theta_{i}-\theta_{i}'\right)^{k}+ \left( \theta_{j}-\theta_{j}' \right)^{k} \right) \frac{\partial^k l_{i,j}\left( \btheta'\right) }{\partial \theta_{i}^{k} }    \right|\\
							&&+ \left|\sum_{k=2}^{K}\frac{1}{p}\frac{1}{k!} \sum_{j=1,\:j \neq i}^{p}\left( Y_{i,j} \sum_{s=1}^{k-1} \tbinom{k}{s}\left( \theta_{i}-\theta_{i}'\right)^{s}\left( \theta_{j}-\theta_{j}' \right)^{k-s} \frac{\partial^k l_{i,j}\left( \btheta'\right) }{\partial \theta_{i}^{s}\partial \theta_{j}^{k-s}}   \right)   \right|\\
							&&+\frac{1}{p}\left|\frac{1}{\left( K+1\right) !} \sum_{j=1,\:j \neq i}^{p}\left( Y_{i,j} \sum_{s=0}^{K+1} \tbinom{K+1}{s}\left( \theta_{i}-\theta_{i}'\right)^{s}\left( \theta_{j}-\theta_{j}' \right)^{K+1-s} \frac{\partial^{K+1} l_{i,j}\left( \btheta^{\xi}\right) }{\partial \theta_{i}^{s}\partial \theta_{j}^{K+1-s}}   \right)   \right|\\
							&=&S^{(1)}_{i} + S^{(2)}_{i} + S^{(3)}_{i},
						\end{eqnarray*}
						where $K=\lfloor c_0\log(np)\rfloor+1$ for some large enough constant $c_0$.
						First consider $S^{(1)}_{i}$. There exist big enough constants $C_1>0$, $c_1>0$ such that, uniformly for all $i=1,\cdots,p$, $k=1,2,\cdots,K$ and all $\btheta\in\bB_{\infty}\left(\btheta',\alpha_{0} \right)$, we have,  with probability greater than $1-(np)^{-c_1}$,
						\begin{eqnarray*}
							&&\frac{1}{p}\frac{1}{k!}\left| \sum_{j=1,\:j \neq i}^{p} Y_{i,j}\left(  \left( \theta_{i}-\theta_{i}'\right)^{k}+ \left( \theta_{j}-\theta_{j}' \right)^{k} \right) \frac{\partial^k l_{i,j}\left( \btheta'\right) }{\partial \theta_{i}^{k} }    \right|\\
							&=&\frac{1}{p}\frac{1}{k!}  \left| \frac{\theta_i}{\alpha}-\frac{\theta_i'} {\alpha}\right|^{k} \left| \sum_{j\neq i,j=1}^{p} Y_{i,j}\frac{\partial^k l_{i,j}\left( \btheta'\right) }{\partial \theta_{i}^{k} } \alpha^k\right|\\
							&&+\frac{1}{p}\frac{1}{k!} \sum_{j\neq i,j=1}^{p}\left| \frac{\theta_j}{\alpha}-\frac{\theta_j'} {\alpha}\right|^{k} \left|  Y_{i,j}\frac{\partial^k l_{i,j}\left( \btheta'\right) }{\partial \theta_{i}^{k} } \alpha^k\right|\\
							&\leq &\frac{1}{p}\frac{1}{k!}\left|\frac{ \theta_{i}-\theta'_{i} }{\alpha}\right|^{k}  \max_{i}\left|\sum_{j\neq i,j=1}^{p}Y_{i,j}\frac{\partial^k l_{i,j}\left( \btheta'\right) }{\partial \theta_{i}^{k} } \alpha^k\right|\\
							&&+ \frac{1}{p}\frac{1}{k!}\left\|\frac{\btheta_{-i}-\btheta'_{-i}} {\alpha}\right\|_{k}^{k}\max_{i,j}\left|Y_{i,j}\frac{\partial^k l_{i,j}\left( \btheta'\right) }{\partial \theta_{i}^{k} } \alpha^k\right|\\
							&\leq& C_{1 }\left( \left|\frac{ \theta_{i}-\theta'_{i} }{\alpha}\right|^{k}  \frac{b_{(p)} \log(np)+ \sigma_{(p)}\sqrt{p\log(np)}}{kp}+ \left\|\frac{\btheta_{-i}-\btheta'_{-i}} {\alpha}\right\|_{k}^{k} \frac{b_{(p)}}{kp}\right).
						\end{eqnarray*}
						Then, from inequality \eqref{thm5_1} we have, there exist big enough constants $C_{2}>0$ and $c_2>0$, such that  with probability greater than $1-(np)^{-c_2}$,
						\begin{eqnarray*}
							S^{(1)}_{i}&\leq&\sum_{k=1}^{K}\frac{1}{p}\frac{1}{k!}\left| \sum_{j=1,\:j \neq i}^{p} Y_{i,j}\left(  \left( \theta_{i}-\theta_{i}'\right)^{k}+ \left( \theta_{j}-\theta_{j}' \right)^{k} \right) \frac{\partial^k l_{i,j}\left( \btheta'\right) }{\partial \theta_{i}^{k} }    \right|\\
							&\leq&  C_{1}\sum_{k=1}^{K}\left(\left|\frac{ \theta_{i}-\theta'_{i} }{\alpha}\right|^{k}  \frac{b_{(p)} \log(np)+ \sigma_{(p)}\sqrt{p\log(np)}}{kp}+ \left\|\frac{\btheta_{-i}-\btheta'_{-i}} {\alpha}\right\|_{k}^{k} \frac{b_{(p)}}{kp}\right)   \\
							&\leq& C_{1}\frac{b_{(p)} \log(np)+ \sigma_{(p)}\sqrt{p\log(np)}}{p} \left( \sum_{k=1}^{K}\frac{1}{k}\left|\frac{ \theta_{i}-\theta'_{i} }{\alpha}\right|^{k}  \right)+ C_{1}\frac{b_{(p)}}{p}\sum_{k=1}^{K}\frac{1}{k}\left\| \frac{\btheta_{-i}-\btheta'_{-i}}{\alpha} \right\|_{k}^{k}  \\
							&\leq& C_{2}\frac{b_{(p)} \log(np)+ \sigma_{(p)}\sqrt{p\log(np)}}{p}\left|\theta_{i}-\theta'_{i}\right| + C_{2}\frac{b_{(p)}}{p} \left\|\btheta_{-i}-\btheta'_{-i}\right\|_{1},
						\end{eqnarray*}
						holds uniformly for $i=1,\cdots,p$ and $\btheta\in\bB_{\infty}\left(\btheta',\alpha_{0}  \right)$. 
						
						Next we derive an upper bound for $S^{(2)}_{i}$. With the fact that $Y_{i,j}^2\leq b_{(p)}^2$ and 
						\begin{equation*}
							\var\left( Y_{i,j}^2\right)\leq\E\left( Y_{i,j}^4\right)
							\leq b_{(p)}^2\E\left( Y_{i,j}^2\right)
							=  b_{(p)}^2\sigma^2_{(p)},
						\end{equation*}
						by Lemma \ref{Bers} we have, there exist big enough constants $C_3,c_3>0$ such that   with probability greater than $1-(np)^{-c_3}$, 
						\begin{align*}
							\max_{i}\left(\sum_{j=1,\:j \neq i}^{p} Y_{i,j}^2\right)^{1/2}&\leq C_{3} \left( b_{(p)}^2\log(np)+ \sigma_{(p)} b_{(p)}\sqrt{p\log(np)} \right) ^{1/2}\\
							&\leq C_{3} \left(\left( p\log(np)\right) ^{1/4}\sqrt{\sigma_{(p)} b_{(p)}}+ b_{(p)}\sqrt{\log(np)}\right) .
						\end{align*}
						Consequently,  there exists a big enough constant $C_4>0$ such that, uniformly for all $i=1,\cdots,p$ and $\btheta\in\bB_{\infty}\left(\btheta',\alpha_{0}  \right)$, we have, with probability greater than $1-(np)^{-c_3}$,
						\begin{eqnarray*}
							S^{(2)}_{i}
							&\leq&\sum_{k=2}^{K}\frac{1}{p}\frac{1}{k!}\sum_{j=1,\:j \neq i}^{p}\left| \left( Y_{i,j} \sum_{s=1}^{k-1} \tbinom{k}{s}\left( \theta_{i}-\theta_{i}'\right)^{s}\left( \theta_{j}-\theta_{j}' \right)^{k-s} \frac{\partial^k l_{i,j}\left( \btheta'\right) }{\partial \theta_{i}^{s}\partial \theta_{j}^{k-s}}   \right)   \right|\\
							&\leq&  \sum_{k=2}^{K}\frac{1}{p}\frac{1}{k!}0.5^{k}\left( \sum_{j=1,\:j \neq i}^{p}   Y_{i,j}^2\sum_{s=1}^{k-1}\left(\frac{\theta_{i}-\theta_{i}'}{\alpha/2}\right)^{2s}\right)^{1/2} \\
							&& \quad \left(\sum_{j=1,\:j \neq i}^{p}\sum_{s=1}^{k-1} \left( \tbinom{k}{s}\left( \frac{\theta_{j}-\theta_{j}'}{\alpha/2} \right)^{k-s}\alpha^{k} \frac{\partial^k l_{i,j}\left( \btheta'\right) }{\partial \theta_{i}^{s}\partial \theta_{j}^{k-s}}\right) ^2 \right) ^{1/2}\\
							&\leq& \sum_{k=2}^{K}\frac{1}{p}\frac{1}{k!}0.5^{k} \left( \frac{\left|\theta_{i}-\theta'_{i} \right|^2} {\alpha^2/4-\left|\theta_{i}-\theta'_{i}\right|^2 }\right) ^{1/2}\max_{i}\left(\sum_{j=1,\:j \neq i}^{p} Y_{i,j}^2\right)^{1/2} \\
							&&\quad  \left(\sum_{s=1}^{k-1}  \tbinom{k}{s}^2\sum_{j=1,\:j \neq i}^{p}\left( \frac{\theta_{j}-\theta_{j}'}{\alpha/2} \right)^{2(k-s)} \right) ^{1/2} \max_{j,s,k}\left|\alpha^{k} \frac{\partial^k l_{i,j}\left( \btheta'\right) }{\partial \theta_{i}^{s}\partial \theta_{j}^{k-s}} \right|\\
							&\leq&C_{4} \left|\theta_{i}-\theta'_{i} \right|\frac{\left( p\log(np)\right) ^{1/4}\sqrt{\sigma_{(p)} b_{(p)}}+ b_{(p)}\sqrt{\log(np)}}{p}\\
							&&  \quad   \sum_{k=2}^{K} \frac{0.5^{k} }{k} \left(\sum_{s=1}^{k-1}  \tbinom{k}{s}^2\sum_{j=1,\:j \neq i}^{p}\left( \frac{\theta_{j}-\theta_{j}'}{\alpha/2} \right)^{2(k-s)} \right) ^{1/2}\\
							&\leq&C_{4} \left|\theta_{i}-\theta'_{i} \right|\frac{\left( p\log(np)\right) ^{1/4}\sqrt{\sigma_{(p)} b_{(p)}}+ b_{(p)}\sqrt{\log(np)}}{p} \\
							&& \times\sum_{k=2}^{K} \sum_{s=1}^{k-1}\sum_{j=1,\:j \neq i}^{p}  \frac{0.5^{k} }{s}\tbinom{k}{s}\left|\frac{\theta_{j}-\theta_{j}'}{\alpha/2} \right|^{k-s}\\
							&\leq&C_{4} \left|\theta_{i}-\theta'_{i} \right|\frac{\left( p\log(np)\right) ^{1/4}\sqrt{\sigma_{(p)} b_{(p)}}+ b_{(p)}\sqrt{\log(np)}}{p} \\
							&& \times    \sum_{s=1}^{K-1} \left( \frac{1}{s}  \left\|\frac{\btheta_{-i}-\btheta_{-i}'}{\alpha/2} \right\|_{s}^{s} \sum_{k=s+1}^{K}\tbinom{k}{s}0.5^{k}\right) \\
							&\leq&2C_{4} \left|\theta_{i}-\theta'_{i} \right|\frac{\left( p\log(np)\right) ^{1/4}\sqrt{\sigma_{(p)} b_{(p)}}+ b_{(p)}\sqrt{\log(np)}}{p} \sum_{s=1}^{K-1}\frac{1}{s}  \left\|\frac{\btheta_{-i}-\btheta_{-i}'}{\alpha/2} \right\|_{s}^{s}\\
							&\leq&4C_{4} \left\|\btheta_{-i}-\btheta_{-i}' \right\|_{1}\left|\theta_{i}-\theta'_{i} \right|\frac{\left( p\log(np)\right) ^{1/4}\sqrt{\sigma_{(p)} b_{(p)}}+ b_{(p)}\sqrt{\log(np)}}{p}.
						\end{eqnarray*}
						Here in the above inequalities we have used that fact that $\sum_{k=s+1}^{K}\tbinom{k}{s}0.5^{k}<2$, and the last step follows from the inequality \eqref{thm5_1}.

						Finally, we derive an upper bound for $S^{(3)}_{i}$.	By condition (L-A1),  by choosing $K=\lfloor c_0\log(np)\rfloor+1$ with $c_0$ to be a large enough constant, there exists a big enough constant $c_4>0$ such that, uniformly for all $i=1,\cdots,p$,  $\btheta^{\xi}$ and  $\btheta\in\bB_{\infty}\left(\btheta',\alpha_{0}  \right)$, we have
						\begin{eqnarray*}
							S^{(3)}_{i}&=&\frac{1}{p}\left|\frac{1}{\left( K+1\right) !} \sum_{j=1,\:j \neq i}^{p}\left( Y_{i,j} \sum_{s=0}^{K+1} \tbinom{K+1}{s}\left( \theta_{i}-\theta_{i}'\right)^{s}\left( \theta_{j}-\theta_{j}' \right)^{K+1-s} \frac{\partial^{K+1} l_{i,j}\left( \btheta^{\xi}\right) }{\partial \theta_{i}^{s}\partial \theta_{j}^{K+1-s}}   \right)   \right|\\
							&\leq &\frac{1}{p}\frac{1}{K+1}  \sum_{j=1,\:j \neq i}^{p} \left| Y_{i,j} \right|\left(\frac{\left| \theta_{i}-\theta_{i}' \right|}{\alpha} + \frac{\left| \theta_{j}-\theta_{j}' \right|}{\alpha}\right)^{K+1} \\
							&\leq & \frac{b_{(p)}}{K+1}\left(\frac{2\left\| \btheta-\btheta' \right\|_{\infty}}{\alpha} \right)^{K+1}\\
							&\leq & \frac{b_{(p)}}{K+1}\left(\frac{2\alpha_{0}}{\alpha} \right)^{K+1}\\
							&\leq & b_{(p)}\left(np \right)^{-c_4}.
						\end{eqnarray*}
						Here the last step will hold if we choose $c_0$ to be large enough such that $(2\alpha_{0}/\alpha)^{c_0/ c_3}$$<e^{-1}$. When $np\rightarrow \infty$, and $c_0, c_3$ are chosen to be large enough,  this bound will be dominated by the upper bounds for $S^{(1)}$ and $S^{(2)}$.
						
						Consequently,  by choosing $K=\lfloor c_0\log(np)\rfloor+1$ with $c_0$ to be a large enough constant, we conclude that, for any given $\btheta'$, there exist large enough constants $C,c>0$ such that uniform for any $\btheta\in \bB_{\infty}\left( \btheta',\alpha_{0}\right)$    and $1\leq i\leq p$, with probability greater than $1-(np)^{-c}$,
						\begin{eqnarray*}
							&&\left|\bL_{i}\left( \btheta\right)-\bL_{i}\left( \btheta'\right)\right|\\
							&\leq& C_{2}\frac{b_{(p)} \log(np)+ \sigma_{(p)}\sqrt{p\log(np)}}{p}\left| \theta_{i}-\theta'_{i}\right| + C_{2}\frac{b_{(p)}}{p} \left\|\btheta_{-i}-\btheta'_{-i}\right\|_{1}\\
							&&+ C_{4}\left\|\btheta_{-i}-\btheta_{-i}' \right\|_{1}\left|\theta_{i} -\theta'_{i} \right|\frac{\left( p\log(np)\right) ^{1/4}\sqrt{\sigma_{(p)} b_{(p)}}+ b_{(p)}\sqrt{\log(np)}}{p} \\
							&\leq& C\frac{b_{(p)}}{p} \left\|\btheta_{-i}-\btheta_{-i}' \right\|_{1}+C\left( \left\|\btheta_{-i}-\btheta_{-i}' \right\|_{1}+1\right) \left|\theta_{i} -\theta'_{i} \right|\frac{b_{(p)} \log(np)+ \sigma_{(p)}\sqrt{p\log(np)}}{p}.
						\end{eqnarray*}
					\end{proof}

						\section{Additional numerical results}
						~\\
						\subsection{Informal justification of the use of TWHM in analyzing the insecta-ant-colony4 dataset}\label{AR1}
						To motivate the use of the proposed TWHM for the analysis of the insecta-ant-colony4 dataset, 
						we plot the autocorrelation function (ACF) and the partial autocorrelation function (PACF) of the degree sequences of two ants. These ants are selected based on their respective highest and lowest values at time $t=1$ according to  
						\begin{equation*}
							\sum_{t=1}^{n}\left| \frac{p-1}{2}- \sum_{j=1,\:j \neq i}^{p}X_{i,j}^{t}\right|.
						\end{equation*}
						Upon examining Figure \ref{fig:acf}, it becomes evident that both degree sequences exhibit patterns reminiscent of a first-order autoregressive model with long memory. This observation serves as a strong motivation for employing the TWHM methodology.
						\begin{figure}[tb]
							\includegraphics[width=0.9\textwidth]{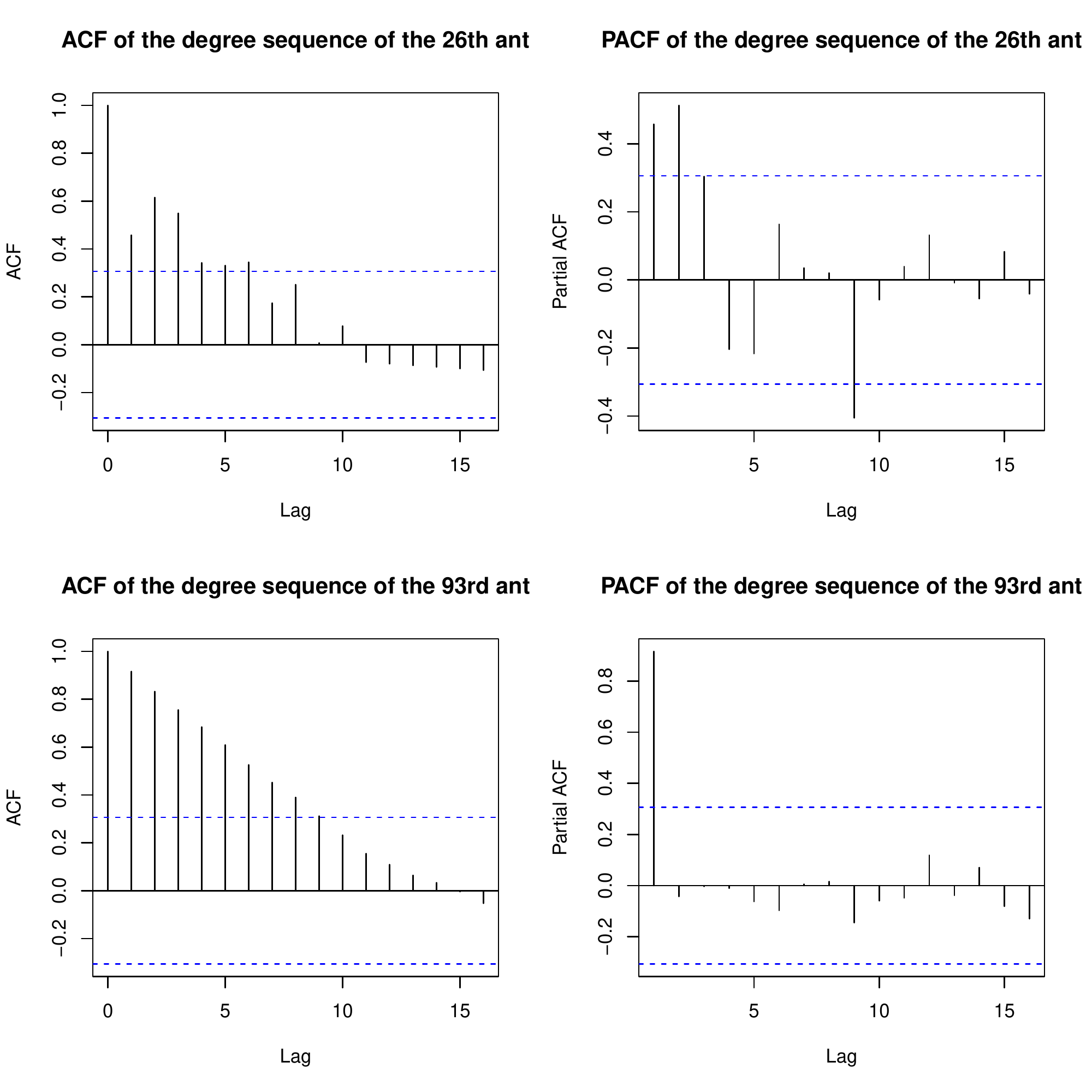}
							\caption{The ACF and PACF plots of the degree sequences of two selected ants. }
							\label{fig:acf}
						\end{figure}			
						\subsection{Community detection under stochastic block structures}
						We conduct additional numerical studies to assess the efficacy of community detection using the estimated $\beta$-parameters. To implement stochastic block structures within the TWHM framework, we partitioned the $p$ nodes into $k$ communities of equal size, ensuring that the parameters $(\beta_{i,0},\beta_{i,1})$ for all nodes $i$ within the same community were identical. Furthermore, we explored scenarios where the networks were independently generated from a Stochastic Block Model (SBM). 
						Specifically, we have considered the following settings:
						\begin{itemize}   
							\item[] Setting 1: The networks are generated under TWHM with  $\bbeta_{i,r}= -0.2, 0, 0.2$ $(r=0,1)$ for all nodes $i$ in communities $1,2$ and $3$ respectively.
							\item[] Setting 2: The networks are generated under TWHM with  $\bbeta_{i,r}= -0.4, 0, 0.4$ $(r=0,1)$ for all nodes $i$ in communities $1,2$ and $3$ respectively.
							\item[] Setting 3: The networks are independently generated under SBM, with the probability matrix among different communities specified as
							\begin{small}
								\begin{align*}
									\left[
									\begin{array}{ccc}
										0.26&	0.1&	0.1\\
										0.1&	0.2&	0.1\\
										0.1&	0.1&	0.14\\ 			
									\end{array}
									\right].
								\end{align*}
							\end{small}	    	
							\item[] Setting 4: The networks are independently generated under  SBM, with the probability matrix among different communities given by	
							\begin{small}
								\begin{align*}
									\left[
									\begin{array}{ccc}
										0.4&	0.3&	0.2\\
										0.3&	0.225&	0.15\\
										0.2&	0.15&	0.1\\ 			
									\end{array}
									\right].
								\end{align*}
							\end{small}
							\item[] Setting 5: The networks are independently generated under  SBM, with  the probability matrix among different communities given by
							\begin{small}
								\begin{align*}
									\left[
									\begin{array}{ccc}
										0.9&	0.5&	0.3\\
										0.5&	0.3&	0.2\\
										0.3&	0.2&	0.15\\ 			
									\end{array}
									\right].
								\end{align*}
							\end{small}
						\end{itemize}  
						Networks in Settings 1-2 are generated with autoregressive dependence using our TWHM model, while networks in Setting 3 are independent samples following the classical SBM structure. In Settings 4-5, networks are also generated with SBM structures, but the probability formation matrices are not full-rank.
						\\
						Once the $\beta$-parameters are estimated, we apply k-means clustering to these parameters to cluster the $p$ nodes, denoting this method as "TWHM-Cluster". For comparison, we apply spectral clustering, widely used for SBM, on the averaged networks, denoting this method as "SBM-Spectral". All experiments are repeated 100 times, and the clustering accuracy is reported in Table \ref{error_SBM} below.
						\\
						From Table \ref{error_SBM}, we observe that community detection using the $\beta$-parameters performs significantly better under Settings 1 and 2, where data were generated from our TWHM model. This improvement is attributed to the fact that parameter estimation has considered the autoregressive structure of the networks.
						\\
						When networks were independently generated from SBM under Setting 3, the performance of TWHM-Cluster is comparable to that of SBM-Cluster. However, when the probability matrix of SBM is not full-rank (Settings 4 and 5), the TWHM model still demonstrates promising performance, while classical spectral clustering can be much less satisfactory.
						\begin{table}[htbp]
							\centering
							\caption{Mean clustering accuracy of TWHM and SBM over 100 replications. Here $k, n, p$ denote the number of communities, the number of network observations, and the number of nodes, respectively.}
							\begin{tabular}{c|c|c|c}
								\hline
								& $(k,n,p)$ & TWHM-Cluster & SBM-Spectral \\
								\hline
								\multirow{3}{*}{Setting 1} & (3,2,300) & 68.6\% & 37.1\% \\
								& (3,10,300) & 95.1\% & 37.0\% \\
								& (3,50,300) & 99.5\% & 38.7\% \\
								\hline
								\multirow{3}{*}{Setting 2} & (3,2,300) & 92.2\% & 39.7\% \\
								& (3,10,300) & 95.6\% & 48.6\% \\
								& (3,50,300) & 100.0\% & 63.0\% \\
								\hline
								\multirow{3}{*}{Setting 3} & (3,10,300) & 92.1\% & 99.8\% \\
								& (3,30,300) & 99.4\% & 100.0\% \\
								& (3,50,300) & 99.3\% & 100.0\% \\
								\hline
								\multirow{3}{*}{Setting 4} & (3,10,500) & 97.0\% & 37.1\% \\
								& (3,30,500) & 93.9\% & 37.0\% \\
								& (3,50,500) & 100\% & 37.5\% \\
								\hline
								\multirow{3}{*}{Setting 5} & (3,10,500) & 80.0\% & 71.2\% \\
								& (3,30,500) & 80.2\% & 70.7\% \\
								& (3,50,500) & 83.8\% & 72.8\% \\
								\hline
							\end{tabular}%
							\label{error_SBM}%
						\end{table}%
						\subsection{Dynamic protein-protein	interaction networks}
						In this section, we applied the proposed TWHM to 12  dynamic protein-protein	interaction networks (PPIN)  of yeast cells  examined in \cite{fu2022dppin}. Each dynamic network comprises 36 network observations. The objective of investigating protein-protein interactions is to glean valuable insights into the cellular function and machinery of a proteome \citep{wu2006prediction}. To provide an overview of these datasets, we present selected statistics in Table \ref{PPIN_stat}.
						\begin{table}[htbp]
							\centering
							\caption{Some statistics of  the 12 protein-protein	interaction network datasets.}
							\begin{tabular}{c|c|c|c}
								\hline
								Dataset & \# of Nodes & Mean degree & Density \\
								\hline
								DPPIN-Uetz  & 922   & 4.68  & 0.22\% \\
								\hline
								DPPIN-Ito  & 2856  & 6.05  & 0.07\% \\
								\hline
								DPPIN-Ho  & 1548  & 54.55 & 0.13\% \\
								\hline
								DPPIN-Gavin  & 2541  & 110.22 & 0.08\% \\
								\hline
								DPPIN-Krogan (LCMS)  & 2211  & 77.01 & 0.09\% \\
								\hline
								DPPIN-Krogan (MALDI)  & 2099  & 74.60 & 0.10\% \\
								\hline
								DPPIN-Yu  & 1163  & 6.19  & 0.17\% \\
								\hline
								DPPIN-Breitkreutz  & 869   & 90.33 & 0.23\% \\
								\hline
								DPPIN-Babu & 5003  & 44.56 & 0.04\% \\
								\hline
								DPPIN-Lambert  & 697   & 19.09 & 0.29\% \\
								\hline
								DPPIN-Tarassov  & 1053  & 9.17  & 0.19\% \\
								\hline
								DPPIN-Hazbun  & 143   & 27.40 & 1.40\% \\
								\hline
							\end{tabular}%
							\label{PPIN_stat}%
						\end{table}%
						We have employed our method for linkage prediction on these PPINs. Similar to the main paper, we utilized TWHM with either a fixed cutoff point of $0.5$ (TWHM$_{0.5}$) or the adaptive rule (TWHM$_{adaptive}$). For comparison, we used a naive estimator $\bX^{t-1}$ to predict $\bX^t$ (Naive). Our training time slot size was set to $n = 10$, and the results are presented in Table \ref{ppin_pre} below. As evident from the table, our approach shows significant promise for accurate link prediction.
						\begin{table}[htbp]
							\small
							\centering
							\caption{Comparison of TWHM$_{0.5}$, TWHM$_{adaptive}$, and Naive in terms of misclassification rate of links ($\times 10^{-5}$).}
						\begin{tabular}{c|Hc|c|c}
							\hline
							Dataset & Statistics & TWHM$_{0.5}$  & TWHM$_{adaptive}$   & Naive \\
							\hline
							{DPPIN-Uetz} & MR & 55 & 55 & 88 \\
							\hline
							{DPPIN-Ito} & MR & 28 & 29 & 50 \\
							\hline
							{DPPIN-Ho} & MR & 108 & 111 & 194 \\
							\hline
							{DPPIN-Gavin} & MR & 139 & 142 & 250 \\
							\hline
							{DPPIN-Krogan(LCMS)} & MR &133 & 141& 222 \\
							\hline
							{DPPIN-Krogan(MALDI)} & MR & 124 & 127 & 215 \\
							\hline
							{DPPIN-Yu} & MR & 54 & 55 & 86 \\
							\hline
							{DPPIN-Breitkreutz} & MR & 297 & 305 & 493 \\
							\hline
							{DPPIN-Babu} & MR & 27 & 28 & 45 \\
							\hline
							{DPPIN-Lambert} & MR & 293 & 298 & 489 \\
							\hline
							{DPPIN-Tarassov} & MR & 72 & 75 & 133 \\
							\hline
							{DPPIN-Hazbun} & MR & 759 & 772 & 1198 \\
							\hline
						\end{tabular}%
						\label{ppin_pre}%
					\end{table}%
			\end{appendix}

			\end{document}